\documentclass[11pt]{article}

\usepackage[letterpaper,margin=1in]{geometry}

\usepackage{amsmath,amssymb,amsthm,mathtools}
\usepackage{graphicx}
\usepackage{xcolor}
\usepackage{microtype}
\usepackage{comment}

\usepackage[colorlinks=true,linkcolor=blue,citecolor=blue,urlcolor=blue]{hyperref}
\usepackage[nameinlink,capitalise,noabbrev]{cleveref}

\usepackage{enumitem}
\usepackage{booktabs}
\usepackage{algorithm}
\usepackage{algorithmic}

\newtheorem{theorem}{Theorem}[section]
\newtheorem{lemma}[theorem]{Lemma}
\newtheorem{proposition}[theorem]{Proposition}
\newtheorem{corollary}[theorem]{Corollary}

\theoremstyle{definition}

\theoremstyle{remark}
\newtheorem{remark}[theorem]{Remark}

\newcommand{\F}{\mathbb{F}}

\newcommand{\BW}{\mathrm{BW}}
\newcommand{\IO}{\mathrm{IO}}

\newcommand{\PGL}{\mathrm{PGL}}

\newcommand{\opt}{\mathrm{opt}}
\newcommand{\rank}{\mathrm{rank}}
\newcommand{\nz}{\mathrm{nz}}

\newcommand{\col}{\mathrm{col}}
\newcommand{\Col}{\mathrm{Col}}
\newcommand{\Span}{\mathrm{span}}
\newcommand{\GL}{\mathrm{GL}}

\newcommand{\id}{\mathrm{id}}
\newcommand{\Tr}{\operatorname{Tr}}

\makeatletter
\newcommand*\l@appsubsection{\@dottedtocline{2}{1.5em}{1.8em}}
\newcommand*\l@appsubsubsection{\@dottedtocline{3}{3.3em}{2.6em}}
\def\toclevel@appsubsection{2}
\def\toclevel@appsubsubsection{3}
\makeatother

\begin{document}

\begin{center}
  {\LARGE\bfseries
  Optimal Repair Bandwidth and Repair I/O \\[0.3em]
  of $(n,n-2,2)$ MDS Array Codes\par}
\end{center}

\vspace{0.6em}
\begin{center}
  {\large
    Huawei Wu\textsuperscript{1,*}
  }\\[0.55em]
  {\normalsize
    \textsuperscript{1}Shanghai Fintelli Box Technology Co., Ltd., Shanghai, 200127, China\\[0.35em]
    \textsuperscript{*}Corresponding author.
    \href{mailto:wuhuawei1996@gmail.com}{\texttt{wuhuawei1996@gmail.com}}
  }
\end{center}

\begin{abstract}
We give a complete determination of the exact optimal worst-case repair
bandwidth and repair I/O for linear exact repair of \((n,n-2,2)\) MDS array
codes over every finite field \(\F_q\) and for every admissible code length
\(3\le n\le q^2+1\). For repair
bandwidth, we prove that the optimum is governed, up to a short explicit list
of small exceptional cases, by the maximum of the sharpened \(n\)-only lower
bound \(\lceil(5n-8)/4\rceil\) and the projective counting, equivalently
incidence-multiplicity, bound \(2n-q-3\). For repair I/O, we obtain the
analogous exact formula with \(\lceil(4n-6)/3\rceil\) in place of
\(\lceil(5n-8)/4\rceil\), with the single special value at \(n=4\). Thus, we
completely resolve the first non-trivial redundancy and sub-packetization
regime \((r,\ell)=(2,2)\) for both repair bandwidth and repair I/O.
  \end{abstract}

%
\newpage
\tableofcontents
\newpage

\section{Introduction}

\emph{Maximum-distance-separable (MDS) array codes} are among the most important code families for distributed storage systems, since they provide the optimal trade-off between storage overhead and erasure tolerance. In such systems, node failures occur frequently, making single-node repair a fundamental operation. Accordingly, designing MDS array codes with low repair cost has become a central objective. The classical metric for repair cost is the \emph{repair bandwidth}, namely, the total amount of data downloaded during repair. 

The \emph{cut-set bound} established in the regenerating-code framework in \cite{dimakis2010network} has served as a fundamental benchmark for the study of repair bandwidth. At the \emph{minimum-storage point} of the resulting trade-off between node storage and repair bandwidth, one obtains \emph{minimum-storage regenerating (MSR) codes}, which remain MDS while achieving the cut-set bound. However, known lower bounds indicate that, under exact repair, achieving this optimum generally requires large sub-packetization, often exponential in the high-rate regime; see, for example, \cite{balaji2018tight,alrabiah2019exponential,balaji2022lower}. This motivates the study of MDS array codes with fixed and small sub-packetization, where a central goal is to characterize the trade-off among repair bandwidth, sub-packetization, and field size, and to develop explicit constructions with optimal or near-optimal repair bandwidth~\cite[Open Problem~9]{ramkumar2022codes}.

Besides repair bandwidth, another repair metric that has attracted growing attention is the \emph{repair I/O}, namely, the total amount of data accessed at the helper nodes during repair; see, for instance, \cite{dau2018repair,li2019costs,liu2024formula,liu2025calculating}. While repair bandwidth measures the communication cost, repair I/O captures the helper-side access cost, and the two quantities can behave quite differently; see, for example, \cite{zhang2025optimal}.

In this paper, as in \cite{zhang2025optimal}, we study \emph{linear exact repair} for $(n,n-2,2)$ MDS array codes, corresponding to the simplest non-trivial case $(r,\ell)=(2,2)$. Somewhat surprisingly, in this regime we can determine the exact optimum of the worst-case repair bandwidth and repair I/O for every admissible code length $n$.

Before presenting our main results, we briefly introduce the necessary terminology and notation. Throughout, we work directly with formal mathematical definitions and omit their interpretations in the distributed storage setting. Readers interested in the latter are referred to \cite{zhang2025optimal}.

Let $q$ be a prime power, and let $n,k,\ell$ be positive integers with $k<n$. Set $r:=n-k$ and write $[n]:=\{1,2,\dots,n\}$. An $(n,k,\ell)$ MDS array code over $\F_q$ is an $\F_q$-linear subspace
\[
\mathcal{C} \le (\F_q^\ell)^n
\]
of dimension $k\ell$. A codeword of $\mathcal{C}$ is written as
\[
\mathbf{c}=(\mathbf{c}_1,\mathbf{c}_2,\dots,\mathbf{c}_n),
\qquad \mathbf{c}_i\in\F_q^\ell,
\]
where $\mathbf{c}_i$ is the block stored in node $i$.

Equivalently, $\mathcal{C}$ can be specified by a block parity-check matrix
\[
H=[H_1\,H_2\,\cdots\,H_n]\in \F_q^{r\ell\times n\ell},
\qquad H_i\in \F_q^{r\ell\times \ell},
\]
of full row rank $r\ell$, via
\[
\mathcal{C}=\ker(H)=\bigl\{\mathbf{c}\in(\F_q^\ell)^n:\; H_1\mathbf{c}_1+\cdots+H_n\mathbf{c}_n=0\bigr\}.
\]
The code $\mathcal{C}$ is called \emph{MDS} if, for every subset $I\subseteq [n]$ with $|I|=r$, the square block matrix
\[
H_I:=[\,H_i\,]_{i\in I}\in \F_q^{r\ell\times r\ell}
\]
is invertible.

Now fix $i\in[n]$, and suppose that node $i$ fails, so that the block $\mathbf{c}_i$ is unavailable. A \emph{linear exact repair scheme} for node $i$ is specified by a matrix
\[
M\in \F_q^{\ell\times r\ell}
\]
such that $MH_i$ is invertible. Indeed, applying $M$ to the parity-check equation
\[
\sum_{j=1}^n H_j\mathbf{c}_j=0
\]
gives
\[
\sum_{j=1}^n MH_j\mathbf{c}_j=0.
\]
Since $MH_i$ is invertible, the missing block $\mathbf{c}_i$ is uniquely determined by
\[
\mathbf{c}_i
=
-(MH_i)^{-1}\sum_{j\ne i} MH_j\mathbf{c}_j.
\]
Accordingly, whenever $MH_i$ is invertible, we call $M$ a \emph{repair matrix} for node $i$.

For a repair matrix $M$ for node $i$, define its \emph{repair bandwidth} and \emph{repair I/O} by
\[
\BW_i(M):=\sum_{j\ne i}\rank(MH_j),
\qquad
\IO_i(M):=\sum_{j\ne i}\nz(MH_j),
\]
respectively, where, for a matrix $A$, $\nz(A)$ denotes the number of non-zero columns of $A$. The quantity $\rank(MH_j)$ represents the number of $\F_q$-symbols transmitted from helper node $j$, whereas $\nz(MH_j)$ represents the number of subsymbols of $\mathbf{c}_j$ that must be accessed to generate this transmission. It is clear that
\[
\IO_i(M)\ge \BW_i(M)
\]
for every repair matrix $M$.

Let
\[
\mathcal{M}_i(\mathcal{C})
:=
\{\,M\in\F_q^{\ell\times r\ell}: MH_i \text{ is invertible}\,\}
\]
denote the set of all repair matrices for node $i$. The optimal repair bandwidth and optimal repair I/O for repairing node $i$ are then defined by
\[
\beta_i(\mathcal{C})
:=
\min_{M\in \mathcal{M}_i(\mathcal{C})}\BW_i(M),
\qquad
\gamma_i(\mathcal{C})
:=
\min_{M\in \mathcal{M}_i(\mathcal{C})}\IO_i(M),
\]
respectively. Accordingly, the worst-case repair bandwidth and worst-case repair I/O of $\mathcal{C}$ are defined by
\[
\beta(\mathcal{C})
:=
\max_{i\in[n]}\beta_i(\mathcal{C}),
\qquad
\gamma(\mathcal{C})
:=
\max_{i\in[n]}\gamma_i(\mathcal{C}),
\]
respectively. Finally, let $\mathfrak{C}_{q,n,r,\ell}$ denote the set of all $(n,n-r,\ell)$ MDS array codes over $\F_q$. We define
\[
\beta_{\opt}^{(q,n,r,\ell)}
:=
\min_{\mathcal{C}\in \mathfrak{C}_{q,n,r,\ell}} \beta(\mathcal{C}),
\qquad
\gamma_{\opt}^{(q,n,r,\ell)}
:=
\min_{\mathcal{C}\in \mathfrak{C}_{q,n,r,\ell}} \gamma(\mathcal{C}).
\]
That is, $\beta_{\opt}^{(q,n,r,\ell)}$ and $\gamma_{\opt}^{(q,n,r,\ell)}$ denote, respectively, the minimum achievable worst-case repair bandwidth and the minimum achievable worst-case repair I/O over all $(n,n-r,\ell)$ MDS array codes over $\F_q$.

For general parameters $(q,n,r,\ell)$, determining the exact values of $\beta_{\opt}^{(q,n,r,\ell)}$ and $\gamma_{\opt}^{(q,n,r,\ell)}$ seems a difficult problem. However, in the case $r=\ell=2$, our main results completely resolve this problem.

\begin{theorem}[Main Result I]\label{thm:main-1}
  Let \(q\) be a prime power and let \(n\) be an integer satisfying
  \(3\le n\le q^2+1\). Then
  \[
    \beta_{\opt}^{(q,n,2,2)}
    =
    \max\left\{
      \left\lceil \frac{5n-8}{4}\right\rceil,\,
      2n-q-3
    \right\},
  \]
  except in the following cases, where the value of
  \(\beta_{\opt}^{(q,n,2,2)}\) is as specified below:
  \begin{enumerate}[label=(\alph*)]
    \item If \(n=5\) and either \(q\in\{3,4\}\) or \(q\ge 7\), then $\beta_{\opt}^{(q,5,2,2)}=4$.
    \item If \(n=6\) and either \(q=4\) or \(q\ge 7\), then $\beta_{\opt}^{(q,6,2,2)}=5$.
    \item If \(n=9\) and either \(q\) is odd with \(q\ge 7\), or \(q\) is even with \(q\ge 16\), then $\beta_{\opt}^{(q,9,2,2)}=9$. 
    \item If \(n=10\) and either \(q\) is odd with \(q\ge 7\) and \(q\neq 9\), or \(q\) is even with \(q\ge 16\), then $\beta_{\opt}^{(q,10,2,2)}=10$. 
  \end{enumerate}
  \end{theorem}

  \begin{theorem}[Main Result II]\label{thm:main-2}
    Let $q$ be a prime power and let $n$ be an integer with $3\le n\le q^2+1$. Then
    \[
    \gamma_{\opt}^{(q,n,2,2)}
    =
    \begin{cases}
      \max\left\{
    \left\lceil \frac{4n-6}{3}\right\rceil,\,
    2n-q-3\right\}&n\neq 4,\\
    3&n=4.
    \end{cases}
    \]
    \end{theorem}

  By \cite[Lemma~3.3]{liu2026linear}, every $(n,n-2,2)$ MDS array code over
$\F_q$ satisfies $n\le q^2+1$. Hence the range $3\le n\le q^2+1$ in
Theorem~\ref{thm:main-1} and Theorem~\ref{thm:main-2} already cover all possible code lengths. Therefore, these two theorems give a complete determination of
$\beta_{\opt}^{(q,n,2,2)}$ and $\gamma_{\opt}^{(q,n,2,2)}$.

Moreover, the values appearing in Theorem~\ref{thm:main-1} and Theorem~\ref{thm:main-2} are not ad hoc. Rather, they are dictated by the two lower-bound mechanisms that are currently known to be effective in the case $r=\ell=2$. More precisely, the first terms
\[
\left\lceil \frac{5n-8}{4}\right\rceil
\qquad\text{and}\qquad
\left\lceil \frac{4n-6}{3}\right\rceil
\]
originate from the lower bounds established in \cite{zhang2025optimal} for $(n,n-2,2)$ MDS array codes. The bounds proved there differ slightly from the above expressions; in Section~\ref{sec:lb} we sharpen them to the form stated in Theorem~\ref{thm:main-1} and Theorem~\ref{thm:main-2}. The second term
\[
2n-q-3
\]
comes from the projective counting bound proved in \cite{liu2026linear}, and, more generally, from the incidence-multiplicity bound established in \cite{wu2026incidence}. In the special case $r=\ell=2$, both bounds reduce to the same quantity $2n-q-3$. Therefore, Theorem~\ref{thm:main-1} and Theorem~\ref{thm:main-2} show that, in the regime $r=\ell=2$, the two known lower bounds already capture the exact optimum of the worst-case repair bandwidth and repair I/O.

The remainder of the paper is organized as follows. In Section~\ref{sec:prelim}, we recall the
intrinsic subspace reformulation of linear exact repair proposed in \cite{liu2026linear}. In Section~\ref{sec:lb}, we revisit the lower bounds
from~\cite{zhang2025optimal} and sharpen them to
\[
\beta(\mathcal C)\ge \left\lceil \frac{5n-8}{4}\right\rceil,
\qquad
\gamma(\mathcal C)\ge \left\lceil \frac{4n-6}{3}\right\rceil.
\]
In Section~\ref{sec:short}, we construct $(n,n-2,2)$ MDS array codes attaining
these bounds in the short-length regime. In Section~\ref{sec:long}, we turn to the $q$-dependent bound $2n-q-3$
and construct codes attaining it in the long-length regime. Finally, in Section~\ref{sec:mainproof}, we assemble the preceding results
and complete the proofs of Theorems~\ref{thm:main-1} and~\ref{thm:main-2},
using the small-parameter results proved in the appendix.

\section{The Subspace Reformulation of Linear Exact Repair}\label{sec:prelim}

Following the intrinsic viewpoint developed in
\cite{liu2026linear,wu2026incidence}, we reformulate linear exact repair in
terms of linear subspaces. We only record the definitions and notation that will be needed later.

Let
\[
\mathcal C=\ker(H)\le (\F_q^\ell)^n,
\qquad
H=[H_1\,H_2\,\cdots\,H_n]\in \F_q^{r\ell\times n\ell},
\qquad
H_i\in \F_q^{r\ell\times \ell},
\]
be an $(n,n-r,\ell)$ MDS array code over $\F_q$. Set
\[
\mathbb{V}:=\F_q^{r\ell}.
\]
For each $i\in[n]$, define the \emph{node subspace}
\[
\mathcal H_i:=\Col(H_i)\le \mathbb{V},
\]
where $\Col(H_i)$ denotes the column space of $H_i$.
Since $\mathcal C$ is MDS, each $\mathcal H_i$ has dimension $\ell$, and for
every subset $I\subseteq[n]$ with $|I|=r$,
\[
\mathbb{V}=\bigoplus_{i\in I}\mathcal H_i.
\]
Conversely, any family of $\ell$-dimensional subspaces
\[
\mathcal H_1,\dots,\mathcal H_n\le \mathbb V
\]
satisfying
\[
\mathbb V=\bigoplus_{i\in I}\mathcal H_i
\qquad\text{for every }I\subseteq [n]\text{ with }|I|=r
\]
determines an $(n,n-r,\ell)$ MDS array code over $\F_q$.

Now fix $i\in[n]$ and let
\[
M\in \mathcal M_i(\mathcal C)
=
\{\,M\in\F_q^{\ell\times r\ell}: MH_i \text{ is invertible}\,\}
\]
be a repair matrix for node $i$. Write
\[
W:=\ker(M)\le \mathbb{V}.
\]
Then $W$ has dimension $(r-1)\ell$, and the condition that $MH_i$ be invertible
is equivalent to
\[
W\cap \mathcal H_i=\{0\}.
\]
This leads to the family of \emph{feasible repair subspaces}
\[
\mathcal W_i:=
\{\,W\le \mathbb{V}:\ \dim W=(r-1)\ell,\ W\cap \mathcal H_i=\{0\}\,\}.
\]
Conversely, every $W\in\mathcal W_i$ is the kernel of some repair matrix for
node $i$.

The repair bandwidth admits a simple intrinsic description in terms of
intersection dimensions. If $M\in\mathcal M_i(\mathcal C)$ and $W=\ker(M)$,
then for every $j\ne i$,
\[
\rank(MH_j)=\ell-\dim(W\cap \mathcal H_j).
\]
Hence
\[
\BW_i(M)
=
\ell(n-1)-\sum_{j\ne i}\dim(W\cap \mathcal H_j).
\]
Accordingly, define
\[
\alpha_i(\mathcal C)
:=
\max_{W\in\mathcal W_i}
\sum_{j\ne i}\dim(W\cap \mathcal H_j).
\]
Then
\[
\beta_i(\mathcal C)=\ell(n-1)-\alpha_i(\mathcal C).
\]

To encode repair I/O, one must retain not only the node subspaces
$\mathcal H_i$, but also the column directions inside each block $H_i$. Write
\[
  H_i=[\,h_{i,1}\ \cdots\ h_{i,\ell}\,],\qquad h_{i,t}\in \mathbb{V},
\]
and define the associated \emph{projective column set}
\[
\mathcal X_i
:=
\{\,\langle h_{i,1}\rangle,\dots,\langle h_{i,\ell}\rangle\,\}
\subseteq \mathbb{P}(\mathcal H_i),
\]
where $\mathbb{P}(U)$ denotes the set of $1$-dimensional subspaces of a non-zero vector space $U$.

Let $M\in\mathcal M_i(\mathcal C)$ and $W=\ker(M)\in\mathcal W_i$. For each
$j\ne i$, define
\[
z_j(W):=\bigl|\mathcal X_j\cap \mathbb{P}(W)\bigr|.
\]
Equivalently, $z_j(W)$ is the number of columns of $H_j$ that lie in $W$. Since
the $t$-th column of $MH_j$ is zero if and only if $h_{j,t}\in W$, we have
\[
\nz(MH_j)=\ell-z_j(W).
\]
Therefore
\[
\IO_i(M)
=
\ell(n-1)-\sum_{j\ne i} z_j(W).
\]
Accordingly, define
\[
\lambda_i(\mathcal C)
:=
\max_{W\in\mathcal W_i}
\sum_{j\ne i}\bigl|\mathcal X_j\cap \mathbb{P}(W)\bigr|.
\]
Then
\[
\gamma_i(\mathcal C)=\ell(n-1)-\lambda_i(\mathcal C).
\]

In summary, linear exact repair of node $i$ is encoded by a feasible repair
subspace $W\in\mathcal W_i$. Under this reformulation, the quantities $\beta_i(\mathcal C)$ and
$\gamma_i(\mathcal C)$ are determined by $\alpha_i(\mathcal C)$ and
$\lambda_i(\mathcal C)$, respectively, where $\alpha_i(\mathcal C)$ and
$\lambda_i(\mathcal C)$ are obtained by maximizing
\[
\sum_{j\ne i}\dim(W\cap \mathcal H_j)
\qquad\text{and}\qquad
\sum_{j\ne i}\bigl|\mathcal X_j\cap \mathbb{P}(W)\bigr|,
\]
respectively, over all $W\in\mathcal W_i$.

This reformulation is particularly convenient for the study of the projective
counting bound, and it will be used repeatedly in
Section~\ref{sec:long}.

\section{%
  \texorpdfstring{A Sharpening of the $n$-Only Lower Bounds}%
                {A Sharpening of the n-Only Lower Bounds}%
}\label{sec:lb}

We begin by recalling the $n$-only lower bounds established in \cite{zhang2025optimal} for $(n,n-2,2)$ MDS array codes. 

\begin{theorem}[{\cite[Corollaries~2 and~4]{zhang2025optimal}}]\label{thm:ZLH-lb}
Let $\mathcal C$ be an $(n,n-2,2)$ MDS array code over $\F_q$. Then
\[
\beta(\mathcal C)\ge \left\lceil \frac{5n-10}{4}\right\rceil,
\qquad
\gamma(\mathcal C)\ge \left\lceil \frac{4n-7}{3}\right\rceil.
\]
\end{theorem}

The authors of \cite{zhang2025optimal} also gave two explicit constructions,
denoted by $C_1$ and $C_2$, which already suggest that these bounds are not yet
best possible. More precisely, the construction $C_1$ is defined over $\F_q$
for $n\le q-3$ and satisfies
\[
\beta(C_1)=\left\lceil \frac{5n-8}{4}\right\rceil,
\]
while the construction $C_2$ is defined over $\F_q$ for $n\le q-1$ and
satisfies
\[
\gamma(C_2)=\left\lceil \frac{4n-6}{3}\right\rceil.
\]

One notices that
\[
\left\lceil \frac{5n-8}{4}\right\rceil-\left\lceil \frac{5n-10}{4}\right\rceil
=
\begin{cases}
1, & n\equiv 1,2 \pmod 4,\\
0, & n\equiv 0,3 \pmod 4,
\end{cases}
\]
and
\[
\left\lceil \frac{4n-6}{3}\right\rceil-\left\lceil \frac{4n-7}{3}\right\rceil
=
\begin{cases}
1, & n\equiv 1 \pmod 3,\\
0, & n\equiv 0,2 \pmod 3.
\end{cases}
\]
In particular, when $n\equiv 1,2\pmod 4$, the value $\beta(C_1)$ exceeds the
lower bound in Theorem~\ref{thm:ZLH-lb} by exactly $1$, and when $n\equiv 1\pmod 3$, the value $\gamma(C_2)$ exceeds the corresponding
lower bound in Theorem~\ref{thm:ZLH-lb} by exactly $1$.

We shall show in this section that the issue does not lie in the constructions
$C_1$ and $C_2$: they are already optimal (apart from some small-parameter cases). Rather, the lower bounds in
Theorem~\ref{thm:ZLH-lb} are not yet tight and need to be sharpened to
\[
\beta(\mathcal C)\ge \left\lceil \frac{5n-8}{4}\right\rceil,
\qquad
\gamma(\mathcal C)\ge \left\lceil \frac{4n-6}{3}\right\rceil.
\]

In the remainder of this section, we always assume, as at the beginning of Section~3 of~\cite{zhang2025optimal}, that the parity-check matrix~$H$ is in systematic form:
\[
    H=
    \begin{bmatrix}
    A_1&\cdots&A_{n-2}&I_2&0\\
    B_1&\cdots&B_{n-2}&0&I_2
    \end{bmatrix},
    \]
  with each \(A_j,B_j\in\GL_2(\F_q)\); this is without loss of generality. To facilitate comparison with \cite{zhang2025optimal}, we align our notation in this section with theirs as far as possible. Although this may clash with some symbols used earlier in the present paper, no ambiguity should arise.

\subsection{The Repair Bandwidth Lower Bound}\label{subsec:repair-bd-lb}

Before proceeding, we recall some notation from \cite[Subsection~3.1]{zhang2025optimal}. For any $i\in [n]$, let $M_i\in\mathcal{M}_i(\mathcal{C})$ be such that
$$\beta_i(\mathcal{C})=\BW_i(M_i).$$
By \cite[Lemma~6]{zhang2025optimal}, we may assume that
\[
M_iH_j\neq 0
\qquad\text{for all }j\in[n].
\]
For any matrix
\(M\in\mathbb{F}_q^{2\times 4}\), define
\[
\mathcal R(M):=\{\,j\in[n]: \rank(MH_j)=2\,\}.
\]
Then
\[
\beta_i(\mathcal C)=\BW_i(M_i)=n-2+|\mathcal R(M_i)|.
\]

Let \(\prec\) denote the total order on subsets of \([n]\) introduced in \cite[Definition~1]{zhang2025optimal}. For each \(i\in[n]\), let
\(\mathcal R_i\) denote the minimum set in
\[
\{\,\mathcal R(M): M\in\mathcal M_i(\mathcal C)\,\}
\]
with respect to \(\prec\), and let $\widetilde{M}_i\in\mathcal{M}_i(\mathcal{C})$ be such that $\mathcal R(\widetilde{M}_i)=\mathcal R_i$. Then we have
$$\beta_i(\mathcal{C})\ge n-2+|\mathcal R_i|=n-2+|\mathcal R(\widetilde{M}_i)|\geq\BW_i(\widetilde{M}_i).$$
By definition, we have $\beta_i(\mathcal{C})\le\BW_i(\widetilde{M}_i)$, which implies that $\beta_i(\mathcal{C})=\BW_i(\widetilde{M}_i)$ and thus
$$\widetilde{M}_iH_j\ne 0\qquad\text{for all }j\in[n].$$
Therefore, we may assume that $M_i=\widetilde{M}_i$ and 
$$\beta_i(\mathcal{C})=n-2+|\mathcal R_i|\qquad\text{for all }i\in[n].$$

After relabeling the coordinates if necessary, we may assume that \[ \mathcal R_1\preceq \mathcal R_2\preceq \cdots \preceq \mathcal R_n. \] 
\cite[Lemma~8]{zhang2025optimal} ensures that there exists some \(i\in[n-2]\) such that
\[
\mathcal R_i\prec \mathcal R_{n-1}.
\]
Hence we can define
\[
t:=\max\{\,i\in[n-2]: \mathcal R_i\prec \mathcal R_{n-1}\,\}.
\]

Since we only study repair bandwidth in this subsection, we may further apply a suitable block-wise invertible column transformation and assume that
\[
H_j=
\begin{bmatrix}
I_2\\
W_j
\end{bmatrix}
\quad (j\in[n-2]),\qquad
H_{n-1}=
\begin{bmatrix}
I_2\\
0
\end{bmatrix},
\qquad
H_n=
\begin{bmatrix}
0\\
I_2
\end{bmatrix},
\]
with each \(W_j\in \GL_2(\F_q)\). For each \(i\in[t]\), let $\mathbf{u}_i$ and $\mathbf{v}_i$ be the vectors in $\F_q^2$ defined in \cite[Lemma~9]{zhang2025optimal}.

The following lemma will be used repeatedly in what follows.

\begin{lemma}\label{lem:three-witness}
  Let \(A,B,C\) be three pairwise distinct sets among
  \[
  \mathcal R_1,\dots,\mathcal R_t.
  \]
  Let \(P\subseteq [n-2]\) be such that \(P\cap A=P\cap B=P\cap C=\varnothing\). Then all matrices \(W_j\) with \(j\in P\) lie in the same projective class.
\end{lemma}

\begin{proof}
  Fix \(j,j'\in P\). Choose \(a,b,c\in[t]\) such that
  \[
  \mathcal R_a=A,\qquad \mathcal R_b=B,\qquad \mathcal R_c=C.
  \]
  Since \(j,j'\notin \mathcal R_a,\mathcal R_b,\mathcal R_c\), \cite[Lemma~10]{zhang2025optimal} gives
  \[
  \langle \mathbf v_aW_j\rangle=\langle \mathbf u_a\rangle
  =\langle \mathbf v_aW_{j'}\rangle,
  \]
  and similarly
  \[
  \langle \mathbf v_bW_j\rangle=\langle \mathbf u_b\rangle
  =\langle \mathbf v_bW_{j'}\rangle,
  \qquad
  \langle \mathbf v_cW_j\rangle=\langle \mathbf u_c\rangle
  =\langle \mathbf v_cW_{j'}\rangle.
  \]
  Arguing exactly as in the proof of \cite[Lemma~12]{zhang2025optimal}, one sees that
  \[
  \langle \mathbf u_a\rangle,\ \langle \mathbf u_b\rangle,\ \langle \mathbf u_c\rangle
  \]
  are pairwise distinct, and so are
  \[
  \langle \mathbf v_a\rangle,\ \langle \mathbf v_b\rangle,\ \langle \mathbf v_c\rangle.
  \]
  Hence \cite[Theorem~11]{zhang2025optimal} implies that \(W_j\) and \(W_{j'}\) lie in the same projective class.
\end{proof}

The case \(n\equiv 2\pmod 4\) is comparatively simple; we treat it first.

\begin{proposition}\label{prop:beta-mod-2}
  Let \(\mathcal{C}\) be an \((n,n-2,2)\) MDS array code. Assume that $n\ge 14$ and $n\equiv 2 \pmod 4$. Then
  \[
  \beta(\mathcal{C})\ge \left\lceil \frac{5n-8}{4}\right\rceil.
  \]
  \end{proposition}
  
  \begin{proof}
  By Theorem~\ref{thm:ZLH-lb}, it suffices to rule out the equality case
  \[
  \beta(\mathcal C)=\left\lceil \frac{5n-10}{4}\right\rceil.
  \]
  Write $n=4m+2$ with $m\ge 3$. Then
  \[
  \beta(\mathcal C)=\left\lceil \frac{5n-10}{4}\right\rceil=5m.
  \]

  Since \(\beta_i(\mathcal C)=n-2+|\mathcal R_i|\) and
  \(\beta(\mathcal C)=\max_{i\in[n]}\beta_i(\mathcal C)\), we have
  \[
  |\mathcal R_i|\le m
  \qquad\text{for all }i\in[n].
  \]
  On the other hand, \cite[Theorem~1]{zhang2025optimal} gives
  \[
  \bar\beta(\mathcal C)\ge \frac{5(n-2)}{4}=5m,
  \]
  where
  \[
  \bar\beta(\mathcal C):=\frac{1}{n}\sum_{i\in[n]}\beta_i(\mathcal C).
  \]
  Since \(\bar\beta(\mathcal C)\le \beta(\mathcal C)=5m\), it follows that
  \[
  \bar\beta(\mathcal C)=\beta(\mathcal C)=5m.
  \]
  Hence
  \[
  |\mathcal R_i|=m
  \qquad\text{for all }i\in[n].
  \]
  
  We claim that \(t=n-2\). Indeed, if \(t<n-2\), then Case~1 in the proof of
  \cite[Theorem~1]{zhang2025optimal} yields
  \[
  \sum_{i=1}^{n} |\mathcal R_i|
     \ge \frac{(n-1)n}{4},
  \]
  whereas from $|\mathcal R_i|=m$ for all $i\in[n]$, we have
  \[
  \sum_{i=1}^{n} |\mathcal R_i|
     =nm=\frac{n(n-2)}{4},
  \]
  a contradiction. Therefore $t=n-2$.
  
  We are thus in Case~2 of the proof of \cite[Theorem~1]{zhang2025optimal}. Let
  \[
  S_1,\dots,S_d
  \]
  be the distinct sets among \(\mathcal R_1,\dots,\mathcal R_{n-2}\), and let
  \[
  s_a:=|\{\,i\in[n-2]: \mathcal R_i=S_a\,\}|
  \qquad (a\in[d])
  \]
  be their multiplicities. Then \(d\le 4\) by \cite[Lemma~4]{zhang2025optimal}, and the proof of \cite[Theorem~1]{zhang2025optimal} gives
  \[
  \sum_{i=1}^{n-2}|\mathcal R_i|
  =
  \sum_{a=1}^d s_a|S_a|
  \ge \sum_{a=1}^d s_a^2
  \ge \frac{1}{4}\Bigl(\sum_{a=1}^d s_a\Bigr)^2.
  \]
  Since \(|\mathcal R_i|=m\) for all \(i\in[n]\), we have
  \[
  \sum_{i=1}^{n-2}|\mathcal R_i|=(n-2)m=4m^2
  =\frac14(4m)^2
  =\frac{1}{4}\Bigl(\sum_{a=1}^d s_a\Bigr)^2.
  \]
  Hence every inequality above is an equality. In particular,
  \[
  d=4,
  \qquad
  s_1=s_2=s_3=s_4=m,
  \qquad
  |S_a|=m \ \text{for all }a.
  \]
  
  For each \(a\in\{1,2,3,4\}\), define
  \[
  P_a:=\{\,i\in[n-2]: \mathcal R_i=S_a\,\}.
  \]
  Then \(|P_a|=s_a=m\), the sets \(P_1,P_2,P_3,P_4\) form a partition of \([n-2]\), and
  since \(i\in\mathcal R_i\) for every \(i\in[n]\), we have \(P_a\subseteq S_a\).
  As \(|P_a|=|S_a|=m\), it follows that
  \[
  P_a=S_a
  \qquad\text{for each }a.
  \]
  Hence
  \[
  [n-2]=P_1\sqcup P_2\sqcup P_3\sqcup P_4,
  \qquad |P_a|=m,
  \]
  and
  \[
  \mathcal R_i=P_a
  \qquad\text{for every }i\in P_a.
  \]
  
  By Lemma~\ref{lem:three-witness}, for each
\(a\in\{1,2,3,4\}\), all matrices \(W_j\) with \(j\in P_a\) lie in the same
projective class. Thus there exist scalars
\(\lambda_j\in\F_q^\times\) and an invertible matrix \(T_a\) such that
\[
W_j=\lambda_jT_a
\qquad (j\in P_a).
\]
Moreover, the scalars \(\lambda_j\) \((j\in P_a)\) are pairwise distinct; otherwise, if
\(\lambda_j=\lambda_{j'}\) for some distinct \(j,j'\in P_a\), then \(W_j=W_{j'}\), hence
\(H_j=H_{j'}\), contradicting the MDS property.

  Write
  \[
  M_{n-1}=[\,U_{n-1}\ \ V_{n-1}\,]\qquad\text{with }U_{n-1},V_{n-1}\in \F_q^{2\times 2}.
  \]
  Since \(n-1\in\mathcal R(M_{n-1})\), $U_{n-1}=M_{n-1}H_{n-1}$ is invertible. For any \(j\in P_a\), one has
  \[
  M_{n-1}H_j=U_{n-1}+V_{n-1}W_j=U_{n-1}+\lambda_jV_{n-1}T_a.
  \]
  Therefore
  \[
  \det(M_{n-1}H_j)=\det(U_{n-1}+\lambda_jV_{n-1}T_a),
  \]
  which is a polynomial in \(\lambda_j\) of degree at most \(2\), and is not
  identically zero because \(\det(U_{n-1})\neq 0\). Hence, for each block \(P_a\), there are at most two indices \(j\in P_a\) such that \(\rank(M_{n-1}H_j)=1\), that is, at most two indices in \(P_a\) lie outside \(\mathcal R(M_{n-1})\). Summing over the four blocks, there are at most \(8\) indices \(j\in[n-2]\) outside \(\mathcal R(M_{n-1})\).
  
  On the other hand, since \(|\mathcal R(M_{n-1})|=m\) and \(n-1\in\mathcal R(M_{n-1})\),
  among the \(n-2=4m\) indices at least \(3m+1\) lie outside \(\mathcal R(M_{n-1})\), which is \(>8\) for \(m\ge 3\). This contradiction shows that
  \[
  \beta(\mathcal C)\neq \left\lceil \frac{5n-10}{4}\right\rceil.
  \]
  Therefore \(\beta(\mathcal C)\ge \left\lceil \frac{5n-8}{4}\right\rceil\).
  \end{proof}

  The case \(n\equiv 1\pmod 4\) is somewhat more involved; we organize the
  complete proof into several lemmas.

    \begin{lemma}\label{lem:mod-1-case2}
      Let \(\mathcal C\) be an \((n,n-2,2)\) MDS array code. Assume that
      \[
      n=4m+5,\qquad m\ge 2,
      \]
      and
      \[
      \beta(\mathcal C)=\left\lceil \frac{5n-10}{4}\right\rceil=5m+4.
      \]
      Then
      \[
      t\neq n-2.
      \]
    \end{lemma}
    
    \begin{proof}
      Since \(\beta_i(\mathcal C)=n-2+|\mathcal R_i|\) and
      \(\beta(\mathcal C)=\max_{i\in[n]}\beta_i(\mathcal C)\), we have
      \[
      |\mathcal R_i|\le m+1
      \qquad\text{for all }i\in[n].
      \]
    
      Assume that \(t=n-2\). By \cite[Lemma~12]{zhang2025optimal}, the distinct sets among $\mathcal R_1,\dots,\mathcal R_{n-2}$ are at most four. Let these distinct sets be
  \[
  S_1\prec S_2\prec \cdots \prec S_d
  \qquad (d\le 4),
  \]
  and let
  \[
  P_a:=\{\,i\in[n-2]: \mathcal R_i=S_a\,\},
  \qquad
  s_a:=|P_a|
  \qquad (a\in[d]).
  \]
  Then
  \[
  \sum_{a=1}^d s_a=n-2=4m+3.
  \]
  Since \(i\in\mathcal R_i\) for every \(i\in[n]\), one has
  \[
  P_a\subseteq S_a,
  \qquad
  s_a\le |S_a|\le m+1
  \qquad\text{for all }a.
  \]
  If \(d\le 3\), then
  \[
  4m+3=\sum_{a=1}^d s_a\le d(m+1)\le 3(m+1)=3m+3,
  \]
  a contradiction. Hence $d=4$. Moreover, since
  \[
  \sum_{a=1}^4 s_a=4m+3
  \qquad\text{and}\qquad
  s_a\le m+1
  \quad\text{for all }a,
  \]
  there exists a unique index \(b\in\{1,2,3,4\}\) such that $s_b=m$ while $s_a=m+1$ for all $a\neq b$. Therefore
  \[
  P_a=S_a
  \qquad\text{for all }a\neq b.
  \]
  Set $E_b:=S_b\setminus P_b$. Since \(P_b\subseteq S_b\), \(|P_b|=m\), and \(|S_b|\le m+1\), we have $|E_b|\le 1$.
    
      For each \(a\neq b\), we show that all matrices \(W_j\) with $j\in P_a\setminus E_b$ lie in the same projective class. Let
      \[
      \{c,d\}=\{1,2,3,4\}\setminus\{a,b\}.
      \]
      For every \(j\in P_a\setminus E_b\), we have $j\notin S_b$ by the definition of \(E_b\), and
      \[
      j\notin S_c,\qquad j\notin S_d
      \]
      since \(S_c=P_c\) and \(S_d=P_d\) are disjoint from \(P_a\). Hence, by
      Lemma~\ref{lem:three-witness}, all matrices \(W_j\) with \(j\in P_a\setminus E_b\)
      lie in the same projective class.
  
      For each \(a\neq b\), choose an invertible matrix \(A_a\) representing the
      corresponding projective class. Then for every \(j\in P_a\setminus E_b\), $W_j=\lambda_jA_a$ for some \(\lambda_j\in\F_q^\times\). These scalars are pairwise distinct;
      otherwise, if \(\lambda_j=\lambda_{j'}\) for some distinct \(j,j'\), then
      \(W_j=W_{j'}\), hence \(H_j=H_{j'}\), contradicting the MDS property.
    
      Write
      \[
      M_{n-1}=[\,U_{n-1}\ \ V_{n-1}\,]
      \qquad\text{with }U_{n-1},V_{n-1}\in\F_q^{2\times 2}.
      \]
      Since \(n-1\in\mathcal R(M_{n-1})\), the matrix $U_{n-1}=M_{n-1}H_{n-1}$ is invertible. For every \(a\neq b\) and every \(j\in P_a\setminus E_b\), one has
      \[
      M_{n-1}H_j=U_{n-1}+V_{n-1}W_j
               =U_{n-1}+\lambda_jV_{n-1}A_a.
      \]
      Therefore
      \[
      \det(M_{n-1}H_j)=\det(U_{n-1}+\lambda_jV_{n-1}A_a),
      \]
      which is a polynomial in \(\lambda_j\) of degree at most \(2\), and is not
      identically zero because \(\det(U_{n-1})\neq 0\). Hence, for each exact block
      \(P_a\) with \(a\neq b\), at most two indices in \(P_a\setminus E_b\) lie outside
      \(\mathcal R(M_{n-1})\).
    
      Since \(|E_b|\le 1\), among the three exact blocks \(P_a\) with \(a\neq b\), at most
      one can intersect \(E_b\). Therefore at most one of them contributes \(3\) indices
      outside \(\mathcal R(M_{n-1})\), while the other two contribute at most \(2\) each.
      Together with the small block \(P_b\), this gives
      \[
      |[n-2]\setminus \mathcal R(M_{n-1})|\le m+2+2+3=m+7.
      \]
    
      On the other hand, since \(|\mathcal R(M_{n-1})|=|\mathcal R_{n-1}|\le m+1\) and
      \(n-1\in\mathcal R(M_{n-1})\), among the \(n-2=4m+3\) indices in \([n-2]\) at least
      \[
      (4m+3)-m=3m+3
      \]
      lie outside \(\mathcal R(M_{n-1})\). If \(m\ge 3\), then $3m+3>m+7$, a contradiction. Hence it remains to consider the case \(m=2\), that is, \(n=13\).
  
      We now assume \(m=2\). Since
      \[
      9=3m+3\le |[11]\setminus \mathcal R(M_{12})|\le m+7=9,
      \]
      we have
      \[
      |[11]\setminus \mathcal R(M_{12})|=9,
      \qquad
      |\mathcal R(M_{12})\cap [11]|=2.
      \]
      Since \(12\in \mathcal R(M_{12})\) and \(|\mathcal R(M_{12})|=|\mathcal R_{12}|\le 3\), it follows that $|\mathcal R_{12}|=3$. Analyzing the equality conditions in the preceding estimates, we obtain:
      \begin{itemize}
          \item \(P_b\) contributes exactly \(2\) indices outside \(\mathcal R(M_{12})\);
          \item \(|E_b|=1\);
          \item exactly one of the three sets \(P_a\) with \(a\neq b\) intersects \(E_b\), this set contributes exactly \(3\) indices outside \(\mathcal R(M_{12})\), and each of the other two contributes exactly \(2\);
          \item the set \(\mathcal R(M_{12})\cap [11]\) consists of exactly one index from each of the two sets \(P_a\) with \(a\neq b\) and \(P_a\cap E_b=\varnothing\). We call such a set $P_a$ a \emph{clean block} (for convenient reference below).   
      \end{itemize}
  
      For every \(a>b\) and every \(j\in P_a\), we have $S_b\prec S_a=\mathcal R_j$, so \cite[Lemma~7]{zhang2025optimal} implies that \(j\notin S_b\). It follows that $E_b\cap P_a=\varnothing$ for every $a>b$. Since $|E_b|=1$ and $E_b\cap P_a\neq\varnothing$ for some $a\ne b$, we have
  \[
  E_b\subseteq P_1\sqcup\cdots\sqcup P_{b-1}.
  \]
  In particular, \(b\neq 1\). 
  
  We claim that $E_b\subseteq P_{b-1}$. Since \(|E_b|=1\) and \(|P_b|=2=|P_{b-1}|-1\), we have
  \[
  |S_b|=|P_b\cup E_b|=3=|P_{b-1}|=|S_{b-1}|.
  \]
  If the unique element \(e\in E_b\) lay in some \(P_a\) with \(a\le b-2\), then
  \(e<\min(P_{b-1})\), because the blocks \(P_1,P_2,P_3,P_4\) occur in this order in
  \([11]\). Hence, when comparing \(S_b=P_b\cup\{e\}\) and \(S_{b-1}=P_{b-1}\) in
  the total order \(\prec\), the first differing element is smaller for \(S_b\). Therefore
  \[
  S_b\prec S_{b-1},
  \]
  contradicting
  \[
  S_{b-1}\prec S_b.
  \]
  It follows that $E_b\subseteq P_{b-1}$.
  
  Since $t=n-2=11$, we have $S_4\prec \mathcal R_{12}$. However, this is impossible. First, we have $|\mathcal{R}_{12}|=|S_4|=3$. Moreover,
    \begin{itemize}
      \item if \(b=2\), then $E_2\subseteq P_1$ and the two clean blocks are \(P_3\) and \(P_4\), so the
      smallest element of \(\mathcal R_{12}\) lies in \(P_3\), whereas every element of
      \(S_4=P_4\) lies after \(P_3\); hence \(\mathcal R_{12}\prec S_4\);
      \item if \(b=3\), then $E_3\subseteq P_2$ and the clean blocks are \(P_1\) and \(P_4\), so the
      smallest element of \(\mathcal R_{12}\) lies in \(P_1\), whereas every element of
      \(S_4=P_4\) lies after \(P_1\); hence \(\mathcal R_{12}\prec S_4\);
      \item if \(b=4\), then $E_4\subseteq P_3$ and the clean blocks are \(P_1\) and \(P_2\), so the smallest element of
      \(\mathcal R_{12}\) lies in \(P_1\), whereas the smallest element of \(S_4\) lies
      in \(P_3\); hence \(\mathcal R_{12}\prec S_4\).
    \end{itemize}
    This contradiction shows that \(t\neq n-2\).
    \end{proof}

  \begin{lemma}\label{lem:mod-1-case1-structure}
    Let \(\mathcal C\) be an \((n,n-2,2)\) MDS array code. Assume that
    \[
    n=4m+5,
    \qquad
    \beta(\mathcal C)=5m+4,
    \qquad m\ge 0,\qquad
    t<n-2.
    \]
    Then
    \[
    t=3m+3,
    \qquad
    |\mathcal R_i|=m+1
    \quad\text{for all }i\in[n],
    \]
    and there exists a partition
    \[
    [n-2]=P_1\sqcup P_2\sqcup P_3\sqcup Q,
    \]
    where
    \[
    |P_1|=|P_2|=|P_3|=m+1,
    \qquad
    |Q|=m,
    \]
    such that
    \[
    \mathcal R_i=P_a
    \qquad(i\in P_a,\ a\in\{1,2,3\}),
    \]
    and
    \[
    \mathcal R_i=Q\cup\{n-1\}
    \qquad\text{for all }i\in Q\cup\{n-1\}.
    \]
  \end{lemma}
  
  \begin{proof}
    Since \(\beta_i(\mathcal C)=n-2+|\mathcal R_i|\) and
    \(\beta(\mathcal C)=\max_{i\in[n]}\beta_i(\mathcal C)\), we have
    \[
    |\mathcal R_i|\le m+1
    \qquad\text{for all }i\in[n].
    \]
    Hence
    \begin{equation}\label{eq:mod-1-case1-sum-R}
    \sum_{i=1}^{n}|\mathcal R_i|
    \le n(m+1)=(4m+5)(m+1).
    \end{equation}
  
    Since \(t<n-2\), we are in Case~1 of the proof of \cite[Theorem~1]{zhang2025optimal},
    where
    \begin{equation}\label{eq:mod-1-case1-sum-R-ineq}
    \sum_{i=1}^{n}|\mathcal R_i|
    \ge
    \frac{t^2}{3}+(n-t)\max\!\left\{\frac{t}{3},\,n-1-t\right\}.
    \end{equation}
    If \(t\le 3m+3\), then
    \[
    \frac{t^2}{3}+(4m+5-t)(4m+4-t)-(4m+5)(m+1)
    =
    \frac{(4t-12m-15)(t-3m-3)}{3}\ge 0,
    \]
    with equality if and only if \(t=3m+3\). If \(t\ge 3m+3\), then
    \[
    \frac{t^2}{3}+(4m+5-t)\frac{t}{3}-(4m+5)(m+1)
    =
    \frac{(4m+5)(t-3m-3)}{3}\ge 0,
    \]
    with equality if and only if \(t=3m+3\). Therefore
    \begin{equation}\label{eq:mod-1-case1-sum-R-ineq-2}
    \frac{t^2}{3}+(n-t)\max\!\left\{\frac{t}{3},\,n-1-t\right\}\ge (4m+5)(m+1),
    \end{equation}
    and equality can occur only when
    \[
    t=3m+3.
    \]
    Combining Inequalities~\eqref{eq:mod-1-case1-sum-R},~\eqref{eq:mod-1-case1-sum-R-ineq} and~\eqref{eq:mod-1-case1-sum-R-ineq-2}, we obtain
    \[
      t=3m+3,
    \qquad
   |\mathcal R_i|=m+1
    \quad\text{for all }i\in[n].
    \]

    As in the equality analysis in Case~1 of the proof of \cite[Theorem~1]{zhang2025optimal},
    let $S_1,\dots,S_d\ (d\le 3)$ be the distinct sets among \(\mathcal R_1,\dots,\mathcal R_t\), and let
    \[
    s_a:=|\{\,i\in[t]: \mathcal R_i=S_a\,\}|
    \qquad (a\in[d]).
    \]
    Then
    \[
    \sum_{i=1}^{t}|\mathcal R_i|
    =
    \sum_{a=1}^d s_a|S_a|
    \ge \sum_{a=1}^d s_a^2
    \ge\frac1d\Bigl(\sum_{a=1}^3 s_a\Bigr)^2\ge \frac13\Bigl(\sum_{a=1}^3 s_a\Bigr)^2
    =\frac{t^2}{3}.
    \]
    Since \(t=3m+3\) and \(|\mathcal R_i|=m+1\) for all \(i\), the left-hand side is
    \[
    t(m+1)=\frac{t^2}{3}.
    \]
    Hence every inequality above is an equality. In particular,
    \[ 
    d=3,\qquad s_1=s_2=s_3=m+1,
    \qquad
    |S_a|=m+1
    \quad\text{for all }a.
    \]
  
    For each \(a\in\{1,2,3\}\), define
    \[
    P_a:=\{\,i\in[t]: \mathcal R_i=S_a\,\}.
    \]
    Then \(|P_a|=m+1\), the sets \(P_1,P_2,P_3\) form a partition of \([t]\), and
    \(P_a\subseteq S_a\). Since \(|P_a|=|S_a|=m+1\), it follows that
    \[
    P_a=S_a
    \qquad\text{for each }a.
    \]
    Hence
    \[
    \mathcal R_i=P_a
    \qquad(i\in P_a,\ a\in\{1,2,3\}).
    \]
  
    Let
    \[
    Q:=\{t+1,t+2,\dots,n-2\}.
    \]
    Then \(|Q|=(n-2)-t=m\). Since
    \[
    t=\max\{\,i\in[n-2]: \mathcal R_i\prec \mathcal R_{n-1}\,\},
    \]
    we have
    \[
    \mathcal R_i=\mathcal R_{n-1}
    \qquad\text{for all }i\in Q.
    \]
    Since \(i\in\mathcal R_i\) for every \(i\in Q\) and \(n-1\in\mathcal R_{n-1}\), we get
    \[
    Q\cup\{n-1\}\subseteq \mathcal R_{n-1}.
    \]
    As both sides have size \(m+1\), it follows that
    \[
    \mathcal R_{n-1}=Q\cup\{n-1\}.
    \]
    Hence
    \[
    \mathcal R_i=Q\cup\{n-1\}
    \qquad\text{for all }i\in Q\cup\{n-1\}.
    \]
  \end{proof}
  
  \begin{lemma}\label{lem:mod-1-geometry}
    No \((n,n-2,2)\) MDS array code \(\mathcal C\), where \(n=4m+5\) with \(m\ge 2\),
    admits a partition
    \[
    [n-2]=P_1\sqcup P_2\sqcup P_3\sqcup Q,
    \]
    with
    \[
    |P_1|=|P_2|=|P_3|=m+1,
    \qquad
    |Q|=m,
    \]
    such that
    \[
    \mathcal R_i=P_a
    \qquad(i\in P_a,\ a\in\{1,2,3\}),
    \]
    and
    \[
    \mathcal R_i=Q\cup\{n-1\}
    \qquad\text{for all }i\in Q\cup\{n-1\}.
    \]
  \end{lemma}

  \begin{proof}
    Assume to the contrary that such a partition exists. For each \(a\in\{1,2,3\}\), let \(\{b,c\}=\{1,2,3\}\setminus\{a\}\), and choose
    \(i_b\in P_b\) and \(i_c\in P_c\). Since
    \[
    \mathcal R_{i_b}=P_b,\qquad \mathcal R_{i_c}=P_c,
    \]
    we have \(P_a\cap \mathcal R_{i_b}=P_a\cap \mathcal R_{i_c}=\varnothing\). Hence,
    for every \(j\in P_a\), \cite[Lemma~10]{zhang2025optimal} gives
    \[
    \langle \mathbf v_{i_b}W_j\rangle=\langle \mathbf u_{i_b}\rangle,
    \qquad
    \langle \mathbf v_{i_c}W_j\rangle=\langle \mathbf u_{i_c}\rangle.
    \]
    Exactly as in the proof of \cite[Lemma~12]{zhang2025optimal}, the fact that
    \(\mathcal R_{i_b}\neq \mathcal R_{i_c}\) implies
    \[
    \langle \mathbf u_{i_b}\rangle\neq \langle \mathbf u_{i_c}\rangle,
    \qquad
    \langle \mathbf v_{i_b}\rangle\neq \langle \mathbf v_{i_c}\rangle.
    \]

    Then we may choose
invertible matrices \(A_a,B_a\in\GL_2(\F_q)\) such that
\[
\langle \mathbf v_{i_b}B_a^{-1}\rangle=\langle (1,0)\rangle,\qquad
\langle \mathbf v_{i_c}B_a^{-1}\rangle=\langle (0,1)\rangle,
\]
and
\[
\langle (1,0)A_a\rangle=\langle \mathbf u_{i_b}\rangle,\qquad
\langle (0,1)A_a\rangle=\langle \mathbf u_{i_c}\rangle.
\]
For each \(j\in P_a\), set
\[
\widetilde W_j:=A_a^{-1}W_jB_a^{-1}.
\]
Then
\[
\langle (1,0)\widetilde W_j\rangle=\langle (1,0)\rangle,\qquad
\langle (0,1)\widetilde W_j\rangle=\langle (0,1)\rangle.
\]
Hence \(\widetilde W_j\) is diagonal, say
\[
\widetilde W_j=
\begin{bmatrix}
x_j&0\\
0&y_j
\end{bmatrix},
\qquad x_j,y_j\in\F_q^\times.
\]
Therefore
\[
W_j=
A_a
\begin{bmatrix}
x_j&0\\
0&y_j
\end{bmatrix}
B_a.
\] Moreover, the \(x_j\)'s are pairwise distinct
    and the \(y_j\)'s are pairwise distinct on \(P_a\); otherwise
    \(W_j-W_{j'}\) would be singular for some distinct \(j,j'\in P_a\), contradicting
    the MDS property.
    
    We now use the repair matrix \(M_{n-1}\). Write
    \[
    M_{n-1}=[\,U_{n-1}\ \ V_{n-1}\,].
    \]
    Since
    \[
    \mathcal R_{n-1}=Q\cup\{n-1\},
    \]
    we have \(n-1\in\mathcal R_{n-1}\) and \(n\notin\mathcal R_{n-1}\). Hence
    \[
    U_{n-1}=M_{n-1}H_{n-1}\in\GL_2(\F_q),
    \qquad
    \det(V_{n-1})=\det(M_{n-1}H_n)=0.
    \]
    For every \(j\in P_a\), we have \(j\notin\mathcal R_{n-1}\), and therefore
    \[
    \det(M_{n-1}H_j)=\det(U_{n-1}+V_{n-1}W_j)=0.
    \]
    Since $U_{n-1}$ is invertible, we have
    $$\det(I_2+U_{n-1}^{-1}V_{n-1}W_j)=0.$$
    Using Sylvester’s determinant identity, we have 
    $$\det(I_2+T_aD_j)=0,$$
    where
    \[
    T_a=B_aU_{n-1}^{-1}V_{n-1}A_a,\qquad D_j:=\begin{bmatrix}x_j&0\\0&y_j\end{bmatrix}.
    \]
    Since $\det(V_{n-1})=0$, we have $\det(T_a)=0$. If 
    $$T_a=\begin{bmatrix}\alpha_a&*\\ *&\beta_a\end{bmatrix},\qquad\alpha_a,\beta_a\in\mathbb{F}_q,$$
    then
    \[
    0=\det(I_2+T_aD_j)=1+\alpha_ax_j+\beta_ay_j+\det(T_a)x_jy_j=1+\alpha_ax_j+\beta_ay_j.
    \]
    Thus all points \((x_j,y_j)\) with
    \(j\in P_a\) lie on the affine line
    \[
    \alpha_a x+\beta_a y+1=0.
    \]
    
    Next we use the repair matrix \(M_n\). Write
    \[
    M_n=[\,U_n\ \ V_n\,].
    \]
    Since \(n\in\mathcal R_n\), the matrix $V_n=M_nH_n$ is invertible. Moreover, \(|\mathcal R_n|=m+1=|\mathcal R_{n-1}|\) and
    \[
    \mathcal R_n\succeq \mathcal R_{n-1}=Q\cup\{n-1\}.
    \]
    By the definition of \(\prec\), this forces
    \[
    \mathcal R_n\cap [t]=\varnothing.
    \]
    Hence
    \[
    P_a\cap \mathcal R_n=\varnothing
    \qquad (a=1,2,3).
    \]
    So for every \(j\in P_a\),
    \[
    \det(U_n+V_nW_j)=0.
    \]
    Setting
    \[
    C_a:=A_a^{-1}V_n^{-1}U_nB_a^{-1},
    \]
    we obtain
    \[
    0=\det(C_a+D_j)
     =x_jy_j+\alpha'_ax_j+\beta'_ay_j+\gamma'_a
    \]
    for suitable \(\alpha'_a,\beta'_a,\gamma'_a\in\F_q\). Hence the points
    \[
    (x_j,y_j),\qquad j\in P_a,
    \]
    all lie on the affine conic
    \[
    xy+\alpha'_a x+\beta'_a y+\gamma'_a=0.
    \] 
    
    Since \(|P_a|=m+1\ge 3\), \(P_a\) provides at least three distinct
    points lying simultaneously on the above line and conic. Therefore the line must
    be a linear factor of the conic polynomial
    \[
    xy+\alpha'_a x+\beta'_a y+\gamma'_a.
    \]
    But any linear factor of this polynomial must be of the form
    \[
    x+c=0
    \qquad\text{or}\qquad
    y+d=0,
    \]
    because the coefficients of \(x^2\) and \(y^2\) are zero while the coefficient
    of \(xy\) is non-zero. This contradicts the pairwise distinctness of the \(x_j\)'s
    and the \(y_j\)'s on \(P_a\). Hence the configuration in the statement of Lemma~\ref{lem:mod-1-geometry} cannot occur.
    \end{proof}
  
  \begin{proposition}\label{prop:beta-mod-1}
    Let \(\mathcal{C}\) be an \((n,n-2,2)\) MDS array code. Assume that \(n\ge 13\) and \(n\equiv 1 \pmod 4\). Then
    \[
    \beta(\mathcal{C})\ge \left\lceil \frac{5n-8}{4}\right\rceil.
    \]
  \end{proposition}
  \begin{proof}
    By Theorem~\ref{thm:ZLH-lb}, it suffices to rule out the equality case
    \[
    \beta(\mathcal C)=\left\lceil \frac{5n-10}{4}\right\rceil.
    \]
    Write $n=4m+5$ with $m\ge 3$. Assume, for contradiction, that
    \[
    \beta(\mathcal C)=\left\lceil \frac{5n-10}{4}\right\rceil=5m+4.
    \]
  
    By Lemma~\ref{lem:mod-1-case2}, one has
    \[
    t\neq n-2.
    \]
    Hence \(t<n-2\). By Lemma~\ref{lem:mod-1-case1-structure}, there exists a partition
    \[
    [n-2]=P_1\sqcup P_2\sqcup P_3\sqcup Q
    \]
    satisfying the conditions in Lemma~\ref{lem:mod-1-case1-structure}. This contradicts
    Lemma~\ref{lem:mod-1-geometry}. Therefore
    \[
    \beta(\mathcal C)\ge 5m+5=\left\lceil \frac{5n-8}{4}\right\rceil.
    \]
  \end{proof}

  Theorem~\ref{thm:ZLH-lb}, Proposition~\ref{prop:beta-mod-2} and Proposition~\ref{prop:io-nmod1} together give the following result:

  \begin{corollary}\label{cor:beta-small-n}
    Let $\mathcal{C}$ be an $(n,n-2,2)$ MDS array code over $\F_q$ with $n\not\in\{5,6,9,10\}$. Then 
    \[
    \beta(\mathcal{C})\ge\left\lceil \frac{5n-8}{4}\right\rceil.
    \]
  \end{corollary}

  \subsection{The Repair I/O Lower Bound}

    Before proceeding, we recall some notation from \cite[Subsection~3.2]{zhang2025optimal}. For any $i\in [n]$, let $\widetilde{M}_i'\in\mathcal{M}_i(\mathcal{C})$ be such that
    $$\gamma_i(\mathcal{C})=\IO_i(\widetilde{M}_i').$$
    By \cite[Lemma~6]{zhang2025optimal}, we may assume that
    \[
    \widetilde{M}_i'H_j\neq 0
    \qquad\text{for all }j\in[n].
    \]
    For any matrix $M\in\mathcal{M}_i(\mathcal{C})$, define
    \[
    \mathcal{N}(M):=\{\,j\in[n]: \nz(MH_j)=2\,\}.
    \]
    Then
    $$\gamma_i(\mathcal{C})=\IO_i(\widetilde{M}_i')=n-2+|\mathcal{N}(\widetilde{M}_i')|.$$

    For each $i\in [n]$, define 
    \[\mathcal M_i'(\mathcal{C})
    :=
    \{\,M\in\F_q^{2\times 4}: \rank(M)=2,\ \nz(MH_i)=2\,\},\]
    let $\mathcal{N}_i$ denote the minimum set in 
    $$\{\,\mathcal{N}(M): M\in\mathcal{M}_i'(\mathcal{C})\,\},$$
    with respect to the total order \(\prec\), and let $M_i'\in\mathcal{M}_i'(\mathcal{C})$ be such that
    $$\mathcal{N}(M_i')=\mathcal{N}_i.$$
    It is clear that $\mathcal{M}_i(\mathcal{C})\subseteq\mathcal{M}_i'(\mathcal{C})$ for every $i\in [n]$, which implies that $\widetilde{M}_i'\in\mathcal{M}_i'(\mathcal{C})$. It follows that
    \begin{equation}\label{eq:gamma-nmod1}
      \gamma_i(\mathcal{C})=n-2+|\mathcal{N}(\widetilde{M}_i')|\ge n-2+|\mathcal{N}_i|.
    \end{equation} 
    Henceforth we will not refer to the matrices $\widetilde{M}_i'$ again, and the reader may disregard them.

    After relabeling the coordinates if necessary, we may assume that
    \[
    \mathcal{N}_1\preceq \mathcal{N}_2\preceq \cdots \preceq \mathcal{N}_n.
    \]
    \cite[Lemma~14]{zhang2025optimal} ensures that there exists some \(i\in[n-2]\) such that
    \[
    \mathcal{N}_i\prec \mathcal{N}_{n-1}.
    \]
    Hence we can define
    \[
    t:=\max\{\,i\in[n-2]: \mathcal{N}_i\prec \mathcal{N}_{n-1}\,\}.
    \]
    
    Recall that we have assumed that the parity-check matrix $H$ is in systematic form:
    \[
    H=
    \begin{bmatrix}
    A_1&\cdots&A_k&I_2&0\\
    B_1&\cdots&B_k&0&I_2
    \end{bmatrix},
    \]
    where each \(A_j,B_j\in\GL_2(\F_q)\). Furthermore, for each \(i\in[t]\), by \cite[Lemma~15]{zhang2025optimal} we may assume that
    \[
    M_i'=
    \begin{bmatrix}
    \mathbf{u}_i&0\\
    0&\mathbf{v}_i
    \end{bmatrix},\qquad \langle \mathbf{u}_i\rangle,\langle \mathbf{v}_i\rangle\in\{\langle \mathbf{e}_1\rangle,\langle \mathbf{e}_2\rangle\},
    \]
    where
    \[
    \mathbf{e}_1=(1,0),\qquad \mathbf{e}_2=(0,1).
    \]

    The proof of Proposition~\ref{prop:io-nmod1} is divided into several lemmas. We begin with one establishing node-wise monomial invariance for repair I/O.

    \begin{lemma}\label{lem:node-wise-monomial-io}
      Let \(\mathcal{C}\) be an \((n,n-2,2)\) MDS array code over \(\F_q\), with parity-check blocks
      \[
      H=[\,H_1\ \cdots\ H_n\,],\qquad H_j\in \F_q^{4\times 2}.
      \]
      For each \(j\in[n]\), let \(P_j\in \F_q^{2\times 2}\) be a monomial matrix, and define
      \[
      \widetilde H_j:=H_jP_j.
      \]
      Let \(\mathcal C'\) be the code with parity-check blocks
      \[
      \widetilde H=[\,\widetilde H_1\ \cdots\ \widetilde H_n\,].
      \]
      Then \(\mathcal C'\) is again an \((n,n-2,2)\) MDS array code. Moreover, for every
      \(M\in \F_q^{2\times 4}\) and every \(j\in[n]\),
      \[
      \nz(M\widetilde H_j)=\nz(MH_j).
      \]
      In particular,
      \[
      \gamma(\mathcal C')=\gamma(\mathcal C).
      \]
      \end{lemma}
      
      \begin{proof}
      Since \(P_j\) is monomial, right-multiplication by \(P_j\) only permutes the two columns of \(MH_j\)
      and rescales them by non-zero scalars. Hence
      \[
      \nz(MH_jP_j)=\nz(MH_j)
      \]
      for every \(M\in\F_q^{2\times 4}\) and every \(j\in[n]\).
      
      For any distinct \(a,b\in[n]\),
      \[
      [\widetilde H_a\ \widetilde H_b]
      =
      [H_a\ H_b]
      \begin{bmatrix}
      P_a&0\\
      0&P_b
      \end{bmatrix},
      \]
      so \([\widetilde H_a\ \widetilde H_b]\) is invertible if and only if \([H_a\ H_b]\) is invertible. Thus
      \(\mathcal C'\) is again an \((n,n-2,2)\) MDS array code.
      
      For each \(i\in[n]\) and every \(M\in\F_q^{2\times 4}\),
      \[
      M\widetilde H_i=MH_iP_i,
      \]
      and \(P_i\) is invertible. Hence \(MH_i\) is invertible if and only if \(M\widetilde H_i\) is invertible, so
      \[
      \mathcal M_i(\mathcal C')=\mathcal M_i(\mathcal C).
      \]
      Therefore \(\IO_i^{\mathcal C'}(M)=\IO_i^{\mathcal C}(M)\) for every
      \(M\in\mathcal M_i(\mathcal C)=\mathcal M_i(\mathcal C')\), and taking minima gives
      \[
      \gamma_i(\mathcal C')=\gamma_i(\mathcal C)
      \qquad\text{for all }i\in[n].
      \]
      Thus \(\gamma(\mathcal C')=\gamma(\mathcal C)\).
      \end{proof}
      
      \begin{lemma}\label{lem:io-mod-1-structure}
      Let \(\mathcal{C}\) be an \((n,n-2,2)\) MDS array code over \(\F_q\). Assume that
      \[
      n=3m+4
      \qquad\text{with }m\ge 1,
      \]
      and
      \[
      \gamma(\mathcal{C})=4m+3.
      \]
      Then
      \[
      t=2m+2,
      \qquad
      |\mathcal{N}_i|=m+1
      \quad\text{for all }i\in[n],
      \]
      and there exists a partition
      \[
      [n-2]=P_1\sqcup P_2\sqcup Q,
      \]
      where
      \[
      |P_1|=|P_2|=m+1,
      \qquad
      |Q|=m,
      \]
      such that
      \[
      \mathcal{N}_i=P_1 \quad(i\in P_1),
      \qquad
      \mathcal{N}_i=P_2 \quad(i\in P_2),
      \]
      and
      \[
      \mathcal{N}_i=Q\cup\{n-1\}
      \qquad\text{for all }i\in Q\cup\{n-1\}.
      \]
      \end{lemma}
      
      \begin{proof}
     By Inequality~\eqref{eq:gamma-nmod1}, we have
      \[
      |\mathcal{N}_i|
      \le\gamma_i(\mathcal{C})-(n-2)
      \le \gamma(\mathcal{C})-(n-2)
      =(4m+3)-(3m+2)=m+1
      \]
      for all \(i\in[n]\) Let $\bar\gamma(\mathcal C):=\frac{1}{n}\sum_{i=1}^{n}\gamma_i(\mathcal C)$. Then by definition, we have
      \begin{equation}\label{eq:bar-gamma-ineq}
        \bar\gamma(\mathcal C)\le \gamma(\mathcal C)=4m+3.
      \end{equation}
      
      We now revisit the proof of \cite[Theorem~3]{zhang2025optimal}. The final inequalities of the proof yield
      $$\bar\gamma(\mathcal{C})\ge\frac{1}{n}\sum\limits_{i=1}^{n}|\mathcal{N}_i|+n-2\ge\frac{4n-7}{3}=4m+3.$$
      Combining this with Inequality~\eqref{eq:bar-gamma-ineq}, we obtain $\bar\gamma(\mathcal{C})=4m+3$ and
      \begin{equation}\label{eq:sum-N-ineq}
        \sum\limits_{i=1}^n|\mathcal{N}_i|=(4m+3)n-(n-2)n=(m+1)n.
      \end{equation}
      Since $|\mathcal{N}_i|\le m+1$ for every $i\in[n]$, it follows that $|\mathcal{N}_i|=m+1$ for every $i\in[n]$

      Moreover, the proof of \cite[Theorem~3]{zhang2025optimal} shows that
      \[
      \sum_{i=1}^{n}|\mathcal{N}_i|
      \ge
      F(t):=
      \frac{t^2}{2}+(n-t)\max\!\left\{\frac t2,\ n-1-t\right\}.
      \]
      \begin{itemize}
      \item If \(t\ge 2m+2\), then \(\frac t2\ge n-1-t\), and hence
      \[
      F(t)=\frac{t^2}{2}+\frac{t}{2}(n-t)=\frac{nt}{2}\ge(m+1)n,
      \]
      with equality if and only if \(t=2m+2\).
      \item If \(t\le 2m+1\), then \(n-1-t\ge \frac t2\), and hence
      \[
      F(t)=\frac{t^2}{2}+(n-t)(n-1-t)=\frac{t^2}{2}+(3m+4-t)(3m+3-t).
      \]
      Differentiating with respect to $t$ yields
      \[
      F'(t)=3t-(6m+7),
      \]
      which is strictly negative for $t\le 2m+1$, so $F(t)$ is strictly decreasing on the interval \(t\le 2m+1\). Therefore its minimum on
      that interval occurs at $(-\infty,2m+1]$, where
      \[
      F(2m+1)
      =(3m+4)(m+1)+\frac52
      >(3m+4)(m+1)=(m+1)n.
      \]
      \end{itemize}
      Thus Equality~\eqref{eq:sum-N-ineq} forces
      \[
      t=2m+2.
      \]
      
      We next analyze the prefix \((\mathcal{N}_i)_{i\in[t]}\). By \cite[Lemma~17(1)]{zhang2025optimal}, one has \(|\{\mathcal N_i:i\in[t]\}|\le 3\). In the proof of \cite[Theorem~3]{zhang2025optimal}, if
      \[
      |\{N_i:i\in[t]\}|=1,
      \]
      then
      \[
      \sum_{i=1}^t |\mathcal N_i|\ge t^2=(2m+2)^2>t(m+1),
      \]
      contradicting \(|\mathcal N_i|=m+1\) for any \(i\in[n]\). If
      \[
      |\{\mathcal N_i:i\in[t]\}|=3,
      \]
      then \cite[Lemma~17(1)]{zhang2025optimal} gives \(t=n-2\), contradicting
      \[
      t=2m+2<3m+2=n-2.
      \]
      Hence there are exactly two distinct sets among \(\mathcal N_1,\dots,\mathcal N_t\). Let them be
      \[
      S_1\prec S_2,
      \]
      and define
      \[
      P_1:=\{\,i\in[t]: \mathcal N_i=S_1\,\},
      \qquad
      P_2:=\{\,i\in[t]: \mathcal N_i=S_2\,\}.
      \]
      Then
      \[
      P_1\sqcup P_2=[t].
      \]
      
      As in Case~2 of the proof of \cite[Theorem~3]{zhang2025optimal},
      \[
      \sum_{i=1}^t |\mathcal N_i|
      =
      |P_1||S_1|+|P_2||S_2|
      \ge |P_1|^2+|P_2|^2
      \ge \frac{(|P_1|+|P_2|)^2}{2}
      =\frac{t^2}{2}.
      \]
      Since
      \[
      \sum_{i=1}^t |\mathcal N_i|=t(m+1)=(2m+2)(m+1)=\frac{t^2}{2},
      \]
      every inequality above is an equality. Therefore
      \[
      |P_1|=|P_2|=m+1,
      \qquad
      |S_1|=|S_2|=m+1.
      \]
      Since \(i\in \mathcal N_i\) for every \(i\in[n]\), we have \(P_a\subseteq S_a\), and hence
      \[
      S_1=P_1,
      \qquad
      S_2=P_2.
      \]
      Thus
      \[
      [t]=P_1\sqcup P_2,
      \qquad
      |P_1|=|P_2|=m+1,
      \]
      and
      \[
      \mathcal N_i=P_1 \quad(i\in P_1),
      \qquad
      \mathcal N_i=P_2 \quad(i\in P_2).
      \]
      
      Let
      \[
      Q:=\{t+1,\dots,n-2\}=\{2m+3,\dots,3m+2\}.
      \]
      Then
      \[
      |Q\cup\{n-1\}|=m+1.
      \]
      By the definition of \(t\), we have
      \[
      \mathcal N_{t+1}=\cdots=\mathcal N_{n-1}.
      \]
      Moreover, for every \(i\in[n-1]\), one has $i\in \mathcal N_i$. Hence, for every
      \[
      i\in \{t+1,\dots,n-1\}=Q\cup\{n-1\},
      \]
      we have $i\in \mathcal N_i=\mathcal N_{n-1}$. Therefore
      \[
      Q\cup\{n-1\}\subseteq \mathcal N_{n-1}.
      \]
      Since \(|\mathcal N_{n-1}|=m+1=|Q\cup\{n-1\}|\), it follows that
      \[
      \mathcal N_{n-1}=Q\cup\{n-1\}.
      \]
      Consequently,
      \[
      \mathcal N_i=Q\cup\{n-1\}
      \qquad
      \text{for all }i\in Q\cup\{n-1\}.
      \]
      \end{proof}
      
      \begin{lemma}\label{lem:io-mod-1-geometry}
      No \((n,n-2,2)\) MDS array code \(\mathcal{C}\), where \(n=3m+4\) with \(m\ge 1\), admits a partition
      \[
      [n-2]=P_1\sqcup P_2\sqcup Q,
      \]
      with
      \[
      |P_1|=|P_2|=m+1,
      \qquad
      |Q|=m,
      \]
      such that
      \[
      \mathcal N_i=P_1 \quad(i\in P_1),
      \qquad
      \mathcal N_i=P_2 \quad(i\in P_2),
      \]
      and
      \[
      \mathcal N_i=Q\cup\{n-1\}
      \qquad\text{for all }i\in Q\cup\{n-1\}.
      \]
      \end{lemma}
      
      \begin{proof}
      Assume to the contrary that such a partition exists. Since \(\mathcal{N}_i\in\{P_1,P_2\}\) for every \(i\in[t]\), and since \(P_1\cup P_2\subsetneq [n-2]\), \cite[Lemma~16]{zhang2025optimal}
      implies that the corresponding pairs
      \[
      (\langle \mathbf{u}_i\rangle,\langle \mathbf{v}_i\rangle),\qquad i\in[t]
      \]
      either coincide or differ in both coordinates. We may assume that
      \[
      (\langle \mathbf{u}_i\rangle,\langle \mathbf{v}_i\rangle)=(\langle \mathbf{e}_1\rangle,\langle \mathbf{e}_1\rangle)
      \qquad (i\in P_1),
      \]
      and
      \[
      (\langle \mathbf{u}_i\rangle,\langle \mathbf{v}_i\rangle)=(\langle \mathbf{e}_2\rangle,\langle \mathbf{e}_2\rangle)
      \qquad (i\in P_2).
      \]
      
      Fix \(j\in P_1\), and choose any \(i\in P_2\). Since $j\notin \mathcal N_i=P_2$, we have $nz(M_i'H_j)=1$. Because \((\langle \mathbf{u}_i\rangle,\langle \mathbf{v}_i\rangle)=(\langle \mathbf{e}_2\rangle,\langle \mathbf{e}_2\rangle)\), this
      means that the second rows of \(A_j\) and \(B_j\) are supported in the same single
      column. Therefore, after a node-wise column permutation, we may assume that
      \[
      A_j=
      \begin{bmatrix}
      a_j & b_j\\
      0 & d_j
      \end{bmatrix},
      \qquad
      B_j=
      \begin{bmatrix}
      \widetilde{a}_j & \widetilde{b}_j\\
      0 & \widetilde{d}_j
      \end{bmatrix}
      \qquad \forall\,j\in P_1.
      \]
      Lemma~\ref{lem:node-wise-monomial-io} guarantees that this can be done without loss of generality. Likewise, we may assume that
      \[
      A_j=
      \begin{bmatrix}
      a_j & 0\\
      c_j & d_j
      \end{bmatrix},
      \qquad
      B_j=
      \begin{bmatrix}
      \widetilde{a}_j & 0\\
      \widetilde{c}_j & \widetilde{d}_j
      \end{bmatrix}
      \qquad \forall\,j\in P_2.
      \]
      
      For distinct \(j,j'\in P_1\), the MDS property implies that $[H_j,H_{j'}]$ is invertible. Note that
      \[
        [H_j,H_{j'}]=
        \begin{bmatrix}
        a_j &b_j& a_{j'}&b_{j'}\\
        0 & d_j& 0& d_{j'}\\
        \widetilde{a}_j & \widetilde{b}_j& \widetilde{a}_{j'}& \widetilde{b}_{j'}\\
        0 & \widetilde{d}_j& 0& \widetilde{d}_{j'}
        \end{bmatrix}
      \]
      Looking at the first and third columns of this matrix, we see that $(a_j,\widetilde{a}_j)\nparallel (a_{j'},\widetilde{a}_{j'})$. Looking at the second and fourth rows of this matrix, we see that the matrix 
      \[
      \begin{bmatrix}
      d_j & d_{j'}\\
      \widetilde{d}_j & \widetilde{d}_{j'}
      \end{bmatrix}
      \]
      is invertible, which implies that $(d_j,\widetilde{d}_j)\nparallel (d_{j'},\widetilde{d}_{j'})$. Similarly, we can show that for any distinct \(j,j'\in P_2\),
      \[
      (a_j,\widetilde{a}_j)\nparallel (a_{j'},\widetilde{a}_{j'}),
      \qquad
      (d_j,\widetilde{d}_j)\nparallel (d_{j'},\widetilde{d}_{j'}).
      \]
      
      Fix any $i\in Q\cup\{n-1\}$ and let $W_i:=\ker(M_i')$. Since $\rank(M_i')=2$, the space $W_i$ has dimension $2$. By assumption, \(n-1\in \mathcal{N}_i\), and thus the matrix \(M_i'H_{n-1}\) has two non-zero columns, so neither column of \(H_{n-1}\) lies in \(W_i\), i.e., $e_1,e_2\notin W_i$. On the other hand, since \(n\notin \mathcal{N}_i\), the matrix \(M_i'H_n\) has exactly one non-zero column,
      so exactly one of $e_3,\ e_4$ belongs to \(W_i\). For any \(j\in P_1\cup P_2\), since \(j\notin\mathcal{N}_i\), the matrix \(M_i'H_j\) has
      exactly one non-zero column. Equivalently, exactly one of the two columns of \(H_j\)
      lies in \(W_i\). For any $j\in P_1\cup P_2$, let $\col_j^{(1)}$ and $\col_j^{(2)}$ be the first and second columns of $H_j$, respectively.
      
      Now consider first the case \(e_3\in W_i\). For every \(j\in P_1\), we have
      \[
      \col_j^{(1)}=(a_j,0,\widetilde a_j,0)^T\in U_1:=\Span\{e_1,e_3\}.
      \]
      Since \(A_j\in\GL_2(\F_q)\), we have \(a_j\neq 0\), so \(\col_j^{(1)}\notin \langle e_3\rangle\). Hence, if \(\col_j^{(1)}\in W_i\), then
      \[
      U_1=\Span\{e_3,\col_j^{(1)}\}\subseteq W_i,
      \]
      which would imply \(e_1\in W_i\), a contradiction. Therefore \(\col_j^{(1)}\notin W_i\) for every \(j\in P_1\). Since exactly one column of \(H_j\) lies in \(W_i\), it follows that
      \[
      \col_j^{(2)}\in W_i\qquad\forall\,j\in P_1.
      \]
      Because \(|P_1|=m+1\ge 2\), we may choose distinct \(j,j'\in P_1\). Then
      \[
      e_3,\ \col_j^{(2)},\ \col_{j'}^{(2)}\in W_i.
      \]
      Since \(\dim W_i=2\), these three vectors are linearly dependent. Restricting to the second and fourth coordinates, we obtain
      \[
      (d_j,\widetilde d_j)\parallel (d_{j'},\widetilde d_{j'}),
      \]
      contradicting the pairwise non-collinearity of \((d_j,\widetilde d_j)\) within \(P_1\).
      
      Next consider the case \(e_4\in W_i\). For every \(j\in P_2\), we have
      \[
      \col_j^{(2)}=(0,d_j,0,\widetilde d_j)^T\in U_2:=\Span\{e_2,e_4\}.
      \]
      Since \(A_j\in\GL_2(\F_q)\), we have \(d_j\neq 0\), so \(\col_j^{(2)}\notin \langle e_4\rangle\). Hence, if \(\col_j^{(2)}\in W_i\), then
      \[
      U_2=\Span\{e_4,\col_j^{(2)}\}\subseteq W_i,
      \]
      which would imply \(e_2\in W_i\), a contradiction. Therefore \(\col_j^{(2)}\notin W_i\) for every \(j\in P_2\). Since exactly one column of \(H_j\) lies in \(W_i\), it follows that
      \[
      \col_j^{(1)}\in W_i\qquad\forall\,j\in P_2.
      \]
      Because \(|P_2|=m+1\ge 2\), we may choose distinct \(s,s'\in P_2\). Then
      \[
      e_4,\ \col_s^{(1)},\ \col_{s'}^{(1)}\in W_i.
      \]
      Since \(\dim W_i=2\), these three vectors are linearly dependent. Restricting to the first and third coordinates, we obtain
      \[
      (a_s,\widetilde a_s)\parallel (a_{s'},\widetilde a_{s'}),
      \]
      contradicting the pairwise non-collinearity of \((a_s,\widetilde a_s)\) within \(P_2\).
      
      In either case we obtain a contradiction. Therefore the configuration in the statement of Lemma~\ref{lem:io-mod-1-geometry} cannot occur.
      \end{proof}

      \begin{proposition}\label{prop:io-nmod1}
      Let \(\mathcal{C}\) be an \((n,n-2,2)\) MDS array code over \(\F_q\). Assume that $n\ge 7$ and $n\equiv 1 \pmod 3$.
      Then
      \[
      \gamma(\mathcal{C})\ge \left\lceil \frac{4n-6}{3}\right\rceil.
      \]
      \end{proposition}
      
      \begin{proof}
        By Theorem~\ref{thm:ZLH-lb}, it suffices to rule out the equality case
        \[
        \gamma(\mathcal{C})=\left\lceil \frac{4n-7}{3}\right\rceil.
        \]
        Write $n=3m+4$ with $m\ge 2$. Assume, for contradiction, that
        \[
        \gamma(\mathcal{C})=\left\lceil \frac{4n-7}{3}\right\rceil=4m+3.
        \]

      By Lemma~\ref{lem:io-mod-1-structure}, there exists a partition
      \[
      [n-2]=P_1\sqcup P_2\sqcup Q
      \]
      satisfying the conditions in Lemma~\ref{lem:io-mod-1-structure}. This contradicts Lemma~\ref{lem:io-mod-1-geometry}.
      Therefore
      \[
      \gamma(\mathcal{C})\ge 4m+4=\left\lceil \frac{4n-6}{3}\right\rceil.
      \]
      \end{proof}

      Theorem~\ref{thm:ZLH-lb} and Proposition~\ref{prop:io-nmod1} together give the following result:

      \begin{corollary}\label{cor:gamma-small-n}
        Let $\mathcal{C}$ be an $(n,n-2,2)$ MDS array code over $\F_q$ with $n\ne 4$. Then 
        \[
        \gamma(\mathcal{C})\ge\left\lceil \frac{4n-6}{3}\right\rceil.
        \]
      \end{corollary}

      The exceptional case $n=4$ will be treated in Proposition~\ref{prop:gamma-opt-4-2-2}.

      \section{The Short-Length Regime}\label{sec:short}

      \subsection{Repair Bandwidth}\label{subsec:bw-short}
      
      For any prime power \(q\), put
      \[
      N_{\BW}(q)=
      \begin{cases}
      \dfrac{4q}{3},&q\equiv0\pmod3,\\[2mm]
      \dfrac{4q+2}{3},&q\equiv1\pmod3,\\[2mm]
      \dfrac{4(q+1)}{3},&q\equiv2\pmod3.
      \end{cases}
      \]
      In this subsection we study the repair bandwidth in the range
      \[
      3\le n\le N_{\BW}(q).
      \]

      We begin with a four-class orbit construction that will be used throughout this
      section. Let \(E:=\F_{q^2}\), viewed as a two-dimensional vector space over
      \(\F_q\). Choose a primitive element \(\xi\in E^\times\), and let
      \[
      A:E\to E,\qquad x\mapsto \xi x,
      \]
      be the \(\F_q\)-linear map given by multiplication by \(\xi\).
      
      The image of \(A\) in \(\PGL(E)\cong\PGL_2(\F_q)\) has order \(q+1\).
      Indeed, the image of \(A^m\) is trivial if and only if multiplication by
      \(\xi^m\) is an \(\F_q\)-scalar map, equivalently, $\xi^m\in\F_q^\times$. Since \(\xi\) generates \(E^\times=\F_{q^2}^{\times}\), the smallest positive
      such \(m\) is
      \[
      [E^\times:\F_q^\times]
      =
      \frac{q^2-1}{q-1}
      =
      q+1.
      \]
      
      Since \(A\) is an \(\F_q\)-linear automorphism of the two-dimensional
      \(\F_q\)-space \(E\), the Cayley--Hamilton theorem gives
      \[
      A^2=\sigma A+\tau\,\id_E
      \]
      for some \(\sigma\in\F_q\) and \(\tau\in\F_q^\times\). Moreover,
      \(\sigma\neq0\). Indeed, if \(\sigma=0\), then
      \[
      \xi^2=\tau\in\F_q^\times,
      \]
      so the order of \(\xi\) would divide \(2(q-1)\), contradicting that \(\xi\)
      has order \(q^2-1\).
      
      We claim that for any \(x\in E\setminus\{0\}\), the projective points
      \[
      \langle x\rangle,\,\langle A(x)\rangle,\,\langle A^2(x)\rangle,\,\dots,\,\langle A^q(x)\rangle\in\mathbb P(E)
      \]
      are pairwise distinct. Indeed, if $\langle A^i(x)\rangle=\langle A^j(x)\rangle$ for some \(0\le i<j\le q\), then $A^{j-i}(x)=\xi^{j-i}x=\lambda x$ for some \(\lambda\in\F_q^\times\). Since \(x\neq0\), this implies that $\xi^{j-i}=\lambda\in\F_q^\times$, contradicting the fact that \(q+1\) is the smallest positive integer \(m\)
      such that \(\xi^m\in\F_q^\times\). Hence the above \(q+1\) projective points
      are pairwise distinct. Since $|\mathbb P(E)|=q+1$, they are precisely all points of \(\mathbb P(E)\cong\mathbb{P}^1(\F_q)\). In particular, for
      every \(x\in E\setminus\{0\}\), the vectors \(x\) and \(A(x)\) form an
      \(\F_q\)-basis of \(E\).
      
      Identify $\mathbb V:=\F_q^4$ with \(E\oplus E\). For any \(x\in E\setminus\{0\}\), define four
      two-dimensional \(\F_q\)-subspaces of \(\mathbb V\) as follows:
      \begin{align*}
        &  \mathcal H_1(x):=\Span_{\F_q}\{(x,0),\,(-A(x),A(x))\}, & \mathcal H_2(x):=\Span_{\F_q}\{(x,-x),\,(0,A(x))\},\\
        &  \mathcal H_3(x):=\Span_{\F_q}\{(x,0),\,(0,A^2(x))\}, & \mathcal H_4(x):=\Span_{\F_q}\{(x,0),\,(0,x)\}.
      \end{align*}
      Fix \(x_0\in E\setminus\{0\}\), and define
      \[
      x_t:=A^t(x_0)=\xi^t x_0,
      \qquad t\in\mathbb Z/(q+1)\mathbb Z.
      \]
      For \(z\in[4]\) and \(t\in\mathbb Z/(q+1)\mathbb Z\), put
      \[
      \mathcal H_{z,t}:=\mathcal H_z(x_t).
      \]
      
      We also define four two-dimensional subspaces of \(\mathbb V\):
      \[
      W_1:=\{0\}\oplus E,
      \qquad
      W_2:=E\oplus\{0\},
      \qquad
      W_3:=\{(u,-u):u\in E\},
      \]
      and
      \[
      W_4:=\{(u,-B(u)):u\in E\},
      \qquad
      B:=\id_E+\frac{\sigma}{\tau}A.
      \]
      For a coordinate whose node subspace is of the form \(\mathcal H_{z,t}\), we
use \(W_z\) as its repair subspace. Note that, for every
      \(t\in\mathbb Z/(q+1)\mathbb Z\), one has
      \[
      \Span_{\F_q}\{x_t,B(x_t)\}
      =
      \Span_{\F_q}\{x_t,A(x_t)\}
      =
      E.
      \]
      In particular, \(B(x_t)\neq0\).

\begin{lemma}\label{lem:orbit-template-subspace}
  For every \(z,z'\in[4]\) and every
  \(t,t'\in\mathbb Z/(q+1)\mathbb Z\), the following properties hold.

\begin{enumerate}[label=(\alph*)]
\item One has
\[
\dim_{\F_q}\bigl(W_z\cap\mathcal H_{z',t'}\bigr)
=
\begin{cases}
0,&z'=z,\\
1,&z'\neq z.
\end{cases}
\]

\item Suppose \(z<z'\), and write
\[
\delta=t'-t\in\mathbb Z/(q+1)\mathbb Z.
\]
Then
\[
\mathcal H_{z,t}\cap\mathcal H_{z',t'}\neq\{0\}
\]
can occur only for the following values of \(\delta\):
\[
\begin{array}{c|c}
(z,z')&\text{exceptional values of }\delta\\ \hline
(1,2)&\{\pm1\}\\
(1,3)&\{0,-1\}\\
(1,4)&\{0,+1\}\\
(2,3)&\{0,-1\}\\
(2,4)&\{0,+1\}\\
(3,4)&\{0,+2\}
\end{array}
\]

\item For \(z=z'\), one has
\[
\mathcal H_{z,t}\cap\mathcal H_{z,t'}\neq\{0\}
\quad\Longleftrightarrow\quad
t=t'.
\]
\end{enumerate}
\end{lemma}

\begin{proof}
We prove the three assertions separately.

\begin{enumerate}[label=(\alph*)]
\item It is convenient to use the quotient maps whose kernels are the subspaces
\(W_z\). Define
\[
\pi_1(u,v)=u,\qquad
\pi_2(u,v)=v,\qquad
\pi_3(u,v)=u+v,\qquad
\pi_4(u,v)=B(u)+v.
\]
Then $W_z=\ker(\pi_z)$ for every $z\in[4]$. Since each \(\mathcal H_{z',t'}\) has dimension \(2\), the rank-nullity theorem
applied to \(\pi_z|_{\mathcal H_{z',t'}}\) gives
\[
\dim_{\F_q}(W_z\cap\mathcal H_{z',t'})
=
2-\dim_{\F_q}\pi_z(\mathcal H_{z',t'}).
\]
Thus it suffices to compute the dimensions of the images
\(\pi_z(\mathcal H_{z',t'})\).

\indent\quad We first consider \(\pi_1\). By definition,
\[
\pi_1(\mathcal H_{1,t'})
=
\Span_{\F_q}\{x_{t'},-A(x_{t'})\}
=
E,
\]
because \(x_{t'}\) and \(A(x_{t'})\) are \(\F_q\)-linearly independent. On the other
hand,
\[
\pi_1(\mathcal H_{2,t'})
=
\pi_1(\mathcal H_{3,t'})
=
\pi_1(\mathcal H_{4,t'})
=
\Span_{\F_q}\{x_{t'}\}.
\]
Hence \(\pi_1\) has two-dimensional image on \(\mathcal H_{1,t'}\) and
one-dimensional image on the other three subspaces.

\indent\quad Similarly,
\[
\pi_2(\mathcal H_{2,t'})
=
\Span_{\F_q}\{-x_{t'},A(x_{t'})\}
=
E,
\]
while
\[
\pi_2(\mathcal H_{1,t'})=\Span_{\F_q}\{A(x_{t'})\},
\qquad
\pi_2(\mathcal H_{3,t'})=\Span_{\F_q}\{A^2(x_{t'})\},
\qquad
\pi_2(\mathcal H_{4,t'})=\Span_{\F_q}\{x_{t'}\}.
\]
Thus \(\pi_2\) has two-dimensional image exactly on \(\mathcal H_{2,t'}\).

\indent\quad For \(\pi_3\), we have
\[
\pi_3(\mathcal H_{3,t'})
=
\Span_{\F_q}\{x_{t'},A^2(x_{t'})\}
=
E.
\]
Indeed, \(x_{t'}\) and \(A^2(x_{t'})\) are linearly independent, since otherwise
$\langle A^2(x_{t'})\rangle=\langle x_{t'}\rangle$, contradicting the fact that the projective orbit has
length \(q+1\). For the remaining three subspaces, one has
\[
\pi_3(\mathcal H_{1,t'})
=
\pi_3(\mathcal H_{4,t'})
=
\Span_{\F_q}\{x_{t'}\},
\qquad
\pi_3(\mathcal H_{2,t'})
=
\Span_{\F_q}\{A(x_{t'})\}.
\]
Thus \(\pi_3\) has two-dimensional image exactly on \(\mathcal H_{3,t'}\).

\indent\quad It remains to consider \(\pi_4\). Recall that
\[
B=\id_E+\frac{\sigma}{\tau}A.
\]
Since $A^2=\sigma A+\tau\,\id_E$, we have $A^2=\tau B$. Moreover,
\[
-BA+A
=
-\left(\id_E+\frac{\sigma}{\tau}A\right)A+A
=
-\frac{\sigma}{\tau}A^2
=
-\sigma B.
\]
Therefore
\[
\pi_4(\mathcal H_{1,t'})
=
\Span_{\F_q}\{B(x_{t'}),\,(-BA+A)(x_{t'})\}
=
\Span_{\F_q}\{B(x_{t'})\},
\]
and
\[
\pi_4(\mathcal H_{3,t'})
=
\Span_{\F_q}\{B(x_{t'}),\,A^2(x_{t'})\}
=
\Span_{\F_q}\{B(x_{t'})\}.
\]
For \(\mathcal H_{2,t'}\), we get
\[
\pi_4(\mathcal H_{2,t'})
=
\Span_{\F_q}\{(B-\id_E)(x_{t'}),\,A(x_{t'})\}
=
\Span_{\F_q}\{A(x_{t'})\},
\]
because
\[
(B-\id_E)(x_{t'})=\frac{\sigma}{\tau}A(x_{t'}).
\]
Finally, one has
\[
\pi_4(\mathcal H_{4,t'})
=
\Span_{\F_q}\{B(x_{t'}),x_{t'}\}
=
E.
\]

\indent\quad Combining the four computations, we have shown that
\[
\dim_{\F_q}\pi_z(\mathcal H_{z',t'})
=
\begin{cases}
2,&z'=z,\\
1,&z'\neq z,
\end{cases}
\]
which implies that
\[
\dim_{\F_q}(W_z\cap\mathcal H_{z',t'})=
2-\dim_{\F_q}\pi_z(\mathcal H_{z',t'})=
\begin{cases}
0,&z'=z,\\
1,&z'\neq z.
\end{cases}
\]
\item Write
\[
x_{t'}=ax_t+bA(x_t),\qquad a,b\in\F_q.
\]
This expression is unique because \(x_t\) and \(A(x_t)\) form an \(\F_q\)-basis of \(E\). From
\[
A^2=\sigma A+\tau\,\id_E
\]
we obtain that
\[
A^2(x_t)=\tau x_t+\sigma A(x_t),\qquad A(x_{t'})=b\tau x_t+(a+b\sigma)A(x_t),
\]
and
\begin{align*}
  A^2(x_{t'})&=b\tau A(x_t)+(a+b\sigma)A^2(x_t)\\
  &=b\tau A(x_t)+(a+b\sigma)(\tau x_t+\sigma A(x_t))\\
  &=(a+b\sigma)\tau x_t+(a\sigma+b\sigma^2+b\tau) A(x_t).
\end{align*}

\indent\quad The condition
\[
\mathcal H_z(x_t)\cap\mathcal H_{z'}(x_{t'})\ne\{0\}
\]
is equivalent to the four displayed generators of these two subspaces being
linearly dependent in \(E\oplus E\). Using the basis
\(\{x_t,A(x_t)\}\) of \(E\), this gives the following equivalent conditions:
\[
\begin{array}{c|c}
(z,z')&\text{condition for a non-zero intersection}\\ \hline
(1,2)&a(a+b\sigma)=0\\
(1,3)&b(a+b\sigma)=0\\
(1,4)&ab=0\\
(2,3)&b(a+b\sigma)=0\\
(2,4)&ab=0\\
(3,4)&b(a\sigma-b\tau)=0
\end{array}
\]
\indent\quad For illustration, we spell out the computation for \((z,z')=(1,2)\). Use $\{x_t,A(x_t)\}$ as an \(\F_q\)-basis of \(E\). Then with respect to the induced basis of
\(E\oplus E\), we have
\[
\mathcal H_{1,t}
=
\Span_{\F_q}\{(1,0,0,0),(0,-1,0,1)\},
\]
while
\[
\mathcal H_{2,t'}
=
\Span_{\F_q}
\{(a,b,-a,-b),(0,0,b\tau,a+b\sigma)\}.
\]
Hence
\[
\mathcal H_{1,t}\cap \mathcal H_{2,t'}\neq\{0\}
\]
if and only if the matrix
\[
\begin{bmatrix}
1&0&0&0\\
0&-1&0&1\\
a&b&-a&-b\\
0&0&b\tau&a+b\sigma
\end{bmatrix}
\]
is singular. Adding \(-a\) times the first row and \(b\) times the second row
to the third row, we obtain
\[
\begin{bmatrix}
1&0&0&0\\
0&-1&0&1\\
0&0&-a&0\\
0&0&b\tau&a+b\sigma
\end{bmatrix}.
\]
Thus the matrix is singular if and only if
\[
\det
\begin{bmatrix}
-a&0\\
b\tau&a+b\sigma
\end{bmatrix}
=0\quad\Longleftrightarrow\quad a(a+b\sigma)=0.
\]
The remaining entries in the table are obtained in exactly the same way.

\indent\quad The four possible linear factors, $a,b,a+b\sigma,a\sigma-b\tau$, correspond to the following projective
relations. First,
\[
b=0
\quad\Longleftrightarrow\quad
\langle x_{t'}\rangle=\langle x_t\rangle
\quad\Longleftrightarrow\quad
\delta=0.
\]
Second,
\[
a=0
\quad\Longleftrightarrow\quad
\langle x_{t'}\rangle=\langle A(x_t)\rangle
\quad\Longleftrightarrow\quad
\delta=1.
\]
Third,
\[
a+b\sigma=0\quad\Longleftrightarrow\quad\det
\begin{bmatrix}
a&b\\
-\sigma&1
\end{bmatrix}=0,
\]
which means that \(x_{t'}\) is proportional to \(Ax_t-\sigma x_t\). Since
\[
A(A(x_t)-\sigma x_t)=A^2(x_t)-\sigma A(x_t)=\tau x_t,
\]
we have
\[
\langle A(x_t)-\sigma x_t\rangle=\langle A^{-1}x_t\rangle,
\]
and hence
\[
a+b\sigma=0
\quad\Longleftrightarrow\quad
\delta=-1.
\]
Finally,
\[
a\sigma-b\tau=0\quad\Longleftrightarrow\quad\det
\begin{bmatrix}
a&b\\
\tau&\sigma
\end{bmatrix}=0,
\]
which means that \(x_{t'}\) is proportional to
\[
\tau x_t+\sigma A(x_t)=A^2(x_t),
\]
and hence
\[
a\sigma-b\tau=0
\quad\Longleftrightarrow\quad
\delta=2.
\]

Substituting these four possibilities into the table above gives exactly the
exceptional values listed in the statement.

\item We prove that
\[
\mathcal H_{z,t}\cap\mathcal H_{z,t'}\neq\{0\}
\]
only when \(\langle x_t\rangle=\langle x_{t'}\rangle\), which is equivalent to \(t=t'\).

\indent\quad For \(z=4\), this is immediate from
\[
\mathcal H_{4,t}=\Span_{\F_q}\{(x_t,0),(0,x_t)\}.
\]
Indeed, if \(x_t\) and \(x_{t'}\) are linearly independent in \(E\), then
\[
\Span_{\F_q}\left\{(x_t,0),(0,x_t)\right\}\cap\Span_{\F_q}\left\{(x_{t'},0),(0,x_{t'})\right\}=\{0\}.
\]
The same argument applies to \(z=3\), since \(A^2\) is invertible.

\indent\quad For \(z=1\), an element of \(\mathcal H_{1,t}\) has the form
\[
(a x_t-b A(x_t),\ b A(x_t)),
\qquad a,b\in\F_q.
\]
If this also lies in \(\mathcal H_{1,t'}\), then its second component lies in
both \(\Span_{\F_q}\left\{A(x_{t})\right\}\) and
\(\Span_{\F_q}\left\{A(x_{t'})\right\}\). If $t\neq t'$, then $\langle A(x_t)\rangle\neq \langle A(x_{t'})\rangle$, so the second component is zero and thus \(b=0\). The first component then lies in both
\(\Span_{\F_q}\left\{x_t\right\}\) and \(\Span_{\F_q}\left\{x_{t'}\right\}\), and is therefore
zero. Thus the intersection is trivial if $t\neq t'$. The proof for \(z=2\) is identical. This completes the proof.
\end{enumerate}
\end{proof}

Having prepared the required subspace construction, we now treat the endpoint case $n=N_{\BW}(q)$.

\begin{lemma}\label{lem:endpoint-short}
  Let \(q\) be a prime power. Write $q=3r+\varepsilon$ with $\varepsilon\in\{0,1,2\}$. For \(i\in[4]\), set
  \[
  c_i=
  \begin{cases}
  r+1,&1\le i\le 2\varepsilon,\\
  r,&2\varepsilon<i\le 4.
  \end{cases}
  \]
  Then there exists an \((N_{\BW}(q),N_{\BW}(q)-2,2)\) MDS array code
  \(\mathcal C\) over \(\F_q\) whose node subspaces form a disjoint union
  \[
  G_1\sqcup G_2\sqcup G_3\sqcup G_4,
  \qquad |G_i|=c_i,
  \]
  such that, for every \(z,z'\in[4]\) and every \(\mathcal H\in G_{z'}\), one has
\[
\dim_{\F_q}(W_z\cap \mathcal H)
=
\begin{cases}
0,&z'=z,\\
1,&z'\neq z.
\end{cases}
\]
  Consequently, this construction gives
  \[
  \beta(\mathcal C)
  \le
  N_{\BW}(q)-2+\max_{1\le z\le4}c_z
  =
  \left\lceil \frac{5N_{\BW}(q)-8}{4}\right\rceil.
  \]
  \end{lemma}
  
  \begin{proof}
  Define
  \[
  a_0:=r+\min\{\varepsilon,1\},
  \qquad
  a_1:=r+\max\{\varepsilon-1,0\}.
  \]
  Then
  \[
  (a_0,a_0,a_1,a_1)
  =
  \begin{cases}
  (r,r,r,r),&\varepsilon=0,\\
  (r+1,r+1,r,r),&\varepsilon=1,\\
  (r+1,r+1,r+1,r+1),&\varepsilon=2.
  \end{cases}
  \]
  Thus
  \[
  (a_0,a_0,a_1,a_1)=(c_1,c_2,c_3,c_4).
  \]
  
  Choose subsets of \(\mathbb Z/(q+1)\mathbb Z\) by
  \[
  S_1=S_2=\{3a:0\le a<a_0\},
  \]
  and
  \[
  S_3=\{3a+1:0\le a<a_1\},
  \qquad
  S_4=\{3a+2:0\le a<a_1\}.
  \]
  For \(z\in[4]\), set
  \[
  G_z:=\{\mathcal H_{z,t}:t\in S_z\}.
  \]
  Then
  \[
  |G_1|=|G_2|=a_0,\qquad |G_3|=|G_4|=a_1,
  \]
  and hence
  \[
  (|G_1|,|G_2|,|G_3|,|G_4|)
  =
  (c_1,c_2,c_3,c_4).
  \]
  Moreover,
  \[
  |G_1|+|G_2|+|G_3|+|G_4|
  =
  2a_0+2a_1
  =
  4r+2\varepsilon
  =
  N_{\BW}(q).
  \]
  
  We first verify the MDS property. Since each \(\mathcal H_{z,t}\) has
  dimension \(2\) in the four-dimensional space \(\mathbb V\), it is enough to
  show that any two distinct selected subspaces intersect trivially.
  
  For two distinct subspaces from the same class \(G_z\), say
\(\mathcal H_{z,t}\) and \(\mathcal H_{z,t'}\), one has \(t\neq t'\).
Hence Lemma~\ref{lem:orbit-template-subspace}(c) gives
\[
\mathcal H_{z,t}\cap\mathcal H_{z,t'}=\{0\}.
\]
Now consider two different classes \(G_z\) and \(G_{z'}\), with \(z<z'\).
It remains to check that the corresponding differences of indices avoid the
exceptional values in Lemma~\ref{lem:orbit-template-subspace}(b). The relevant
differences have the following forms:
\[
\begin{array}{c|c|c}
\text{pair of sets} & \text{form of }t'-t & \text{exceptional values}\\ \hline
(S_1,S_2) & 3u,\quad -(a_0-1)\le u\le a_0-1 & \{\pm1\}\\
(S_1,S_3),(S_2,S_3) & 3u+1,\quad -(a_0-1)\le u\le a_1-1 & \{0,-1\}\\
(S_1,S_4),(S_2,S_4) & 3u+2,\quad -(a_0-1)\le u\le a_1-1 & \{0,+1\}\\
(S_3,S_4) & 3u+1,\quad -(a_1-1)\le u\le a_1-1 & \{0,+2\}
\end{array}
\]
We verify the first row; the other rows are checked in exactly the same way.

For a pair from \(S_1\) and \(S_2\), the difference has the form
  \[
  3u,\qquad -(a_0-1)\le u\le a_0-1.
  \]
  We claim that
  \[
  3u\not\equiv \pm1\pmod {q+1}.
  \]
  Indeed, if \(\varepsilon=2\), then \(q+1=3r+3\) is divisible by \(3\), so the
  congruence is impossible. If \(\varepsilon=0\), then \(q+1=3r+1\) and
  \[
  -(r-1)\le u\le r-1.
  \]
  The congruence \(3u\equiv1\pmod {q+1}\) would force \(u\ge r+\frac{2}{3}\) or $u\le -r$, while
  \(3u\equiv-1\pmod {q+1}\) would force \(u\ge r\) or $u\le -r-\frac{2}{3}$, both outside the range. If
  \(\varepsilon=1\), then \(q+1=3r+2\) and
  \[
  -r\le u\le r.
  \]
  The congruence \(3u\equiv1\pmod {q+1}\) would force \(u\ge r+1\) or $u\le -r-\frac{1}{3}$, while
  \(3u\equiv-1\pmod {q+1}\) would force \(u\ge r+\frac{1}{3}\) or $u\le -r-1$, again outside the range.
  Thus the exceptional values \(\{\pm1\}\) are avoided.

  Thus every selected pair from distinct classes avoids the exceptional values
  in Lemma~\ref{lem:orbit-template-subspace}(b), and hence has trivial
  intersection. Therefore any two selected node subspaces have direct sum
  \(\mathbb V\). By the subspace reformulation, the selected family determines
  an \((N_{\BW}(q),N_{\BW}(q)-2,2)\) MDS array code over \(\F_q\).
  
  It remains to verify the repair property and compute the bandwidth. Let
  \(\mathcal H\in G_z\). By Lemma~\ref{lem:orbit-template-subspace}(a),
  \[
  W_z\cap\mathcal H=\{0\},
  \]
  so \(W_z\) is feasible for repairing \(\mathcal H\). Moreover, for every
  selected subspace \(\mathcal H'\neq \mathcal H\),
  \[
  \dim_{\F_q}(W_z\cap \mathcal H')
  =
  \begin{cases}
  0,&\mathcal H'\in G_z,\\
  1,&\mathcal H'\notin G_z.
  \end{cases}
  \]
  Since here \(\ell=2\), the repair bandwidth corresponding to \(W_z\) is
  \[
  2\bigl(N_{\BW}(q)-1\bigr)
  -
  \sum_{\mathcal H'\neq\mathcal H}
  \dim_{\F_q}(W_z\cap\mathcal H').
  \]
  For \(\mathcal H\in G_z\), the sum of intersection dimensions equals
  \[
  N_{\BW}(q)-|G_z|,
  \]
  because the subspaces outside \(G_z\) contribute \(1\), while the other
  subspaces in \(G_z\) contribute \(0\). Hence the displayed repair has bandwidth
  \[
  2\bigl(N_{\BW}(q)-1\bigr)-\bigl(N_{\BW}(q)-|G_z|\bigr)
  =
  N_{\BW}(q)+|G_z|-2.
  \]
  Therefore this construction gives
  \[
  \beta(\mathcal C)
  \le
  N_{\BW}(q)-2+\max_{1\le z\le4}|G_z|
  =
  N_{\BW}(q)-2+\max_{1\le z\le4}c_z.
  \]
  
  Finally, since
  \[
  N_{\BW}(q)=4r+2\varepsilon
  \]
  and
  \[
  \max_{1\le z\le4}c_z=
  \begin{cases}
  r,&\varepsilon=0,\\
  r+1,&\varepsilon\in\{1,2\},
  \end{cases}
  \]
  we obtain
  \[
  N_{\BW}(q)-2+\max_z c_z
  =
  \begin{cases}
  5r-2,&\varepsilon=0,\\
  5r+1,&\varepsilon=1,\\
  5r+3,&\varepsilon=2.
  \end{cases}
  \]
  On the other hand,
  \[
  \left\lceil\frac{5N_{\BW}(q)-8}{4}\right\rceil
  =
  \left\lceil
  \frac{5(4r+2\varepsilon)-8}{4}
  \right\rceil
  =
  \begin{cases}
  5r-2,&\varepsilon=0,\\
  5r+1,&\varepsilon=1,\\
  5r+3,&\varepsilon=2.
  \end{cases}
  \]
  Thus
  \[
  N_{\BW}(q)-2+\max_z c_z
  =
  \left\lceil\frac{5N_{\BW}(q)-8}{4}\right\rceil.
  \]
  This proves the lemma.
  \end{proof}

  \begin{proposition}\label{prop:short-bw}
    Let \(q\) be a prime power. For every integer \(n\) with
    \[
    3\le n\le N_{\BW}(q),
    \]
    there exists an \((n,n-2,2)\) MDS array code \(\mathcal C\) over \(\F_q\) such that
    \[
    \beta(\mathcal C)
    \le
    \left\lceil\frac{5n-8}{4}\right\rceil.
    \]
    \end{proposition}
    
    \begin{proof}
    Write $q=3r+\varepsilon$ with $\varepsilon\in\{0,1,2\}$ and define \(c_1,c_2,c_3,c_4\) as in Lemma~\ref{lem:endpoint-short}. By
    Lemma~\ref{lem:endpoint-short}, there exists an endpoint code of length
    \(N_{\BW}(q)\) whose node subspaces form a disjoint union
    \[
    G_1\sqcup G_2\sqcup G_3\sqcup G_4,
    \qquad |G_i|=c_i,
    \]
    and such that, for every \(z,z'\in[4]\) and every \(\mathcal H\in G_{z'}\),
    \[
    \dim_{\F_q}(W_z\cap\mathcal H)
    =
    \begin{cases}
    0,&z'=z,\\
    1,&z'\neq z.
    \end{cases}
    \]
    
    We choose integers
    \[
    0\le g_i\le c_i\qquad(i\in[4])
    \]
    such that
    \[
    g_1+g_2+g_3+g_4=n
    \]
    and
    \[
    \max_{1\le i\le4}g_i=\left\lceil\frac n4\right\rceil.
    \]
    Such a choice is always possible. Indeed, write
    \[
    n=4a+b,\qquad a\in\mathbb{N},\,b\in\{0,1,2,3\}.
    \]
    \begin{itemize}
      \item If \(a<r\), take \(b\) of the \(g_i\)'s equal to \(a+1\), and the remaining
      \(4-b\) of them equal to \(a\). Then all \(g_i\le r\le c_i\).
      \item If \(a=r\), then
      \[
      n=4r+b\le N_{\BW}(q)=4r+2\varepsilon,
      \]
      so $b\le 2\varepsilon$. In this case, take \(b\) of the enlarged classes, namely among those indices
      \(i\) with \(c_i=r+1\), to have size \(r+1\), and take all remaining $4-b$ classes
      to have size \(r\).
      \item Finally, if \(a=r+1\), then we must have \(n=4r+4\) and $\epsilon=2$. Indeed, in this case, we have
      \[
      4r+4+b=n\le N_{\BW}(q)=4r+2\varepsilon\le 4r+4,
      \]
      which can happen only when $n=4r+4$ and \(\varepsilon=2\). It follows that
      \[
      c_1=c_2=c_3=c_4=r+1,
      \]
      and we take
      \[
      g_1=g_2=g_3=g_4=r+1.
      \]
    \end{itemize}
    Thus the desired balanced choice of \(g_i\)'s exists in all cases.

    Now choose subsets
    \[
    G_i'\subseteq G_i,
    \qquad |G_i'|=g_i,
    \]
    and keep exactly the node subspaces in
    \[
    G_1'\sqcup G_2'\sqcup G_3'\sqcup G_4'.
    \]
    Since the endpoint code is MDS, any two node subspaces in the endpoint family
    have direct sum \(\mathbb V\). Therefore any two retained node subspaces still
    have direct sum \(\mathbb V\). By the subspace reformulation, the retained
    family determines an \((n,n-2,2)\) MDS array code over \(\F_q\), which we denote
    by \(\mathcal C\).
    
    It remains to estimate the repair bandwidth. Let $\mathcal H\in G_z'$. By the intersection pattern above, \(W_z\cap\mathcal H=\{0\}\), so \(W_z\) is
    feasible for repairing \(\mathcal H\). Moreover, among the retained node
    subspaces other than \(\mathcal H\), exactly \(n-g_z\) lie outside \(G_z'\), and
    each of them contributes intersection dimension \(1\) with \(W_z\); the
    remaining \(g_z-1\) lie in \(G_z'\), and each contributes intersection dimension
    \(0\). Hence
    \[
    \sum_{\substack{\mathcal H'\neq\mathcal H\\ \mathcal H'\ \text{retained}}}
    \dim_{\F_q}(W_z\cap\mathcal H')
    =
    n-g_z.
    \]
    Since \(\ell=2\), the bandwidth of this repair is
    \[
    2(n-1)-(n-g_z)
    =
    n+g_z-2.
    \]
    Therefore
    \[
    \beta(\mathcal C)
    \le
    n-2+\max_{1\le z\le4}g_z
    =
    n-2+\left\lceil\frac n4\right\rceil.
    \]
    Finally,
    \[
    n-2+\left\lceil\frac n4\right\rceil
    =
    \left\lceil\frac{5n-8}{4}\right\rceil.
    \]
    Thus
    \[
    \beta(\mathcal C)
    \le
    \left\lceil\frac{5n-8}{4}\right\rceil,
    \]
    as claimed.
    \end{proof}

    Corollary~\ref{cor:beta-small-n} and Proposition~\ref{prop:short-bw} together give the following result:

  \begin{corollary}\label{cor:short-optimal-bw}
    Let $q$ be a prime power. Then for any integer $n$ with $3\le n\le N_{\BW}(q)$ and $n\not\in\{5,6,9,10\}$, one has 
    $$\beta_{\opt}^{(q,n,2,2)}=\left\lceil \frac{5n-8}{4}\right\rceil.$$
  \end{corollary}

  \subsection{Repair I/O}

  For any prime power \(q\), put
  \[
  N_{\IO}(q):=\left\lfloor \frac{3(q+1)}{2}\right\rfloor
  =
  \begin{cases}
  \dfrac{3(q+1)}{2},&q\text{ odd},\\[2mm]
  \dfrac{3q}{2}+1,&q\text{ even}.
  \end{cases}
  \]
  In this subsection we study the repair I/O in the range
  \[
  3\le n\le N_{\IO}(q).
  \]
  
  We use the same orbit notation as in Subsection~\ref{subsec:bw-short}. Thus $E=\F_{q^2}$ is viewed as a two-dimensional \(\F_q\)-vector space, \(A:E\to E\) is
  multiplication by a primitive element of \(E^\times\), and
  \[
  A^2=\sigma A+\tau\,\id_E
  \]
  for some \(\sigma\in\F_q^\times\) and \(\tau\in\F_q^\times\).
  We also fix \(x_0\in E\setminus\{0\}\) and write
  \[
  x_t:=A^t(x_0),
  \qquad t\in\mathbb Z/(q+1)\mathbb Z.
  \]
  
  Identify
\[
\mathbb V=\F_q^4
\]
with \(E\oplus E\). For \(x\in E\setminus\{0\}\), define three projective
column sets in \(\mathbb P(\mathbb V)\) by
\[
\mathcal X_1(x)
:=
\{\langle(A^{-1}(x),0)\rangle,\langle(-x,x)\rangle\},
\]
\[
\mathcal X_2(x)
:=
\{\langle(x,-x)\rangle,\langle(0,A^{-1}(x))\rangle\},
\]
and
\[
\mathcal X_3(x)
:=
\{\langle(x,0)\rangle,\langle(0,A(x))\rangle\}.
\]
Let
\[
\mathcal K_z(x):=\Span_{\F_q}\mathcal X_z(x)\le \mathbb V,
\qquad z\in[3],
\]
where \(\Span_{\F_q}\mathcal X_z(x)\) denotes the subspace spanned by any
non-zero representatives of the two projective points in \(\mathcal X_z(x)\).
For \(z\in[3]\) and \(t\in\mathbb Z/(q+1)\mathbb Z\), put
\[
\mathcal X_{z,t}:=\mathcal X_z(x_t),
\qquad
\mathcal K_{z,t}:=\mathcal K_z(x_t).
\] 
We use the same repair subspaces \(W_1,W_2,W_3\) as in the repair-bandwidth
construction:
\[
W_1:=\{0\}\oplus E,
\qquad
W_2:=E\oplus\{0\},
\qquad
W_3:=\{(u,-u):u\in E\}.
\]
For a coordinate whose node subspace is of the form \(\mathcal K_{z,t}\), we
use \(W_z\) as its repair subspace.

  \begin{lemma}\label{lem:io-template-subspace}
    For every \(z,z'\in[3]\) and every \(t,t'\in\mathbb Z/(q+1)\mathbb Z\), the following properties hold.
    
    \begin{enumerate}[label=(\alph*)]
    \item One has
    \[
    |\,\mathcal X_{z',t'}\cap\mathbb P(W_z)\,|
    =
    \begin{cases}
    0,&z'=z,\\
    1,&z'\neq z.
    \end{cases}
    \]
    
    \item Suppose \(z<z'\), and write
    \[
    \delta=t'-t\in\mathbb Z/(q+1)\mathbb Z.
    \]
    Then
    \[
    \mathcal K_{z,t}\cap\mathcal K_{z',t'}\neq\{0\}
    \]
    can occur only for the following values of \(\delta\):
    \[
    \begin{array}{c|c}
    (z,z')&\text{exceptional values of }\delta\\ \hline
    (1,2)&\{0\}\\
    (1,3)&\{-1\}\\
    (2,3)&\{0,-2\}
    \end{array}
    \]
    
    \item For \(z=z'\), one has
    \[
    \mathcal K_{z,t}\cap\mathcal K_{z,t'}\neq\{0\}
    \quad\Longleftrightarrow\quad
    t=t'.
    \]
    \end{enumerate}
    \end{lemma}
    
    \begin{proof}
    We prove the three assertions separately.
    
    \begin{enumerate}[label=(\alph*)]
    \item This is immediate from the definitions. For \(W_1=\{0\}\oplus E\), the
    two points in \(\mathcal X_{1,t'}\) do not lie in \(\mathbb P(W_1)\), while
    exactly one point of each of \(\mathcal X_{2,t'}\) and \(\mathcal X_{3,t'}\)
    lies in \(\mathbb P(W_1)\). Hence
    \[
    |\,\mathcal X_{z',t'}\cap\mathbb P(W_1)\,|
    =
    \begin{cases}
    0,&z'=1,\\
    1,&z'\neq1.
    \end{cases}
    \]
    The same verification for $W_2=E\oplus\{0\}$ gives
    \[
    |\,\mathcal X_{z',t'}\cap\mathbb P(W_2)\,|
    =
    \begin{cases}
    0,&z'=2,\\
    1,&z'\neq2.
    \end{cases}
    \]
    Finally, for $W_3=\{(u,-u):u\in E\}$, the point \(\langle(-x_{t'},x_{t'})\rangle\) from \(\mathcal X_{1,t'}\) and the point
    \(\langle(x_{t'},-x_{t'})\rangle\) from \(\mathcal X_{2,t'}\) lie in
    \(\mathbb P(W_3)\), while neither point of \(\mathcal X_{3,t'}\) does. Thus
    \[
    |\,\mathcal X_{z',t'}\cap\mathbb P(W_3)\,|
    =
    \begin{cases}
    0,&z'=3,\\
    1,&z'\neq3.
    \end{cases}
    \]
    This proves (a).
    
    \item Write
    \[
    x_{t'}=ax_t+bA(x_t),\qquad a,b\in\F_q.
    \]
    This expression is unique because \(x_t\) and \(A(x_t)\) form an \(\F_q\)-basis of
    \(E\). Using
    \[
    A^2=\sigma A+\tau\,\id_E
    \]
    and thus
    \[
    A^{-1}=\frac1{\tau}(A-\sigma\,\id_E),
    \]
    one checks that the condition
    \[
    \mathcal K_{z,t}\cap\mathcal K_{z',t'}\neq\{0\}
    \]
    is equivalent to the following conditions:
    \[
    \begin{array}{c|c}
    (z,z')&\text{condition for a non-zero intersection}\\ \hline
    (1,2)&b=0\\
    (1,3)&a+\sigma b=0\\
    (2,3)&b\bigl(\sigma a+(\sigma^2+\tau)b\bigr)=0
    \end{array}
    \]
    
    \indent\quad The three possible linear factors, $b,a+\sigma b,\sigma a+(\sigma^2+\tau)b$, correspond to the following projective relations. We have shown in the proof of Lemma~\ref{lem:orbit-template-subspace} that
    \[
    b=0
    \quad\Longleftrightarrow\quad
    \langle x_{t'}\rangle=\langle x_t\rangle
    \quad\Longleftrightarrow\quad
    \delta=0,
    \]
    and 
    \[
      a+\sigma b=0\quad\Longleftrightarrow\quad\langle A(x_t)-\sigma x_t\rangle=\langle A^{-1}x_t\rangle\quad\Longleftrightarrow\quad\delta=1.
    \]
    Finally,
    \[
    \sigma a+(\sigma^2+\tau)b=0\quad\Longleftrightarrow\quad\det\begin{bmatrix}
    a&b\\
    \sigma^2+\tau&-\sigma
    \end{bmatrix}=0,
    \]
    which means that \(x_{t'}\) is proportional to
    \[
    (\sigma^2+\tau)x_t-\sigma A(x_t).
    \]
    This vector is proportional to \(A^{-2}(x_t)\). Indeed,
    \begin{align*}
      &A^2\bigl((\sigma^2+\tau)\id_E-\sigma A\bigr)\\
      =\,&(\sigma A+\tau\,\id_E)\bigl((\sigma^2+\tau)\id_E-\sigma A\bigr)\\
      =\,&\sigma(\sigma^2+\tau)A+\tau(\sigma^2+\tau)\id_E-\sigma^2A^2-\sigma\tau A\\
      =\,&\sigma^3A+\tau(\sigma^2+\tau)\id_E-\sigma^2(\sigma A+\tau\,\id_E)\\
      =\,&-\tau^2\id_E,
    \end{align*}
    which implies that
    \[
    (\sigma^2+\tau)(x_t)-\sigma A(x_t)=-\tau^2A^{-2}(x_t).
    \]
    Hence
    \[
    \sigma a+(\sigma^2+\tau)b=0
    \quad\Longleftrightarrow\quad
    \delta=-2.
    \]
    Substituting these possibilities into the table gives the exceptional values
    listed in (b).
    
    \item If \(z=z'\), the same argument as in Lemma~\ref{lem:orbit-template-subspace}
    shows that
    \[
    \mathcal K_{z,t}\cap\mathcal K_{z,t'}\neq\{0\}
    \]
    only when $\langle x_t\rangle=\langle x_{t'}\rangle$. Since the projective points
    \[
    \langle x_0\rangle,\langle x_1\rangle,\dots,\langle x_q\rangle
    \]
    are pairwise distinct, this is equivalent to \(t=t'\). The converse is
    immediate.
    \end{enumerate}
    \end{proof}

  \begin{proposition}\label{prop:short-io}
    Let \(q\) be a prime power. For every integer \(n\) with
    \[
    3\le n\le N_{\IO}(q),
    \]
    there exists an \((n,n-2,2)\) MDS array code \(\mathcal C\) over \(\F_q\) such
    that
    \[
    \gamma(\mathcal C)
    \le
    \left\lceil\frac{4n-6}{3}\right\rceil.
    \]
    \end{proposition}
    
    \begin{proof}
    We first construct endpoint families.
    
    \medskip
    
    \noindent
    \textbf{Case 1: \(q\) odd.}
    Write $q+1=2m$. Choose subsets of \(\mathbb Z/(q+1)\mathbb Z\) by
    \[
    S_1=S_3=\{1,3,5,\dots,2m-1\},
    \qquad
    S_2=\{0,2,4,\dots,2m-2\}.
    \]
    For \(z\in[3]\), set
    \[
    G_z:=\{\mathcal K_{z,t}:t\in S_z\}.
    \]
    Then
    \[
    |G_1|=|G_2|=|G_3|=m,
    \]
    which implies that 
    \[
    \sum\limits_{z=1}^3 |G_z|=3m=\frac{3(q+1)}{2}=N_{\IO}(q).
    \]
    
    We verify the MDS property using Lemma~\ref{lem:io-template-subspace}. For
    \(G_1\) and \(G_2\), only shift \(0\) is exceptional, and $S_1\cap S_2=\varnothing$. For \(G_1\) and \(G_3\), only shift \(-1\) is exceptional. Since
    \[
    S_3+1=\{0,2,4,\dots,2m-2\}=S_2,
    \]
    we have $S_1\cap(S_3+1)=\varnothing$. For \(G_2\) and \(G_3\), the exceptional shifts are \(0\) and \(-2\). We have $S_2\cap S_3=\varnothing$, and since $S_3+2=S_3$, we also have $S_2\cap(S_3+2)=\varnothing$. Thus every two selected node subspaces have trivial intersection, and hence
    the endpoint construction gives an \((3m,3m-2,2)\) MDS array code.
  
    \medskip
    
    \noindent
    \textbf{Case 2: \(q\) even.}
    Write $q+1=2m+1$. Choose subsets of \(\mathbb Z/(q+1)\mathbb Z\) by
    \[
    S_3=\{2,4,\dots,2m\},
    \]
    \[
    S_2=\{0,3,5,\dots,2m-1\},
    \]
    and
    \[
    S_1=\mathbb Z/(2m+1)\mathbb Z\setminus S_2
    =
    \{1,2,4,\dots,2m\}.
    \]
    For \(z\in[3]\), set
    \[
    G_z:=\{\mathcal K_{z,t}:t\in S_z\}.
    \]
    Then
    \[
    |G_1|=m+1,\qquad |G_2|=|G_3|=m,
    \]
    which implies that
    \[
    \sum\limits_{z=1}^3 |G_z|=m+1+m+m=3m+1=\frac{3q}{2}+1=N_{\IO}(q).
    \]
    
    The MDS verification is the same. For \(G_1\) and \(G_2\), shift \(0\) is
    avoided because $S_1\cap S_2=\varnothing$. For \(G_1\) and \(G_3\), shift \(-1\) is avoided because $S_3+1=\{0,3,5,\dots,2m-1\}=S_2$ and hence $S_1\cap(S_3+1)=\varnothing$. For \(G_2\) and \(G_3\), shifts \(0\) and \(-2\) are avoided because $ S_2\cap S_3=\varnothing$ and $S_3+2=\{1,4,6,\dots,2m\}$ is disjoint from \(S_2\). 
    
    \medskip 
    
    It remains to pass from the endpoint length to all shorter lengths. Let $3\le n\le N_{\IO}(q)$. Retain a subfamily of the endpoint node subspaces, with class sizes $g_1,g_2,g_3$. The MDS property is preserved, because every pair of retained node subspaces
    already had trivial intersection in the endpoint construction. By the subspace reformulation, the retained subfamily determines an \((n,n-2,2)\) MDS array code over \(\F_q\), which we denote
    by \(\mathcal C\). Moreover, the
    projective column sets and repair subspaces still satisfy the intersection
    pattern in Lemma~\ref{lem:io-template-subspace}(a). Hence the displayed repair
    scheme has maximum repair I/O
    \[
    n-2+\max\{g_1,g_2,g_3\}.
    \]
    
    When \(q\) is odd, the endpoint capacities are \((m,m,m)\). Write $n=3a+b$ with $a\in\mathbb{N}$ and $b\in\{0,1,2\}$. Choose \(b\) of the \(g_i\)'s to be \(a+1\), and the remaining \(3-b\) of them
    to be \(a\). Then
    \[
    g_1+g_2+g_3=n,
    \qquad
    \max_i g_i=\left\lceil\frac n3\right\rceil.
    \]
    Moreover, \(n\le3m\) implies
    \[
    \left\lceil\frac n3\right\rceil\le m,
    \]
    so
    \[
    0\le g_i\le m
    \qquad(i\in[3]).
    \]
    Thus this choice is compatible with the endpoint capacities.
    
    When \(q\) is even, the endpoint capacities are \((m+1,m,m)\). If
    \(n\le3m\), the same balanced choice as above gives
    \[
    g_1+g_2+g_3=n,
    \qquad
    0\le g_i\le m,
    \qquad
    \max_i g_i=\left\lceil\frac n3\right\rceil.
    \]
    If \(n=3m+1\), we take
    \[
    (g_1,g_2,g_3)=(m+1,m,m).
    \]
    This again respects the endpoint capacities, and in this case
    \[
    \max_i g_i=m+1=\left\lceil\frac{3m+1}{3}\right\rceil
    =
    \left\lceil\frac n3\right\rceil.
    \]
    
    Therefore, in both parity cases, we may retain \(g_i\) node subspaces from the
    \(i\)-th class so that
    \[
    g_1+g_2+g_3=n,
    \qquad
    \max_i g_i=\left\lceil\frac n3\right\rceil.
    \]
   It follows that the constructed code $\mathcal{C}$ satisfies
\[
\gamma(\mathcal C)
\le
n-2+\left\lceil\frac n3\right\rceil.
\]
Finally,
\[
n-2+\left\lceil\frac n3\right\rceil
=
\left\lceil\frac{4n-6}{3}\right\rceil.
\]
This proves the proposition.
    \end{proof}

    Corollary~\ref{cor:gamma-small-n} and Proposition~\ref{prop:short-io} together give the following result:

  \begin{corollary}\label{cor:short-optimal-io}
    Let $q$ be a prime power. Then for any integer $n$ with $3\le n\le N_{\IO}(q)$ and $n\neq 4$, one has 
    $$\gamma_{\opt}^{(q,n,2,2)}=\left\lceil\frac{4n-6}{3}\right\rceil.$$
  \end{corollary}

\section{The Long-Length Regime}\label{sec:long}

We begin by recalling the incidence-multiplicity bound established in~\cite{wu2026incidence}.

\begin{theorem}[{\cite[Theorem 3.3]{wu2026incidence}}]\label{thm:incidence-multiplicity-bound}
  Let $\mathcal{C}$ be an $(n,k,\ell)$ MDS array code over $\mathbb{F}_q$ with redundancy $r=n-k\ge 2$. Then for any $i\in[n]$,
  $$\gamma_i(\mathcal{C})\ge \beta_i(\mathcal{C})\ge\ell(n-1)-(r-1)\frac{q^{\ell}-1}{q-1}.$$
  Hence 
  $$\gamma(\mathcal{C})\ge\beta(\mathcal{C})\ge\ell(n-1)-(r-1)\frac{q^{\ell}-1}{q-1}.$$
\end{theorem}

For $r=\ell=2$, Theorem~\ref{thm:incidence-multiplicity-bound} gives the lower bound $2n-q-3$. 

In~\cite{wu2026incidence}, Wu also presents constructions that attain the
incidence-multiplicity bound for a wide range of code lengths~$n$ for arbitrary $r$ and $\ell$.

\begin{theorem}[{\cite[Theorem 4.4]{wu2026incidence}}]\label{thm:incidence-multiplicity-bound-attainable}
  Assume that $\ell\ge 2$ and $r\ge 2$, and suppose that 
  $$(r-1)\mid(q-1)\qquad\text{and}\qquad\frac{q-1}{r-1}\ge 2.$$
  Then for every integer $n$ satisfying 
  $$2(r-1)\frac{q^{\ell}-1}{q-1}\le n\le q^{\ell}+1,$$
  there exists an $(n,n-r,\ell)$ MDS array code $\mathcal{C}$ over $\mathbb{F}_q$ such that
  $$\beta(\mathcal{C})=\gamma(\mathcal{C})=\ell(n-1)-(r-1)\frac{q^{\ell}-1}{q-1}.$$
\end{theorem}

For $r=\ell=2$, Theorem~\ref{thm:incidence-multiplicity-bound-attainable} shows that for every prime
power $q\ge 3$ and every integer $n$ with $2q+2\le n\le q^2+1$, there exists
an $(n,n-2,2)$ MDS array code $\mathcal{C}$ over $\mathbb{F}_q$ such that
\[
  \beta(\mathcal{C})=\gamma(\mathcal{C})=2n-q-3.
\]

In this section, we further expand the range of code lengths~$n$ for which the
incidence-multiplicity bound is attainable in the case $r=\ell=2$.

\subsection{Repair Bandwidth}

For any prime power \(q\), put
      \[
      \widetilde{N}_{\BW}(q)=\left\lceil\frac{4(q+1)}{3}\right\rceil=
      \begin{cases}
      \dfrac{4q}{3}+2,&q\equiv0\pmod3,\\[2mm]
      \dfrac{4q+2}{3}+1,&q\equiv1\pmod3,\\[2mm]
      \dfrac{4(q+1)}{3},&q\equiv2\pmod3.
      \end{cases}
      \]
      In this subsection we study the repair bandwidth in the range
      \[
        \widetilde{N}_{\BW}(q)\le n\le q^2+1.
      \]

\begin{lemma}\label{lem:cyclic-T}
  Let \(q\) be a prime power. Then there exists a matrix \(T\in\GL_2(\F_q)\)
  such that
  \begin{enumerate}[label=(\roman*)]
    \item the projective class \([T]\in\PGL_2(\F_q)\) has order \(q+1\);
    \item \([T-I]=[T]^2\);
    \item for every non-zero vector \(p_0\in\F_q^2\), the projective points
    \[
    \langle p_0\rangle,\langle Tp_0\rangle,\langle T^2p_0\rangle,\dots,\langle T^qp_0\rangle
    \]
    are pairwise distinct;
    \item \(T\) has irreducible characteristic polynomial over \(\F_q\). Hence the
\(\F_q\)-subalgebra of \(\F_q^{2\times2}\) generated by \(T\), i.e.,
\[
\F_q[T]:=\{f(T):f\in\F_q[x]\}\subseteq \F_q^{2\times2},
\]
is a field isomorphic to \(\F_{q^2}\).
  \end{enumerate}
  \end{lemma}
  
  \begin{proof}
  Choose an element \(\eta\in\F_{q^2}^{\times}\) of order \(q+1\), and set $\theta:=\eta+1$. Since \(\eta\notin\F_q\), we have \(\theta\notin\F_q\). Define
  \[
  c:=\frac{\theta^2}{\theta-1}
  =
  \frac{\theta^2}{\eta}
  \in\F_{q^2}^{\times}.
  \]
  We claim that \(c\in\F_q^\times\). Since \(\eta^q=\eta^{-1}\), one has
  \[
  \theta^q
  =
  \eta^q+1
  =
  \eta^{-1}+1
  =
  \frac{\eta+1}{\eta}
  =
  \frac{\theta}{\eta}.
  \]
  Therefore
  \[
  c^q
  =
  \left(\frac{\theta^2}{\eta}\right)^q
  =
  \frac{\theta^{2q}}{\eta^q}
  =
  \frac{(\theta/\eta)^2}{\eta^{-1}}
  =
  \frac{\theta^2}{\eta}
  =
  c.
  \]
  Thus \(c\in\F_q^\times\).
  
  Now $c(\theta-1)=\theta^2$, so $\theta^2-c\theta+c=0$. Let $T:\F_{q^2}\to\F_{q^2}$ be the \(\F_q\)-linear map given by multiplication by \(\theta\). After fixing
  an \(\F_q\)-linear identification \(\F_{q^2}\cong\F_q^2\), we regard this map
  as a matrix in \(\GL_2(\F_q)\). Then $T^2-cT+cI=0$. Equivalently, $T-I=c^{-1}T^2$. Passing to projective classes gives $[T-I]=[T]^2$.
  
  We next show that \([T]\) has order \(q+1\) in $\PGL_2(\F_q)$. Since
  \[
  \theta^{q+1}
  =
  \theta\theta^q
  =
  \theta\cdot\frac{\theta}{\eta}
  =
  \frac{\theta^2}{\eta}
  =
  c\in\F_q^\times,
  \]
  we have $[T]^{q+1}=1$ in $\PGL_2(\F_q)$. Conversely, if \([T]^m=1\), then \(\theta^m\in\F_q^\times\). Hence
  \[
  1
  =
  \frac{\theta^m}{(\theta^m)^q}
  =
  \left(\frac{\theta}{\theta^q}\right)^m
  =
  \left(\frac{\theta}{\theta/\eta}\right)^m
  =
  \eta^m.
  \]
  Since \(\eta\) has order \(q+1\), this implies
  \[
  (q+1)\mid m.
  \]
  Therefore \([T]\) has order exactly \(q+1\).
  
  We now prove the orbit statement. Let \( p_0\in\F_q^2\setminus\{0\}\), viewed as a
  non-zero element of \(\F_{q^2}\) under the fixed identification. Suppose that $\langle T^ip_0\rangle=\langle T^jp_0\rangle$ for some \(0\le i<j\le q\). Since \(T\) is multiplication by \(\theta\), this
  means that there exists \(\lambda\in\F_q^\times\) such that $\theta^j p_0=\lambda\theta^i p_0$. Since \(p_0\neq0\), we get $\theta^{j-i}=\lambda\in\F_q^\times$. Therefore $[T]^{j-i}=1$. But \(0<j-i<q+1\), contradicting the fact that \([T]\) has order \(q+1\).
  Hence the projective points
  \[
    \langle p_0\rangle,\langle Tp_0\rangle,\langle T^2p_0\rangle,\dots,\langle T^qp_0\rangle
  \]
  are pairwise distinct.
  
  Finally, because \(\theta\notin\F_q\), the characteristic polynomial of \(T\)
  is irreducible over \(\F_q\). Thus \(\F_q[T]\) is a field isomorphic to
  \(\F_{q^2}\).
  \end{proof}
  \begin{remark}\label{rem:FqT-projective-classes}
    By Lemma~\ref{lem:cyclic-T}(iv), the algebra \(\F_q[T]\) is a field
    isomorphic to \(\F_{q^2}\). Hence $\F_q[T]^\times/\F_q^\times$ has order
    \[
    \frac{q^2-1}{q-1}=q+1.
    \]
    This quotient is precisely the set of projective classes of non-zero elements
    of \(\F_q[T]\). Since \([T]\) has order \(q+1\) by
    Lemma~\ref{lem:cyclic-T}(i), it generates this quotient. Therefore, for every
    \(B\in\F_q[T]^\times\), there exists a unique
    \(i\in\mathbb Z/(q+1)\mathbb Z\) such that $[B]=[T]^i=\tau^i$.
    \end{remark}

We now fix a matrix $T$ satisfying the conditions of Lemma~\ref{lem:cyclic-T} and put $\tau=[T]$. Fix a projective point \(p_0\in\mathbb P^1(\F_q)\). By Lemma~\ref{lem:cyclic-T}(iii), the projective points
\[
p_0,\tau(p_0),\dots,\tau^q(p_0)
\]
are pairwise distinct and hence contain all points of \(\mathbb P^1(\F_q)\). We label these points by
\[
t\in\mathbb Z/(q+1)\mathbb Z,
\qquad
t\leftrightarrow \tau^t(p_0).
\]
With this labeling, the action of \(\tau\) is simply
\[
\tau(t)=t+1.
\]

We work in
\[
\mathbb V:=\F_q^2\oplus \F_q^2.
\]
For \(x,y\in\mathbb P^1(\F_q)\), choose non-zero representatives
\(p_x,p_y\in\F_q^2\), and define the following two-dimensional $\mathbb{F}_q$-subspaces
of \(\mathbb V\):
\[
L_{1}(x,y):=\Span_{\F_q}\{(p_x,0),(0,p_y)\},
\]
\[
L_{2}(x,y):=\Span_{\F_q}\{(p_x,0),(p_y,p_y)\},
\]
and
\[
L_{3}(x,y):=\Span_{\F_q}\{(0,p_x),(p_y,p_y)\},
\]
which are the template node subspaces. We also define four repair subspaces
\[
W_1:=\F_q^2\oplus0,\qquad
W_2:=0\oplus\F_q^2,\qquad
W_3:=\Gamma(I),\qquad
W_4:=\Gamma(T),
\]
where
\[
\Gamma(M):=\{(Mv,v):v\in\F_q^2\}.
\]
For brevity, we say that a repair subspace \(W\) \emph{hits} a node subspace \(L\) if $W\cap L\neq\{0\}$; otherwise, we say that \(W\) \emph{misses} \(L\).

\begin{lemma}\label{lem:graph-line-incidence}
  Let \(A\in\GL_2(\F_q)\) be such that \(A-I\in\GL_2(\F_q)\). For
  \(x,y\in\mathbb P^1(\F_q)\), the following equivalences hold:
  \[
  \Gamma(A)\cap L_{1}(x,y)\neq\{0\}
  \quad\Longleftrightarrow\quad
  [A](y)=x,
  \]
  \[
  \Gamma(A)\cap L_{2}(x,y)\neq\{0\}
  \quad\Longleftrightarrow\quad
  [A-I](y)=x,
  \]
  and
  \[
  \Gamma(A)\cap L_{3}(x,y)\neq\{0\}
  \quad\Longleftrightarrow\quad
  [A^{-1}-I](y)=x.
  \]
  \end{lemma}
  
  \begin{proof}
  We first consider \(L_1\). By definition, $\Gamma(A)\cap L_1(x,y)\neq\{0\}$ if and only if there exist \(v\in\F_q^2\setminus\{0\}\) and \(a,b\in\F_q\) such that
  \[
  (Av,v)=a(p_x,0)+b(0,p_y)=(ap_x,bp_y).
  \]
  This is equivalent to $Av=ap_x$ and $v=bp_y$. Since \(v\neq0\), we have \(b\neq0\), and hence $v\in\langle p_y\rangle$. Since \(A\) is invertible, \(Av\neq0\), so \(a\neq0\), and therefore $Av\in\langle p_x\rangle$. Thus the above condition is equivalent to $[A](y)=x$. This proves the first equivalence.
  
  Next consider \(L_2\). By definition, $\Gamma(A)\cap L_2(x,y)\neq\{0\}$ if and only if there exist \(v\in\F_q^2\setminus\{0\}\) and \(a,b\in\F_q\) such that
  \[
  (Av,v)=a(p_x,0)+b(p_y,p_y)=(ap_x+bp_y,bp_y).
  \]
  This is equivalent to $Av=ap_x+bp_y$ and $v=bp_y$. Since \(v\neq0\), we have \(b\neq0\), so $v\in\langle p_y\rangle$. Subtracting the second equation from the first gives $(A-I)v=ap_x$. Because \(A-I\) is invertible and \(v\neq0\), we have \((A-I)v\neq0\), hence
  \(a\neq0\). Therefore $(A-I)v\in\langle p_x\rangle$. Thus the above condition is equivalent to $ [A-I](y)=x$. This proves the second equivalence.
  
  Finally consider \(L_3\). By definition, $\Gamma(A)\cap L_3(x,y)\neq\{0\}$ if and only if there exist \(v\in\F_q^2\setminus\{0\}\) and \(a,b\in\F_q\) such that
  \[
  (Av,v)=a(0,p_x)+b(p_y,p_y)=(bp_y,ap_x+bp_y).
  \]
  This is equivalent to $Av=bp_y$ and $v=ap_x+bp_y$. Since \(A\) is invertible and \(v\neq0\), we have \(Av\neq0\), so \(b\neq0\).
  Put $w:=Av$. Then $w\in\langle p_y\rangle$ and $v=A^{-1}w$. Using the second equation, we get $v-Av=ap_x$. Equivalently, $(A^{-1}-I)w=ap_x$. Since $A^{-1}-I=-A^{-1}(A-I)$ is invertible, and \(w\neq0\), we have \((A^{-1}-I)w\neq0\). Hence
  \(a\neq0\), and therefore $(A^{-1}-I)w\in\langle p_x\rangle$. Since \(w\in\langle p_y\rangle\), this is equivalent to $[A^{-1}-I](y)=x$. This proves the third equivalence.
  \end{proof}

\begin{lemma}\label{lem:standard-line-criteria}
The following intersection criteria hold.

\begin{enumerate}[label=(\alph*)]
\item Two subspaces of the same type satisfy
\[
L_{1}(x,y)\cap L_{1}(x',y')\neq\{0\}
\quad\Longleftrightarrow\quad
x=x'\ \text{or}\ y=y',
\]
and similarly for \(L_{2}\) and \(L_{3}\).

\item For mixed types, one has
\[
L_{1}(x,y)\cap L_{2}(x',y')\neq\{0\}
\quad\Longleftrightarrow\quad
x=x'\ \text{or}\ y=y',
\]
\[
L_{1}(x,y)\cap L_{3}(x',y')\neq\{0\}
\quad\Longleftrightarrow\quad
x=y'\ \text{or}\ y=x',
\]
and
\[
L_{2}(x,y)\cap L_{3}(x',y')\neq\{0\}
\quad\Longleftrightarrow\quad
x=x'\ \text{or}\ y=y'.
\]

\item The
incidence pattern between the repair subspaces \(W_i\) and the node subspaces
\(L_{1},L_{2},L_{3}\) is as follows:
\[
W_1 \text{ hits every } L_{1}(x,y),L_{2}(x,y),
\]
and
\[
W_1\text{ hits }L_{3}(x,y)
\quad\Longleftrightarrow\quad
x=y.
\]
Similarly,
\[
W_2 \text{ hits every } L_{1}(x,y),L_{3}(x,y),
\]
and
\[
W_2\text{ hits }L_{2}(x,y)
\quad\Longleftrightarrow\quad
x=y.
\]
Moreover,
\[
W_3\text{ hits every } L_{2}(x,y),L_{3}(x,y),
\]
and
\[
W_3\text{ hits }L_{1}(x,y)
\quad\Longleftrightarrow\quad
x=y.
\]
Finally,
\[
W_4\text{ hits }L_{1}(x,y)
\quad\Longleftrightarrow\quad
x=\tau(y)=y+1,
\]
\[
W_4\text{ hits }L_{2}(x,y)
\quad\Longleftrightarrow\quad
x=\tau^2(y)=y+2,
\]
and
\[
W_4\text{ hits }L_{3}(x,y)
\quad\Longleftrightarrow\quad
x=\tau(y)=y+1.
\]
\end{enumerate}
\end{lemma}

\begin{proof}
  The intersection criteria among the node subspaces are obtained by solving the
  corresponding two-component linear equations in
  \(\F_q^2\oplus\F_q^2\). The incidence statements for
  \(W_1,W_2,W_3\) are immediate from the definitions. 
  
  It remains only to explain the incidence with $W_4=\Gamma(T)$. By Lemma~\ref{lem:graph-line-incidence}, we have, for instance,
  \[
  W_4\cap L_2(x,y)\neq\{0\}
  \quad\Longleftrightarrow\quad
  [T-I](y)=x.
  \]
  By Lemma~\ref{lem:cyclic-T}(ii), one has $[T-I]=[T]^2=\tau^2$. Therefore
  \[
  W_4\cap L_2(x,y)\neq\{0\}
  \quad\Longleftrightarrow\quad
  x=\tau^2(y)=y+2.
  \]
  The cases \(W_4\cap L_1(x,y)\) and \(W_4\cap L_3(x,y)\) are obtained in the
  same way from Lemma~\ref{lem:graph-line-incidence}, using respectively $[T]=\tau$ and
  \[
  [T^{-1}-I]
  =
  [T^{-1}]\,[T-I]
  =
  [T]^{-1}[T]^2
  =
  [T]
  =
  \tau.
  \]
  This proves the lemma.
  \end{proof}

Having fixed the cyclic input and the standard incidence criteria, we now give
the endpoint skeletons for the three residue classes of \(q+1\) modulo \(3\).

\begin{lemma}\label{lem:long-bw-endpoints}
Let \(q\) be a prime power. Write $q+1=3m+s$ with $m\in\mathbb{N}_+$ and $s\in\{0,1,2\}$. Then there exists a family \(\Omega\) of pairwise skew two-dimensional
subspaces of \(\mathbb V\) of size
\[
|\Omega|
=
\widetilde{N}_{\BW}(q)
=
\begin{cases}
4m,&s=0,\\
4m+2,&s=1,\\
4m+3,&s=2,
\end{cases}
\]
together with a designated repair subspace \(W(L)\in\{W_1,W_2,W_3,W_4\}\)
for every \(L\in\Omega\), such that $W(L)\cap L=\{0\}$ and \(W(L)\) hits exactly \(q+1\) members of
\(\Omega\setminus\{L\}\).
\end{lemma}

\begin{proof}
We give the construction in the three cases \(s=0,1,2\). Note that we adopt the convention that a set of the form $\{a_i:0\le i\le -1\}$ is understood to be empty.

\medskip
\noindent
\textbf{Case 1: \(q+1=3m\).}
Define residue classes in \(\mathbb Z/(q+1)\mathbb Z\) by
\[
A_i:=\{i+3a:0\le a\le m-1\},
\qquad i=0,1,2.
\]
Each has size \(m\). Define
\[
\Omega_1:=\{L_{3}(a+1,a):a\in A_1\},
\]
\[
\Omega_2:=\{L_{2}(b+2,b):b\in A_2\},
\]
\[
\Omega_3:=\{L_{1}(a+1,a):a\in A_1\},
\]
and
\[
\Omega_4:=\{L_{1}(c,c):c\in A_0\}.
\]
Set $\Omega:=\Omega_1\cup\Omega_2\cup\Omega_3\cup\Omega_4$. Then the disjointness of the residue classes \(A_0,A_1,A_2\), together with
Lemma~\ref{lem:standard-line-criteria}(a)(b), shows that $|\Omega|=4m$ and that all subspaces in \(\Omega\)
are pairwise skew.

The intersection criteria from Lemma~\ref{lem:standard-line-criteria}(c) give
\[
W_1\text{ hits all subspaces in }\Omega_2\cup\Omega_3\cup\Omega_4\text{ and misses all subspaces in }\Omega_1,
\]
\[
W_2\text{ hits all subspaces in }\Omega_1\cup\Omega_3\cup\Omega_4\text{ and misses all subspaces in }\Omega_2,
\]
\[
W_3\text{ hits all subspaces in }\Omega_1\cup\Omega_2\cup\Omega_4\text{ and misses all subspaces in }\Omega_3,
\]
and
\[
W_4\text{ hits all subspaces in }\Omega_1\cup\Omega_2\cup\Omega_3\text{ and misses all subspaces in }\Omega_4.
\]
Each $W_i$ ($i\in[4]$) therefore hits exactly \(3m=q+1\) subspaces in $\Omega$. We designate
\[
\Omega_1\leadsto W_1,\qquad
\Omega_2\leadsto W_2,\qquad
\Omega_3\leadsto W_3,\qquad
\Omega_4\leadsto W_4.
\]

\medskip
\noindent
\textbf{Case 2: \(q+1=3m+1\).}
Define
\[
  S:=\{1+3a:0\le a\le m-1\}\cup\{3m\},
\]
\[
C:=\{3+3a:0\le a\le m-2\},
\]
and
\[
D:=\{2+3a:0\le a\le m-2\}.
\]
Then $|S|=m+1$, $|C|=|D|=m-1$ and $\mathbb Z/(q+1)\mathbb Z=S\sqcup C\sqcup D\sqcup\{0,3m-1\}$. Moreover, \(S\cap(S+1)=\varnothing\). Define
\[
\Omega_1:=\{L_{3}(a+1,a):a\in S\},
\]
\[
\Omega_2:=\{L_{2}(d+2,d):d\in D\},
\]
\[
\Omega_3:=\{L_{1}(a+1,a):a\in S\},
\]
and
\[
\Omega_4:=\{L_{1}(c,c):c\in C\}.
\]
Finally define two additional subspaces
\[
R_-:=L_{2}(1,0),
\qquad
R_+:=L_{2}(3m,3m-1).
\]
Set $\Omega:=\Omega_1\cup\Omega_2\cup\Omega_3\cup\Omega_4
\cup\{R_-,R_+\}$. Then
\[
|\Omega|
=
(m+1)+(m-1)+(m+1)+(m-1)+2
=
4m+2.
\]
Using Lemma~\ref{lem:standard-line-criteria}, the pairwise skewness reduces to
the elementary disjointness relations
\[
S\cap(S+1)=\varnothing,\qquad
C\cap S=C\cap(S+1)=\varnothing,
\]
\[
D\cap S=(D+1)\cap S=D\cap C=(D+2)\cap C=\varnothing,
\]
together with the defining choices of \(R_-\) and \(R_+\). Thus the subspaces
in \(\Omega\) are pairwise skew.

The intersection criteria from Lemma~\ref{lem:standard-line-criteria}(c) give
\[
W_1\text{ hits all subspaces in }\Omega_2\cup\Omega_3\cup\Omega_4\cup\{R_-,R_+\}
\text{ and misses all subspaces in }\Omega_1,
\]
\[
W_2\text{ hits all subspaces in }\Omega_1\cup\Omega_3\cup\Omega_4
\text{ and misses all subspaces in }\Omega_2\cup\{R_-,R_+\},
\]
\[
W_3\text{ hits all subspaces in }\Omega_1\cup\Omega_2\cup\Omega_4\cup\{R_-,R_+\}
\text{ and misses all subspaces in }\Omega_3,
\]
and
\[
W_4\text{ hits all subspaces in }\Omega_1\cup\Omega_2\cup\Omega_3
\text{ and misses all subspaces in }\Omega_4\cup\{R_-,R_+\}.
\]
The corresponding hit counts are all equal to $3m+1=q+1$. We designate
\[
\Omega_1\leadsto W_1,\qquad
\Omega_2\leadsto W_2,\qquad
\Omega_3\leadsto W_3,\qquad
\Omega_4\leadsto W_4,
\]
and
\[
R_-\leadsto W_2,\qquad R_+\leadsto W_2.
\]

\medskip
\noindent
\textbf{Case 3: \(q+1=3m+2\).}
Define
\[
S:=\{1+3a:0\le a\le m-1\}\cup\{3m\},
\]
\[
C:=\{3a:0\le a\le m-1\},
\]
and
\[
D:=\{2+3a:0\le a\le m-2\}\cup\{3m+1\}.
\]
Then $|S|=m+1$, $|C|=|D|=m$ and $\mathbb Z/(q+1)\mathbb Z=S\sqcup C\sqcup D\sqcup\{3m-1\}$. Define
\[
\Omega_1:=\{L_{3}(a+1,a):a\in S\},
\]
\[
\Omega_2:=\{L_{2}(d+2,d):d\in D\},
\]
\[
\Omega_3:=\{L_{1}(a+1,a):a\in S\},
\]
and
\[
\Omega_4:=\{L_{1}(c,c):c\in C\}.
\]
Finally define one additional subspace
\[
R_0:=L_{2}(3m,3m-1).
\]
Set $\Omega:=\Omega_1\cup\Omega_2\cup\Omega_3\cup\Omega_4\cup\{R_0\}$. Then
\[
|\Omega|
=
(m+1)+m+(m+1)+m+1
=
4m+3.
\]
Again, Lemma~\ref{lem:standard-line-criteria} reduces pairwise skewness to the
same elementary disjointness checks among \(S,C,D\) and their shifts; the
choice of \(R_0\) is disjoint from all four families. Hence \(\Omega\) is
pairwise skew.

The intersection criteria from Lemma~\ref{lem:standard-line-criteria}(c) give
\[
W_1\text{ hits all subspaces in }\Omega_2\cup\Omega_3\cup\Omega_4\cup\{R_0\}
\text{ and misses all subspaces in }\Omega_1,
\]
\[
W_2\text{ hits all subspaces in }\Omega_1\cup\Omega_3\cup\Omega_4
\text{ and misses all subspaces in }\Omega_2\cup\{R_0\},
\]
\[
W_3\text{ hits all subspaces in }\Omega_1\cup\Omega_2\cup\Omega_4\cup\{R_0\}
\text{ and misses all subspaces in }\Omega_3,
\]
and
\[
W_4\text{ hits all subspaces in }\Omega_1\cup\Omega_2\cup\Omega_3
\text{ and misses all subspaces in }\Omega_4\cup\{R_0\}.
\]
Each $W_i$ ($i\in[4]$) hits exactly $3m+2=q+1$ subspaces in $\Omega$. We designate
\[
\Omega_1\leadsto W_1,\qquad
\Omega_2\leadsto W_2,\qquad
\Omega_3\leadsto W_3,\qquad
\Omega_4\leadsto W_4,\qquad R_0\leadsto W_2.
\]

In all three cases, the constructed family \(\Omega\) has the required size and
is pairwise skew. Moreover, each node subspace \(L\in\Omega\) is assigned
a feasible repair subspace \(W(L)\), and this repair subspace hits
exactly \(q+1\) helper subspaces.
\end{proof}

\begin{remark}
  Note that all node subspaces in the endpoint skeletons are of the restricted forms
  \[
  L_1(y,y),\quad L_1(y+1,y),\quad L_2(y+2,y),\quad
  L_2(y+1,y),\quad L_3(y+1,y).
  \]
  Thus the hit/miss assertions in Lemma~\ref{lem:long-bw-endpoints} and the bad pair analysis in Lemma~\ref{lem:long-bw-extension} are only applied to these
selected shifts, not to arbitrary subspaces of type \(L_i(x,y)\).
  \end{remark}

\begin{lemma}\label{lem:long-bw-extension}
Let \(q\ge5\) be a prime power, and write $q+1=3m+s$ with $s\in\{0,1,2\}$. Let \(\Omega\) be the endpoint skeleton from
Lemma~\ref{lem:long-bw-endpoints}. Then for every integer $t$ with $0\le t\le 2q+2-|\Omega|$, there exists a family \(\Omega_t\) of pairwise skew two-dimensional subspaces
of \(\mathbb V\), containing \(\Omega\), with $|\Omega_t|=|\Omega|+t$, such that every node subspace \(L\in\Omega_t\) has a designated repair
subspace \(W(L)\) satisfying $W(L)\cap L=\{0\}$ and hitting exactly \(q+1\) helper subspaces.
\end{lemma}

\begin{proof}
For $d\in\{2,3\}$ and $\lambda\in\F_q^\times$, define $A_{d,\lambda}:=\lambda T^d$ and $R_{d,\lambda}:=\Gamma(A_{d,\lambda})$. We claim that if $(d,\lambda)\ne (d',\lambda')$, then $\lambda T^d-\lambda' T^{d'}$ is a non-zero element of \(\F_q[T]\). Indeed, if \(\lambda T^d=\lambda'T^{d'}\), then in \(\PGL_2(\F_q)\) we would
have $[T]^{d-d'}=1$. Since \(d,d'\in\{2,3\}\) and \([T]\) has order \(q+1\), this forces
\(d=d'\), and then \(\lambda=\lambda'\). Thus, if \((d,\lambda)\neq(d',\lambda')\), then the matrix \(\lambda T^d-\lambda'T^{d'}\) is non-zero. By Lemma~\ref{lem:cyclic-T}(iv), $\mathbb{F}_q[T]$ is a field. Hence if \((d,\lambda)\neq(d',\lambda')\), then \(\lambda T^d-\lambda'T^{d'}\) is invertible.

We then show that the subspaces \(R_{d,\lambda}\) are pairwise skew. Indeed, suppose that $R_{d,\lambda}\cap R_{d',\lambda'}\neq\{0\}$ with $(d,\lambda)\neq(d',\lambda')$. Then there exists a non-zero vector \(v\in\F_q^2\) such that $(\lambda T^d(v),v)=(\lambda' T^{d'}v,v)$, which implies that $(\lambda T^d-\lambda' T^{d'})v=0$. Since $\lambda T^d-\lambda' T^{d'}$ is invertible, one has $v=0$, a contradiction. Therefore $R_{d,\lambda}\cap R_{d',\lambda'}=\{0\}$ whenever \((d,\lambda)\neq(d',\lambda')\).

Each \(R_{d,\lambda}\) is skew to \(W_1\) and \(W_2\), because
\(\lambda T^d\) is invertible. It is also skew to \(W_3=\Gamma(I)\), because
\(\lambda T^d-I\in\F_q[T]\) is singular only if it is zero, which would imply
\(\tau^d=1\) in \(\PGL_2(\F_q)\). This is impossible for \(d=2,3\) because the order of $\tau$ in $\PGL_2(\F_q)$ is $q+1$. Similarly,
\(R_{d,\lambda}\) is skew to \(W_4=\Gamma(T)\), because
\(\lambda T^d-T=0\) would imply \(\tau^{d-1}=1\) in \(\PGL_2(\F_q)\), also impossible.

We call a pair \((d,\lambda)\), with \(d\in\{2,3\}\) and
\(\lambda\in\F_q^\times\), \emph{bad} if $R_{d,\lambda}\cap L\neq\{0\}$ for some node subspace \(L\) in the endpoint skeleton. Otherwise we
call it \emph{good}. Fix \(d\in\{2,3\}\). We now bound the number of bad values
of \(\lambda\). 

First consider the selected \(L_1\)-type node subspaces. By
Lemma~\ref{lem:graph-line-incidence},
\[
\Gamma(A_{d,\lambda})\cap L_1(x,y)\neq\{0\}
\quad\Longleftrightarrow\quad
[A_{d,\lambda}](y)=x.
\]
Since $[A_{d,\lambda}]=[T]^d=\tau^d$, this can occur only for \(L_1\)-type nodes of the form $L_1(y+d,y)$. However, in the endpoint skeletons, the selected \(L_1\)-type node subspaces have only $L_1(y,y)$ or $L_1(y+1,y)$. Since \(d\in\{2,3\}\), no selected \(L_1\)-type node is met by
\(R_{d,\lambda}\).

It remains to consider the selected \(L_2\)- and \(L_3\)-type node subspaces. Again by Lemma~\ref{lem:graph-line-incidence}, a selected node subspace of the form $L_2(y+e,y)$ is met by \(R_{d,\lambda}\) if and only if
\[
[A_{d,\lambda}-I](y)=y+e=\tau^e(y).
\]
Since \(A_{d,\lambda}-I\in\F_q[T]\) and \(A_{d,\lambda}-I\) is invertible
(as shown when proving that \(R_{d,\lambda}\) is skew to \(W_3=\Gamma(I)\)), Remark~\ref{rem:FqT-projective-classes} implies that $[A_{d,\lambda}-I]=\tau^i$ for a unique \(i\in\mathbb Z/(q+1)\mathbb Z\). It follows that $\tau^i(y)=\tau^e(y)$. By Lemma~\ref{lem:cyclic-T}(iii), the \(\tau\)-orbit of $y$ has length \(q+1\), so
\(i=e\) in $\mathbb{Z}/(q+1)\mathbb{Z}$. Therefore
\[
R_{d,\lambda}\cap L_2(y+e,y)\neq\{0\}
\quad\Longleftrightarrow\quad
[A_{d,\lambda}-I]=\tau^e.
\]
Similarly, a selected node subspace of the form $L_3(y+1,y)$ is met by \(R_{d,\lambda}\) if and only if $[A_{d,\lambda}^{-1}-I]=\tau$. Since $A_{d,\lambda}^{-1}-I=-A_{d,\lambda}^{-1}(A_{d,\lambda}-I)$, the last condition is equivalent to $[A_{d,\lambda}-I]=[A_{d,\lambda}]\tau=\tau^{d+1}$. 

Therefore, using the explicit forms of the selected node subspaces in the proof
of Lemma~\ref{lem:long-bw-endpoints}, all possible values of \(\lambda\in\F_q^\times\) such that $R_{d,\lambda}$ meets the endpoint skeleton are contained among the conditions
\begin{equation}\label{A-condition}
[A_{d,\lambda}-I]=\tau^e,
\end{equation}
where
\[
e\in
\begin{cases}
\{2,d+1\},&q+1\equiv0\pmod3,\\
\{1,2,d+1\},&q+1\equiv1\pmod3,\\
\{1,2,d+1\},&q+1\equiv2\pmod3.
\end{cases}
\]

For fixed \(d\in\{2,3\}\) and a fixed exponent \(e\), we claim that equation \eqref{A-condition} excludes at most one value of \(\lambda\in\F_q^\times\). Indeed, this equation
means that there exists \(\mu\in\F_q^\times\) such that $\lambda T^d-I=\mu T^e$. Suppose
two distinct values \(\lambda_1,\lambda_2\) satisfy the condition. Then there
exist \(\mu_1,\mu_2\in\F_q^\times\) such that
\[
\lambda_iT^d-I=\mu_iT^e
\qquad(i=1,2).
\]
Subtracting gives
\[
(\lambda_1-\lambda_2)T^d=(\mu_1-\mu_2)T^e.
\]
Since \(\lambda_1\ne\lambda_2\), this implies $T^d\in \F_q\cdot T^e$. If \(d\ne e\), then \(T^{d-e}\) is scalar, so $[T]^{d-e}=1$ in $\PGL_2(\F_q)$, contradicting the fact that \([T]\) has order \(q+1\). If \(d=e\), then the
original equation $\lambda T^d-I=\mu T^d$ would imply $I\in\F_q\cdot T^d$, so \(T^d\) is scalar, again contradicting the order of \([T]\). Therefore at
most one value of \(\lambda\) can satisfy equation \eqref{A-condition}.

It follows that, for each fixed \(d\in\{2,3\}\), at most two values of
\(\lambda\in\F_q^\times\) give bad pairs \((d,\lambda)\) when \(s=0\), and at
most three values of \(\lambda\in\F_q^\times\) give bad pairs \((d,\lambda)\)
when \(s=1\) or \(s=2\). Hence the total number of good pairs \((d,\lambda)\) is at
least
\[
B_q:=\begin{cases}
2(q-3),&s=0,\\
2(q-4),&s=1,2.
\end{cases}
\]

We now compare this with the number of additional subspaces needed to reach
length \(2q+2\). Since
\[
|\Omega|=
\begin{cases}
4m,&s=0,\\
4m+2,&s=1,\\
4m+3,&s=2,
\end{cases}
\]
and $q+1=3m+s$, we have
\[
2q+2-|\Omega|
=
\begin{cases}
2m,&s=0,1,\\
2m+1,&s=2.
\end{cases}
\]
For \(q\ge5\), we have $B_q\ge 2q+2-|\Omega|$. Thus, for any integer $t$ with $0\le t\le 2q+2-|\Omega|$, we may choose \(t\) distinct good pairs $(d_1,\lambda_1)$, $\dots$, $(d_t,\lambda_t)$. Set
\[
\Omega_t:=\Omega\cup\{R_{d_1,\lambda_1},\dots,R_{d_t,\lambda_t}\}.
\]
Then \(\Omega_t\) consists of pairwise skew two-dimensional subspaces
of \(\mathbb V\), contains $\Omega$, and has size $|\Omega_t|=|\Omega|+t$.

It remains to assign repair subspaces to the node subspaces in \(\Omega_t\).
For node subspaces already contained in \(\Omega\), we keep the same designated
repair subspaces as in the endpoint skeleton. Since every new node subspace is
skew to each of the four repair subspaces \(W_i\), \(i\in[4]\), these repair
subspaces hit exactly the same helper subspaces as before. Hence every old node
subspace still has a feasible repair subspace hitting exactly \(q+1\) helper
subspaces.

Now let $R_{d,\lambda}\in \Omega_t\setminus\Omega$ be a new node subspace. We assign \(W_1\) as its repair subspace. We have shown that $W_1\cap R_{d,\lambda}=\{0\}$, so \(W_1\) is feasible for repairing \(R_{d,\lambda}\). Moreover, from the
endpoint incidence pattern, \(W_1\) hits exactly \(q+1\) node subspaces in
\(\Omega\). It hits no node subspace in \(\Omega_t\setminus\Omega\), because
all newly added node subspaces are skew to \(W_1\). Therefore every new node
subspace also has a feasible repair subspace hitting exactly \(q+1\) helper
subspaces. This proves the lemma.
\end{proof}

\begin{proposition}\label{pro:long-bw}
Let \(q\ge5\) be a prime power. For every integer $n$ satisfying
\[
  \widetilde{N}_{\BW}(q)\le n\le q^2+1,
\]
there exists an \((n,n-2,2)\) MDS array code $\mathcal{C}$ over \(\F_q\) attaining the incidence-multiplicity bound for repair bandwidth. Equivalently,
for every \(i\in[n]\), one has $\alpha_i(\mathcal{C})=q+1$.
\end{proposition}

\begin{proof}
  For code lengths beyond \(2q+2\), Theorem~\ref{thm:incidence-multiplicity-bound-attainable} already gives the
attainability of the incidence-multiplicity bound for repair bandwidth. Hence we may assume that $\widetilde{N}_{\BW}(q)\le n\le 2q+2$. By Lemma~\ref{lem:long-bw-endpoints}, there is an endpoint skeleton
\(\Omega\) of size $\widetilde{N}_{\BW}(q)$. Given integer $n$ with $\widetilde{N}_{\BW}(q)\le n\le 2q+2$, put $t:=n-|\Omega|$. Then $0\le t\le 2q+2-|\Omega|$. By Lemma~\ref{lem:long-bw-extension}, there exists a pairwise skew family
\(\Omega_t\) of \(n\) two-dimensional subspaces of \(\mathbb V\) such that
every node subspace in $\Omega_t$ has a designated repair subspace hitting exactly \(q+1\) helper subspaces. By the subspace reformulation and Theorem~\ref{thm:incidence-multiplicity-bound}, $\Omega_t$ determines an \((n,n-2,2)\) MDS array code $\mathcal{C}$ over \(\F_q\) such that $\alpha_i(\mathcal{C})=q+1$ for any $i\in [n]$.
\end{proof}

Theorem~\ref{thm:incidence-multiplicity-bound} and Proposition~\ref{pro:long-bw} together give the following result:

\begin{corollary}\label{cor:long-optimal-bw}
  Let $q\ge 5$ be a prime power. Then for any integer $n$ with $\widetilde{N}_{\BW}(q)\le n\le q^2+1$, one has 
  $$\beta_{\opt}^{(q,n,2,2)}=2n-q-3.$$
\end{corollary}

\subsection{Repair I/O}

For any prime power \(q\), put
\[
\widetilde N_{\IO}(q):=
\begin{cases}
\dfrac{3(q+1)}{2},&q\text{ odd},\\[2mm]
\dfrac{3q}{2}+2,&q\text{ even}.
\end{cases}
\]
In this subsection we study attainability of the
incidence-multiplicity bound for repair I/O in the range
\[
\widetilde N_{\IO}(q)\le n\le 2q+2.
\]

We continue to use the notation from the repair-bandwidth subsection. In
particular,
\[
\mathbb V=\F_q^2\oplus\F_q^2,
\]
and \(L_1(x,y),L_2(x,y),L_3(x,y)\) are the three template node subspaces. We
also use the same terminology: a repair subspace \(W\) \emph{hits} a node
subspace \(L\) if \(W\cap L\neq\{0\}\), and otherwise \(W\) \emph{misses}
\(L\).

\begin{lemma}\label{lem:long-io-endpoint}
Let \(q\ge 3\) be a prime power. Then there exists a
family \(\Omega\) of pairwise skew two-dimensional subspaces of \(\mathbb V\)
of size $\widetilde N_{\IO}(q)$, together with projective column sets
\[
\mathcal X_L\subseteq \mathbb P(L),
\qquad |\mathcal X_L|=2
\qquad(L\in\Omega),
\]
such that every node subspace \(L\in\Omega\) has a designated repair subspace
\(W(L)\) satisfying $W(L)\cap L=\{0\}$ and such that \(W(L)\) hits exactly \(q+1\) helper subspaces and contains
exactly one chosen projective column point on each of those \(q+1\) helper subspaces.
\end{lemma}

\begin{proof}
We split the construction according to the parity of \(q\).

\medskip
\noindent
\textbf{Case 1: \(q\) odd.}
Since \(|\mathbb P^1(\F_q)|=q+1\) is even, choose a partition
\[
\mathbb P^1(\F_q)=A\sqcup B,
\qquad
|A|=|B|=\frac{q+1}{2}.
\]
Choose bijections
\[
f:A\to B,\qquad g:B\to A,\qquad h:A\to B.
\]
Define three families of node subspaces:
\[
X:=\{L_1(a,f(a)):a\in A\},
\]
\[
Y:=\{L_2(b,g(b)):b\in B\},
\]
and
\[
Z:=\{L_3(a,h(a)):a\in A\}.
\]
Set
\[
\Omega:=X\cup Y\cup Z.
\]
Then Lemma~\ref{lem:standard-line-criteria}(a)(b) imply that
\[
|\Omega|=\frac{3(q+1)}{2}=\widetilde N_{\IO}(q)
\]
and the node subspaces in \(\Omega\) are pairwise skew.

We assign repair subspaces as follows:
\[
X\leadsto W_3,\qquad
Y\leadsto W_2,\qquad
Z\leadsto W_1.
\]
By Lemma~\ref{lem:standard-line-criteria}(c), \(W_3\) misses every node subspace in
\(X\) and hits all node subspaces in \(Y\cup Z\); \(W_2\) misses every node subspace in \(Y\) and
hits all node subspaces in \(X\cup Z\); and \(W_1\) misses every node subspace in \(Z\) and hits
all node subspaces in \(X\cup Y\). Each of these hit sets has size
\[
\frac{q+1}{2}+\frac{q+1}{2}=q+1.
\]

Now choose projective column sets. Set
\[
\mathcal X_L:=\begin{cases}
  \{L\cap W_1,\ L\cap W_2\},&L\in X,\\
  \{L\cap W_1,\ L\cap W_3\},&L\in Y,\\
  \{L\cap W_2,\ L\cap W_3\},&L\in Z.
\end{cases}
\]
These are two distinct projective points on \(L\), and they span \(L\).
Moreover, for each designated repair subspace, every helper it hits contains
exactly one chosen projective column point on that repair subspace.

\medskip
\noindent
\textbf{Case 2: \(q\) even.}
Recall from the proof of Lemma~\ref{lem:cyclic-T} that, when \(q\) is even,
the matrix \(T\) may be chosen so that its characteristic polynomial is \(x^2+cx+c\) for some \(c\in\F_q^\times\). If
\(v_1\in\mathbb{F}_q^2\setminus\{0\}\) and \(v_2:=Tv_1\), then \(v_1,v_2\) form a $\mathbb{F}_q$-basis of $\mathbb{F}_q^2$. In this basis, one has
\[
Tv_1=v_2,
\qquad
Tv_2=T^2v_1=cTv_1+cv_1=cv_2+cv_1,
\]
Hence we may assume that
\[
T=
\begin{pmatrix}
0&c\\
1&c
\end{pmatrix}
\qquad(c\in\F_q^\times).
\]
We identify \(\mathbb P^1(\F_q)\) with \(\F_q\cup\{\infty\}\) by
\[
z\longleftrightarrow \langle (z,1)^T\rangle,
\qquad
\infty\longleftrightarrow \langle (1,0)^T\rangle.
\]
Then \(\tau\) is given by
\[
\tau(z)=\frac{c}{z+c},\qquad \tau(\infty)=0.
\]

By Lemma~\ref{lem:cyclic-T}(iii), the projective points
$$0,\tau(0),\tau^2(0),\cdots,\tau^q(0)=\infty$$
in $\mathbb{P}^1(\F_q)=\F_q\cup\{\infty\}$ are pairwise distinct and hence contain all points of $\mathbb{P}^1(\F_q)$. Then we can  identify $\mathbb Z/(q+1)\mathbb Z$ with $\mathbb P^1(\F_q)$ via $t\leftrightarrow\tau^t(0)$, so that
that the action of $\tau$ is simply
\[
\tau(t)=t+1.
\]
Under this identification, the affine point \(0\in\F_q\cup\{\infty\}\) corresponds to
\(0\in\mathbb Z/(q+1)\mathbb Z\), and the infinity point
\(\infty\in\F_q\cup\{\infty\}\) corresponds to
\(q\in\mathbb Z/(q+1)\mathbb Z\). 

Define
\[
S:=\{1,3,5,\dots,q-1\}
\subseteq \mathbb Z/(q+1)\mathbb Z,
\]
and set $C:=\mathbb P^1(\F_q)\setminus S$. Then
\[
|S|=\frac q2,
\qquad
|C|=\frac q2+1.
\]
Moreover, the two points $0,\infty\in\F_q\cup\{\infty\}$ are not in $S$.

Define
\[
X:=\{L_1(a+1,a):a\in S\},
\]
\[
Y_0:=\{L_3(a+1,a):a\in S\},
\]
and
\[
Z:=\{L_2(b+1,b):b\in C\}.
\]
Since $S\cap C=\varnothing$ and $S\cap (S+1)=\varnothing$, Lemma~\ref{lem:standard-line-criteria}(a)(b) imply that all subspaces in
\(X\cup Y_0\cup Z\) are pairwise skew. In particular, there are $|S|+|S|+|C|=\frac{3q}{2}+1$ node subspaces in $X\cup Y_0\cup Z$.

We now add one extra node subspace. For \(s\in\F_q^\times\), define
\[
A_s:=
\begin{pmatrix}
0&s\\
1&s
\end{pmatrix},
\qquad
R_s:=\Gamma(A_s).
\]
The matrix \(A_s\) is invertible, and \(A_s-I\) is also invertible. Hence
\[
R_s\cap W_1=\{0\},
\qquad
R_s\cap W_2=\{0\},
\qquad
R_s\cap W_3=\{0\}.
\]
Moreover,
\[
A_s-T
=
\begin{pmatrix}
0&s-c\\
0&s-c
\end{pmatrix}
\]
is singular, so \(R_s\) is hit by \(W_4=\Gamma(T)\).

The map on \(\mathbb P^1(\F_q)=\F_q\cup\{\infty\}\) induced by \(A_s\) is
\[
g_s(z)=\frac{s}{z+s},
\qquad
g_s(\infty)=0.
\]
Moreover, $A_s^2+sA_s+sI=0$, and since \(s\neq0\), this implies
\[
A_s^{-1}-I=s^{-1}A_s.
\]
Hence \(A_s^{-1}-I\) induces the same projective map \(g_s\). A direct calculation gives
\[
g_s(z)=\tau(z)
\quad\Longleftrightarrow\quad
s=c\ \text{ or }\ z\in\{0,\infty\}.
\]
Thus, by Lemma~\ref{lem:graph-line-incidence}, if \(s\neq c\), then \(R_s\) is skew to all node subspaces in \(X\cup Y_0\).

It remains to make \(R_s\) skew to all node subspaces in \(Z\). The map induced
by \(A_s-I\) is
\[
h_s(z)=\frac{z+s}{z+s+1},
\qquad
g_s(\infty)=1.
\]
Then by Lemma~\ref{lem:graph-line-incidence}, it suffices to require that
\[
h_s(z)\neq \tau(z)
\qquad
\text{for every }z\in\mathbb P^1(\F_q).
\]
For \(z\in\F_q\), the equality \(h_s(z)=\tau(z)\) is equivalent to
\[
z^2+sz+c=0.
\]
This quadratic has no root in \(\F_q\) if and only if
\[
\operatorname{Tr}_{\F_q/\F_2}\left(\frac{c}{s^2}\right)=1.
\]
Since the map $s\mapsto \frac{c}{s^2}$ is a bijection of \(\F_q^\times\), and exactly half of the elements of
\(\F_q\) have trace \(1\), there are \(q/2\) values of
\(s\in\F_q^\times\) satisfying this trace condition. At most one of them is
\(c\), so since \(q\ge4\), we may choose $s\in\F_q^\times\setminus\{c\}$ so that
\[
\operatorname{Tr}_{\F_q/\F_2}\left(\frac{c}{s^2}\right)=1.
\]
For this choice of \(s\), the maps \(h_s\) and \(\tau\) have no agreement point
on \(\mathbb P^1(\F_q)\). Thus \(R_s\) is skew to all node subspaces in \(Z\).

Set
\[
R:=R_s,
\qquad
Y:=Y_0\cup\{R\},
\qquad
\Omega:=X\cup Y\cup Z.
\]
Then
\[
|X|=\frac q2,\qquad
|Y|=\frac q2+1,\qquad
|Z|=\frac q2+1,
\]
and hence
\[
|\Omega|=\frac{3q}{2}+2=\widetilde N_{\IO}(q).
\]
The construction above shows that all node subspaces \(\Omega\) are pairwise skew.

We assign repair subspaces as follows:
\[
X\leadsto W_3,\qquad
Y\leadsto W_1,\qquad
Z\leadsto W_4.
\]
By Lemma~\ref{lem:standard-line-criteria}(c), together with the facts that
\(A_s\) and \(A_s-I\) are invertible while \(A_s-T\) is singular, we obtain the
following incidence relations:
\[
W_3\text{ misses }X\cup\{R\}\text{ and hits exactly }Y_0\cup Z,
\]
\[
W_1\text{ misses }Y\text{ and hits exactly }X\cup Z,
\]
and
\[
W_4\text{ misses }Z\text{ and hits exactly }X\cup Y.
\]
The corresponding hit counts are
\[
|Y_0|+|Z|=\frac q2+\left(\frac q2+1\right)=q+1,
\]
\[
|X|+|Z|=\frac q2+\left(\frac q2+1\right)=q+1,
\]
and
\[
|X|+|Y|=\frac q2+\left(\frac q2+1\right)=q+1.
\]

Now choose projective column sets. Set
\[
\mathcal X_L:=\begin{cases}
  \{L\cap W_1,\ L\cap W_4\},&L\in X,\\
  \{L\cap W_3,\ L\cap W_4\},&L\in Y_0,\\
  \{L\cap W_1,\ L\cap W_3\},&L\in Z,\\
  \{R\cap W_4,\ Q_R\},&L=R,
\end{cases}
\]
where \(Q_R\) is any projective point on \(R\) distinct from \(R\cap W_4\).
Again, for each designated repair subspace, every helper it hits contains
exactly one chosen projective column point on that repair subspace. This proves the lemma.
\end{proof}

\begin{lemma}\label{lem:long-io-extension}
Assume that $q\ge 3$ is a prime power. Let \(\Omega\) be
the endpoint skeleton from Lemma~\ref{lem:long-io-endpoint}. Then for
every integer $t$ satisfying $0\le t\le 2q+1-|\Omega|$, there exists a family \(\Omega_t\) of pairwise skew two-dimensional subspaces
of \(\mathbb V\), containing \(\Omega\), with $|\Omega_t|=|\Omega|+t$, and projective column sets on all node subspaces in \(\Omega_t\), such that every node subspace \(L\in\Omega_t\) has a designated repair subspace
\(W(L)\) satisfying $W(L)\cap L=\{0\}$ and such that \(W(L)\) hits exactly \(q+1\) helper subspaces and contains
exactly one chosen projective column point on each of those \(q+1\) helper subspaces.
\end{lemma}

\begin{proof}
We treat the two parities in parallel.

\medskip
\noindent
\textbf{Case 1: \(q\) odd.}
Use the notation from the odd case in the proof of
Lemma~\ref{lem:long-io-endpoint}. For $c\in\F_q^\times\setminus\{1\}$, define $R_c:=\Gamma(cI)$. The subspaces \(R_c\) are pairwise skew. They are skew to \(W_1,W_2,W_3\), and
they are skew to all node subspaces in the endpoint skeleton \(\Omega\). Indeed,
the possible intersections would force equality of two projective labels lying
in different parts of the partition
\[
\mathbb P^1(\F_q)=A\sqcup B.
\]
Thus every \(R_c\) can be added without changing the repair incidence of the
old nodes.

There are $q-2$ available choices of such \(c\). Since
\[
2q+1-|\Omega|
=
2q+1-\frac{3(q+1)}2
=
\frac{q-1}{2}
\le q-2
\qquad(q\ge 3),
\]
for any integer $t$ satisfying $0\le t\le 2q+1-|\Omega|$, we can choose $t$ distinct $R_c$ ($c\in\mathbb{F}_q^{\times}\setminus\{1\}$) and add them to \(\Omega\) to form $\Omega_t$. For each added
subspace \(R_c\), choose any two distinct projective points spanning \(R_c\),
and assign \(W_1\) as its repair subspace. The repair subspace \(W_1\) hits
exactly the old helper subspaces in \(X\cup Y\), whose total number is \(q+1\),
and it misses every $R_c$. Moreover, each node subspace in
\(X\cup Y\) has exactly one chosen projective column point on \(W_1\). Hence $\Omega_t$ satisfies the required conditions.

\medskip
\noindent
\textbf{Case 2: \(q\) even.}
Use the notation from the even case in the proof of
Lemma~\ref{lem:long-io-endpoint}. Thus
\[
\Omega=X\cup Y_0\cup Z\cup\{R_0\},
\qquad
R_0=\Gamma(B)
\]
for some \(B\in\GL_2(\F_q)\), and the designated repair subspaces are
\(W_1,W_3,W_4\).

For \((u,v)\in\F_q^2\), define
\[
A_{u,v}:=uI+vT,
\qquad
R_{u,v}:=\Gamma(A_{u,v}).
\]
Since \(T\) has irreducible characteristic polynomial, every non-zero
\(\F_q\)-linear combination of \(I\) and \(T\) is invertible. Hence distinct
subspaces \(R_{u,v}\) are pairwise skew.

We call a pair \((u,v)\) \emph{good} if \(R_{u,v}\) is skew to every node subspace in
\(\Omega\) and skew to each of \(W_1,W_3,W_4\). We now show that for any integer $t$ satisfying $0\le t\le 2q+1-|\Omega|$, we may find $t$ distinct good pairs.

Fix $u\in\F_q\setminus\{0,1\}$. For such \(u\), the subspace \(R_{u,v}\) is skew to \(W_1,W_3,W_4\), and it is
automatically skew to all node subspaces in \(X\cup Z\). It remains only to
avoid \(Y_0\) and \(R_0\). By Lemma~\ref{lem:graph-line-incidence}, the
condition that \(R_{u,v}\) meet \(Y_0\) is given by a quadratic equation in
\(v\), and hence excludes at most two values of \(v\). Similarly,
\[
R_{u,v}\cap R_0\neq\{0\}
\quad\Longleftrightarrow\quad
\det(A_{u,v}-B)=0,
\]
which is again a quadratic condition in \(v\), and hence excludes at most two
values of \(v\).

Therefore, for each fixed \(u\in\F_q\setminus\{0,1\}\), at least \(q-4\)
values of \(v\) give good pairs. Since there are \(q-2\) choices of \(u\), the
number of good pairs is at least
\[
(q-2)(q-4).
\]
For \(q\ge8\), one has
\[
(q-2)(q-4)\ge \frac q2-1
=
2q+1-\left(\frac{3q}{2}+2\right)
=
2q+1-|\Omega|.
\]
Thus in this case we have enough good pairs.

It remains only to handle \(q=4\). In this case
\[
\F_4^\times=\{1,c,c+1\},
\qquad
c^2+c+1=0.
\]
In the endpoint skeleton, the extra line may be chosen with \(s=1\).
A direct substitution into the conditions above shows that the following four
pairs are good:
\[
(c,1),\qquad (c,c),\qquad (c,c+1),\qquad (c+1,c).
\]
Thus there are at least four good pairs when \(q=4\). Since in
this case
\[
2q+2-|\Omega|
=
10-8
=
2,
\]
we again have enough good pairs.

For any integer $t$ satisfying $0\le t\le 2q+1-|\Omega|$, choose any $t$ distinct good pairs and add the corresponding subspaces $R_{u,v}$ to $\Omega$ to form $\Omega_t$. For every newly added subspace \(R_{u,v}\), choose any two distinct projective
points spanning it, and assign \(W_1\) as its repair subspace. Since the new
subspace is skew to \(W_1\), this repair subspace is feasible. It hits exactly
the old helper subspaces in $X\cup Z$, whose number is
\[
|X|+|Z|=\frac q2+\left(\frac q2+1\right)=q+1,
\]
and it misses all newly added subspaces. Also, each node subspace in
\(X\cup Z\) has exactly one chosen projective column point on \(W_1\). Hence $\Omega_t$ satisfies the required conditions.
\end{proof}

\begin{proposition}\label{pro:long-io}
Let $q\ge 3$ be a prime power. For every integer $n$ satisfying
\[
  \widetilde N_{\IO}(q)\le n\le q^2+1,
\]
there exists an \((n,n-2,2)\) MDS array code \(\mathcal C\) over \(\F_q\) attaining the incidence-multiplicity bound simultaneously for repair bandwidth and repair I/O. Equivalently,
for every \(i\in[n]\), one has $\alpha_i(\mathcal{C})=\lambda_i(\mathcal{C})=q+1$.
\end{proposition}

\begin{proof}
  For code lengths beyond \(2q+1\), Theorem~\ref{thm:incidence-multiplicity-bound-attainable} already gives the simultaneous attainability of the incidence-multiplicity bound for both repair bandwidth and repair I/O. Hence we may assume that $\widetilde{N}_{\IO}(q)\le n\le 2q+1$. By Lemma~\ref{lem:long-io-endpoint}, there is an endpoint skeleton
\(\Omega\) of size $\widetilde{N}_{\IO}(q)$. Given integer $n$ with $\widetilde{N}_{\IO}(q)\le n\le 2q+1$, put $t:=n-|\Omega|$. Then $0\le t\le 2q+1-|\Omega|$. By Lemma~\ref{lem:long-io-extension}, there exists a pairwise skew family
\(\Omega_t\) of \(n\) two-dimensional subspaces of \(\mathbb V\), containing
\(\Omega\), together with projective column sets
\[
\mathcal X_L\subseteq\mathbb P(L),
\qquad |\mathcal X_L|=2
\qquad (L\in\Omega_t),
\]
such that every node subspace \(L\in\Omega_t\) has a designated repair subspace
\(W(L)\) satisfying
\[
W(L)\cap L=\{0\},
\]
and this repair subspace hits exactly \(q+1\) helper subspaces and contains
exactly one chosen projective column point on each of those helpers.

By the subspace reformulation and Theorem~\ref{thm:incidence-multiplicity-bound}, the family \(\Omega_t\), together with these
projective column sets, determines an \((n,n-2,2)\) MDS array code
\(\mathcal C\) over \(\F_q\) such that 
\[
\alpha_i(\mathcal C)=\lambda_i(\mathcal C)=q+1
\qquad(i\in[n]).
\]
\end{proof}

Theorem~\ref{thm:incidence-multiplicity-bound} and Proposition~\ref{pro:long-io} together give the following result:

\begin{corollary}\label{cor:long-optimal-io}
  Let $q\ge 3$ be a prime power. Then for any integer $n$ with $\widetilde{N}_{\IO}(q)\le n\le q^2+1$, one has 
  $$\gamma_{\opt}^{(q,n,2,2)}=2n-q-3.$$
\end{corollary}

\section{Proofs of the Main Theorems}\label{sec:mainproof}

We first prove the main theorem for repair I/O, since it is comparatively simpler.

\begin{proof}[Proof of Theorem~\ref{thm:main-2}]
  Recall that the short-length I/O threshold is
\[
N_{\IO}(q)
=
\begin{cases}
\dfrac{3(q+1)}{2},&q\text{ odd},\\[2mm]
\dfrac{3q}{2}+1,&q\text{ even},
\end{cases}
\]
while the long-length I/O threshold is
\[
\widetilde N_{\IO}(q)
=
\begin{cases}
\dfrac{3(q+1)}{2},&q\text{ odd},\\[2mm]
\dfrac{3q}{2}+2,&q\text{ even}.
\end{cases}
\]
Thus, if \(q\) is odd, then $N_{\IO}(q)=\widetilde N_{\IO}(q)$, whereas if \(q\) is even, then $\widetilde N_{\IO}(q)=N_{\IO}(q)+1$.

By Corollary~\ref{cor:short-optimal-io}, for every integer \(n\) with $3\le n\le N_{\IO}(q)$ and \(n\neq4\), one has
\[
\gamma_{\opt}^{(q,n,2,2)}
=
\left\lceil\frac{4n-6}{3}\right\rceil.
\]
By Corollary~\ref{cor:long-optimal-io}, if $q\ge 3$, for every integer \(n\) with $\widetilde N_{\IO}(q)\le n\le q^2+1$, one has
\[
\gamma_{\opt}^{(q,n,2,2)}
=
2n-q-3.
\]
Finally, by Proposition~\ref{prop:gamma-opt-4-2-2}, for every prime power
\(q\),
\[
\gamma_{\opt}^{(q,4,2,2)}=3.
\]
Since the two ranges 
$$3\le n\le N_{\IO}(q)\qquad\text{and}\qquad\widetilde{N}_{\IO}(q)\le n\le q^2+1$$
together cover precisely the range of code lengths stated in Theorem~\ref{thm:main-2}, the theorem follows for $q\ge 3$.

It remains to handle \(q=2\). In this case the allowed code lengths are $3\le n\le 2^2+1=5$. Since
\[
N_{\IO}(2)=\frac{3\cdot2}{2}+1=4,
\]
Corollary~\ref{cor:short-optimal-io} gives the result for \(q=2,n=3\). Since
Proposition~\ref{prop:gamma-opt-4-2-2} holds for any prime power $q$, it gives the result for \(q=2,n=4\).

For \(n=5\), it is easy to see that the full Desarguesian spread construction in \cite{liu2026linear} attains the
incidence-multiplicity bound for repair I/O. Thus the stated formula also holds for \(q=2,n=5\). This completes the proof.
\end{proof}

Next, we prove the main theorem for repair bandwidth. Recall that the short-length bandwidth threshold is
\[
  N_{\BW}(q)=
  \begin{cases}
  \dfrac{4q}{3},&q\equiv0\pmod3,\\[2mm]
  \dfrac{4q+2}{3},&q\equiv1\pmod3,\\[2mm]
  \dfrac{4(q+1)}{3},&q\equiv2\pmod3,
  \end{cases}
\]
while the long-length bandwidth threshold is
\[
  \widetilde{N}_{\BW}(q)=
  \begin{cases}
  \dfrac{4q}{3}+2,&q\equiv0\pmod3,\\[2mm]
  \dfrac{4q+2}{3}+1,&q\equiv1\pmod3,\\[2mm]
  \dfrac{4(q+1)}{3},&q\equiv2\pmod3.
  \end{cases}
\]

For the cases $q\equiv 1,2\pmod3$, the two ranges
$$3\le n\le N_{\BW}(q)\qquad\text{and}\qquad\widetilde{N}_{\BW}(q)\le n\le q^2+1$$
together cover precisely the range of code lengths stated in Theorem~\ref{thm:main-1}. However, when $q\equiv 0\pmod3$, there remains one gap between these two ranges, namely 
$$n=\frac{4q}{3}+1.$$
We first handle this exceptional length separately.

\begin{lemma}\label{lem:short-bw-extra-mod0}
  Let $q$ be a prime power with \(q\equiv0\pmod3\). Then there exists an $(\frac{4q}{3}+1,\frac{4q}{3}-1,2)$ MDS array code \(\mathcal C\) over \(\F_q\) such that
  \[
  \beta(\mathcal C)
  \le
  \left\lceil\frac{5n-8}{4}\right\rceil,
  \qquad
  n=\frac{4q}{3}+1.
  \]
  \end{lemma}
  
  \begin{proof}
  By Lemma~\ref{lem:long-bw-endpoints}, there exists a
  pairwise skew family \(\Omega\) of two-dimensional node subspaces of size
  \[
  |\Omega|=\widetilde{N}_{\BW}(q)=\frac{4q}{3}+2=4m+2
  \]
  such that every node subspace in $\Omega$ has a designated feasible repair subspace hitting
  exactly \(q+1\) helper subspaces.
  
  Delete an arbitrary node subspace from \(\Omega\), and denote the remaining
  family by \(\Omega'\). Then
  \[
  |\Omega'|=\frac{4q}{3}+1.
  \]
  Since the node spaces in \(\Omega\) are pairwise skew, the node spaces in
  \(\Omega'\) remain pairwise skew. Therefore, \(\Omega'\) determines a $(\frac{4q}{3}+1,\frac{4q}{3}-1,2)$ MDS array code \(\mathcal C\) over \(\F_q\).
  
  For any remaining node subspace, keep its original designated repair subspace.
  This repair subspace is still feasible. Moreover, it originally hit exactly
  \(q+1\) helper subspaces in \(\Omega\); after deleting one node, it hits at least
  \(q\) helper subspaces in \(\Omega'\). Hence the repair bandwidth $\beta(\mathcal{C})$ is at most
  \[
  2(|\Omega'|-1)-q
  =\frac{5q}{3}=\left\lceil\frac{5q}{3}-\frac{3}{4}\right\rceil=\left\lceil\frac{5n-8}{4}\right\rceil.
  \]
  \end{proof}

  \begin{proof}[Proof of Theorem~\ref{thm:main-1}]
    Note that
    \[
      \widetilde{N}_{\BW}(q)-N_{\BW}(q)=
      \begin{cases}
      2,&q\equiv0\pmod3,\\[2mm]
      1,&q\equiv1\pmod3,\\[2mm]
      0,&q\equiv2\pmod3.
      \end{cases}
    \]
 By Corollary~\ref{cor:short-optimal-bw}, for any integer $n$ with $3\le n\le N_{\BW}(q)$ and $n\not\in\{5,6,9,10\}$, one has 
    $$\beta_{\opt}^{(q,n,2,2)}=\left\lceil \frac{5n-8}{4}\right\rceil.$$
    By Proposition~\ref{pro:long-bw}, if \(q\ge5\), then for every integer $n$ with $\widetilde{N}_{\BW}(q)\le n\le q^2+1$, one has 
\[
\beta_{\opt}^{(q,n,2,2)}
=
2n-q-3.
\]

When $q\equiv 1,2\pmod3$, the two ranges
\[
3\le n\le N_{\BW}(q)
\qquad\text{and}\qquad
\widetilde{N}_{\BW}(q)\le n\le q^2+1
\]
together cover precisely the range of code lengths stated in Theorem~\ref{thm:main-1}. Therefore, Theorem~\ref{thm:main-1} holds for
$q\ge5$, $q\equiv 1,2\pmod3$, and $n\not\in\{5,6,9,10\}$.
When $q\equiv 0\pmod3$, there remains exactly one gap, namely
\[
n=\frac{4q}{3}+1.
\]
However, Lemma~\ref{lem:short-bw-extra-mod0} fills this gap. Therefore, Theorem~\ref{thm:main-1} also holds for
$q\ge5$, $q\equiv 0\pmod3$, and $n\not\in\{5,6,9,10\}$.

It remains to consider the cases where $q=2,3,4$ or $n\in\{5,6,9,10\}$.
For $n=6$, the assertion follows from Propositions~\ref{prop:beta-n6},
\ref{prop:beta-n5,6-5}, and~\ref{prop:beta-n6-q-3-4}. For $n=5$, the
assertion follows from Proposition~\ref{prop:beta-n5}.
For $n=9$ or $10$, the assertion in even characteristic follows from
Propositions~\ref{prop:beta-even-n9-n10}, \ref{prop:no-beta10-q98-n10},
and~\ref{prop:no-beta9-q8-n9}; the assertion in characteristic \(3\)
follows from Propositions~\ref{prop:beta-char3-n9-n10},
\ref{prop:beta-q9-n9}, and~\ref{prop:no-beta10-q98-n10}; and the
assertion in all remaining characteristics follows from
Proposition~\ref{prop:beta-n9-n10-odd-not-char3}.

For $q=4$, we have $N_{\BW}(q)=\frac{4\cdot 4+2}{3}=6$. Hence Corollary~\ref{cor:short-optimal-bw} proves the assertion for code
lengths $3$ and $4$, Proposition~\ref{prop:beta-n5} proves the assertion
for code length $5$, and Proposition~\ref{prop:beta-n6-q-3-4} proves the
assertion for code length $6$.
For code length $7$, Lemma~\ref{lem:long-bw-endpoints} gives a construction
attaining the incidence-multiplicity bound, equivalently the
projective counting bound. Hence the assertion follows.
For code length $8$, Theorem~\ref{thm:main-2} shows that the
incidence-multiplicity bound, equivalently the projective counting bound,
is attainable for repair I/O. Since attaining this bound for repair I/O
also implies attaining it for repair bandwidth, the assertion follows.
Moreover, \cite[Theorem~1.3]{liu2026linear} proves the assertion for code
lengths
\[
9\le n\le 4^2+1=17.
\]

For $q=3$, we have $N_{\BW}(q)=\frac{4\cdot 3}{3}=4$. Hence
Corollary~\ref{cor:short-optimal-bw} proves the assertion for code
lengths $3$ and $4$, Proposition~\ref{prop:beta-n5} proves the assertion
for code length $5$, and Proposition~\ref{prop:beta-n6-q-3-4} proves the
assertion for code length $6$.
Moreover, \cite[Theorem~1.3]{liu2026linear} proves the assertion for code
lengths
\[
7\le n\le 3^2+1=10.
\]

It remains to consider \(q=2\). For \(n=4,5\), Theorem~\ref{thm:main-2}
shows that the incidence-multiplicity bound, equivalently the
projective counting bound, is attainable for repair I/O. Since attaining
this bound for repair I/O also implies attaining it for repair bandwidth,
the assertion follows. For \(n=3\), Theorem~\ref{thm:main-2} gives $\gamma_{\opt}^{(2,3,2,2)}=2$. Hence $\beta_{\opt}^{(2,3,2,2)}\le \gamma_{\opt}^{(2,3,2,2)}\le 2$. On the other hand, Theorem~\ref{thm:ZLH-lb} gives $\beta_{\opt}^{(2,3,2,2)}\ge 2$. Therefore \(\beta_{\opt}^{(2,3,2,2)}=2\), as stated in the theorem. This completes the proof.
  \end{proof}

\clearpage
\appendix
\section*{Appendix}
\phantomsection
\addcontentsline{toc}{section}{Appendix}
\setcounter{subsection}{0}
\renewcommand{\thesubsection}{\Alph{subsection}}
\makeatletter
\gdef\theHsection{appendix}%
\makeatother
\numberwithin{theorem}{subsection}

\refstepcounter{subsection}\label{app:beta-opt-exc}
\subsection*{\thesubsection\quad Optimal Repair Bandwidth for Small-Parameter Exceptional Cases}
\addcontentsline{toc}{appsubsection}{\protect\numberline{\thesubsection}Optimal Repair Bandwidth for Small-Parameter Exceptional Cases}

We handle the code lengths \(n=5,6,9,10\), which are not covered by Corollary~\ref{cor:short-optimal-bw}. By Theorem~\ref{thm:ZLH-lb} and Proposition~\ref{prop:short-bw}, when
\(3\le n\le N_{\BW}(q)\), we have
\[
  \beta_{\opt}^{(q,n,2,2)}
  \in
  \left\{
    \left\lceil \frac{5n-10}{4}\right\rceil,
    \left\lceil \frac{5n-8}{4}\right\rceil
  \right\}.
\]
For \(n=5,6,9,10\), the two quantities above differ by exactly \(1\).
Therefore, to prove $\beta_{\opt}^{(q,n,2,2)}=\left\lceil \frac{5n-10}{4}\right\rceil$, it suffices to construct an explicit \((n,n-2,2)\) MDS array code
\(\mathcal C\) over \(\F_q\) such that $\beta(\mathcal C)=\left\lceil \frac{5n-10}{4}\right\rceil$. We will show that, for these values of \(n\), such a construction always exists provided that \(q\) is sufficiently large.

To prove $\beta_{\opt}^{(q,n,2,2)}
=\left\lceil \frac{5n-8}{4}\right\rceil$, one has to rule out the existence of an \((n,n-2,2)\) MDS array code
\(\mathcal C\) satisfying $\beta(\mathcal C)=\left\lceil \frac{5n-10}{4}\right\rceil$. This is often considerably more difficult. However, as mentioned above, for the values \(n=5,6,9,10\) under consideration,
there are only finitely many, and in fact very few, values of \(q\) for which
this situation can occur. We may therefore use an exhaustive computer
search to handle these remaining cases.

To make the exhaustive search feasible, we first quotient out by several
natural equivalences. We may apply an invertible row operation to the whole
parity-check matrix \(H\), and we may also apply an invertible column operation
inside each parity-check block \(H_i\). More precisely, we may replace each $H_i$ by
\[
  U H_i V_i,
  \qquad
  U\in\GL_4(\F_q),\quad V_i\in\GL_2(\F_q).
\]
These operations preserve the MDS property. Indeed, for \(i\neq j\),
\[
  [U H_i V_i,\; U H_j V_j]
  =
  U[H_i,\;H_j]
  \begin{bmatrix}
    V_i&0\\
    0&V_j
  \end{bmatrix},
\]
so \([H_i,H_j]\) is invertible if and only if
\([U H_i V_i,U H_j V_j]\) is invertible. They also preserve all repair ranks appearing in the repair-bandwidth computation. Indeed, given any repair matrix \(M\), replace it by \(M U^{-1}\). Then, for
each \(j\),
\[
  (M U^{-1})(U H_j V_j)=M H_j V_j.
\]
In particular,
\[
  M H_i\in\GL_2(\F_q)
  \quad\Longleftrightarrow\quad
  (M U^{-1})(U H_i V_i)=M H_i V_i\in\GL_2(\F_q),
\]
and
\[
  \sum_{j\ne i}
  \rank\bigl((M U^{-1})(U H_j V_j)\bigr)
  =
  \sum_{j\ne i}\rank(MH_jV_j)
  =
  \sum_{j\ne i}\rank(MH_j),
\]
because right multiplication by an invertible matrix does not change rank.
Thus these equivalences do not change the MDS property or any value of
\(\beta_i\).

In addition, when computing \(\beta_i\) for a fixed \(i\), we may left-multiply
the repair matrix by any invertible \(2\times2\) matrix without changing the
ranks \(\rank(MH_j)\). Indeed, for \(A\in\GL_2(\F_q)\),
\[
  \rank(AMH_j)=\rank(MH_j)
  \qquad\text{for all }j.
\]
In particular, whenever \(MH_i\in\GL_2(\F_q)\), replacing \(M\) by
\((MH_i)^{-1}M\) gives the normalization
\[
  MH_i=I_2.
\]
Hence, in the minimization defining \(\beta_i\), it suffices to enumerate
repair matrices satisfying \(MH_i=I_2\). This normalization substantially reduces the number of free variables in the exhaustive search.

As a final preparation, we prove the following lemma, which allows us to reduce
the number of cases to be considered.

\begin{lemma}\label{lem:reduction-by-puncturing}
Let \(q\) be a prime power. Then the following implications hold:
\begin{enumerate}[label=(\alph*)]
  \item If \(\beta_{\opt}^{(q,6,2,2)}=5\), then $\beta_{\opt}^{(q,5,2,2)}=4$.
  \item If \(\beta_{\opt}^{(q,5,2,2)}=5\), then $\beta_{\opt}^{(q,6,2,2)}=6$.
  \item If \(\beta_{\opt}^{(q,10,2,2)}=10\), then $\beta_{\opt}^{(q,9,2,2)}=9$.
  \item If \(\beta_{\opt}^{(q,9,2,2)}=10\), then $\beta_{\opt}^{(q,10,2,2)}=11$.
\end{enumerate}
\end{lemma}
\begin{proof}
  We prove (a). Suppose that $\beta_{\opt}^{(q,6,2,2)}=5$. Let \(\mathcal C\) be a \((6,4,2)\) MDS array code over \(\F_q\) with
\(\beta(\mathcal C)=5\), and let
\[
  H=[\,H_1\ H_2\ H_3\ H_4\ H_5\ H_6\,]
\]
be a parity-check matrix of \(\mathcal C\), where each
\(H_i\in\F_q^{4\times 2}\). Since \(\beta(\mathcal C)=5\), there exist
repair matrices \(M_i\in\F_q^{2\times 4}\), \(i\in[6]\), such that
\[
  M_iH_i\in\GL_2(\F_q)
\]
and
\[
  \sum_{j\in[6]\setminus\{i\}}\rank(M_iH_j)\le 5
\]
for every \(i\in[6]\).

We claim that $M_iH_j\neq 0$ for all $i\neq j$. Suppose, to the contrary, that \(M_iH_j=0\) for some \(i_0\neq j_0\).
Since \(M_{i_0}H_{i_0}\in\GL_2(\F_q)\), we have \(\rank(M_{i_0})=2\).
Now take any \(i\in[6]\setminus\{i_0,j_0\}\). By the MDS property,
\[
  [\,H_{j_0}\ \ H_{i}\,]\in\GL_4(\F_q).
\]
Hence
\[
  \rank\bigl(M_{i_0}[\,H_{j_0}\ \ H_{i}\,]\bigr)=\rank(M_{i_0})=2.
\]
On the other hand,
\[
  M_{i_0}[\,H_{j_0}\ \ H_i\,]=[\,M_{i_0}H_{j_0}\ \ M_{i_0}H_i\,]=[\,0\ \ M_{i_0}H_i\,].
\]
Therefore
\[
  \rank(M_{i_0}H_i)=2.
\]
Thus
\[
  \sum_{j\in[6]\setminus\{i_0\}}\rank(M_{i_0}H_j)
  \ge
  \sum_{i\in[6]\setminus\{i_0,j_0\}}\rank(M_{i_0}H_i)
  =
  4\cdot 2
  =
  8,
\]
contradicting
\[
  \sum_{j\in[6]\setminus\{i_0\}}\rank(M_{i_0}H_j)\le 5.
\]
This proves the claim.

In particular, for every \(i\in[5]\), we have $\rank(M_iH_6)\ge 1$. Let
\[
  H'=[\,H_1\ H_2\ H_3\ H_4\ H_5\,].
\]
Since \(\mathcal C\) is MDS, \(H'\) defines a \((5,3,2)\) MDS array code
\(\mathcal C'\) over \(\F_q\). We use the same repair matrices \(M_i\),
\(i\in[5]\), for \(\mathcal C'\). Then for every \(i\in[5]\), $M_iH_i\in\GL_2(\F_q)$, and
\[
  \sum_{j\in[5]\setminus\{i\}}\rank(M_iH_j)=
  \sum_{j\in[6]\setminus\{i\}}\rank(M_iH_j)-
  \rank(M_iH_6)\le 5-1=4.
\]
Thus $\beta(\mathcal C')\le 4$.

On the other hand, by Theorem~\ref{thm:ZLH-lb},
\[
  \beta(\mathcal C')
  \ge
  \left\lceil\frac{5\cdot 5-10}{4}\right\rceil
  =4.
\]
Therefore $\beta(\mathcal C')=4$, and hence $\beta_{\opt}^{(q,5,2,2)}=4$. This proves (a).

Part (b) follows from the contrapositive of (a), together with the known
bounds
\[
  \beta_{\opt}^{(q,5,2,2)}\in\{4,5\},
  \qquad
  \beta_{\opt}^{(q,6,2,2)}\in\{5,6\}.
\]
Indeed, if \(\beta_{\opt}^{(q,5,2,2)}=5\), then
\(\beta_{\opt}^{(q,5,2,2)}\neq4\), so by the contrapositive of (a) we have
\(\beta_{\opt}^{(q,6,2,2)}\neq5\). Hence
\[
  \beta_{\opt}^{(q,6,2,2)}=6.
\]
The proofs of (c) and (d) are analogous.
\end{proof}

\refstepcounter{subsubsection}\label{app:beta-opt-6-2-2}
\subsubsection*{\thesubsubsection\quad The Exceptional Case \texorpdfstring{$n=5,6$}{n=5,6}}
\addcontentsline{toc}{appsubsubsection}{\protect\numberline{\thesubsubsection}The Exceptional Case \texorpdfstring{$n=5,6$}{n=5,6}}

\begin{proposition}\label{prop:beta-n6}
  Let \(q\ge 7\) be a prime power. Then
  \[
  \beta_{\opt}^{(q,6,2,2)}=5.
  \]
  \end{proposition}
  
  \begin{proof}
  By Theorem~\ref{thm:ZLH-lb}, it suffices to construct a \((6,4,2)\) MDS array code over \(\F_q\)
  with $\beta(\mathcal C)=5$. Choose elements \(s,A,B\in\F_q\) such that
  \[
 s(s-1)AB(A-B)(A+1)(B+1)(A+s)(B+s)(S-A)(S-B)(S-AB-A-B)\neq0,
  \]
  where $S:=s^2-2s$. Such a choice is possible for every \(q\ge 7\). Indeed, first choose
  \(s\notin\{0,1\}\). Then choose
  \[
  A\notin \{0,-1,-s,S\}.
  \]
  Finally choose
  \[
  B\notin
  \left\{
  0,-1,-s,S,A,\frac{S-A}{A+1}
  \right\}.
  \]
  At each step at most \(6\) elements are excluded, so this is possible since
  \(q\ge 7\).
  
  Put
  \[
  \Delta:=AB+A+B-S.
  \]
  For \(b,c\in\F_q\), define
  \[
  W(b,c):=
  \begin{bmatrix}
  1&b\\
  c&1+b+bc
  \end{bmatrix}.
  \]
  Set
  \[
  (b_1,c_1):=(0,0),
  \]
  \[
  (b_2,c_2):=
  \left(
  A,\,
  -\frac{A(s-1)}{s(A+s)}
  \right),
  \]
  \[
  (b_3,c_3):=
  \left(
  B,\,
  -\frac{B(s-1)}{s(B+s)}
  \right),
  \]
  and
  \[
  (b_4,c_4):=
  \left(
  \frac{\Delta}{(s-1)^2},\,
  -\frac{\Delta}{(A+s)(B+s)}
  \right).
  \]
  For \(1\le i\le 4\), let
  \[
  W_i:=W(b_i,c_i).
  \]
  Define the six parity-check blocks by
  \[
  H_i=
  \begin{bmatrix}
  I_2\\
  W_i
  \end{bmatrix}
  \qquad(1\le i\le 4),\qquad H_5=
  \begin{bmatrix}
  I_2\\
  0
  \end{bmatrix},
  \qquad
  H_6=
  \begin{bmatrix}
  0\\
  I_2
  \end{bmatrix}.
  \]
  
  We first verify the MDS property. For any \(b,c,b',c'\in\F_q\), direct
  calculation gives
  \[
  \det\bigl( W(b,c)\bigr)=1+b
  \]
  and
  \[
  \det\bigl(W(b,c)-W(b',c')\bigr)=-(b-b')(c-c').
  \]
  The choice of \(s,A,B\) implies that \(b_i\neq -1\) for every \(i\in[4]\), that
the \(b_i\)'s are pairwise distinct, and that the \(c_i\)'s are pairwise
distinct. To see this, we record the following elementary simplifications. For
the \(b_i\)'s, one has
\[
b_4=0
\iff
\Delta=0
\iff
S-AB-A-B=0,
\]
\[
b_4=-1
\iff
(A+1)(B+1)=0,
\]
and
\[
b_4=A
\iff
(A+1)(S-B)=0,
\qquad
b_4=B
\iff
(B+1)(S-A)=0.
\]
For the \(c_i\)'s, one computes
\[
c_2-c_3
=
-\frac{(A-B)(s-1)}{(A+s)(B+s)},
\]
\[
c_4-c_2
=
\frac{S-B}{s(B+s)},
\qquad
c_4-c_3
=
\frac{S-A}{s(A+s)}.
\]
All numerators and denominators appearing above are non-zero by the choice of
\(s,A,B\). Hence \(b_i\neq -1\) for every \(i\), the \(b_i\)'s are pairwise
distinct, and the \(c_i\)'s are pairwise distinct, which implies that \(W_i\in\GL_2(\F_q)\) for any \(i\in[4]\), and
  \[
  W_i-W_j\in\GL_2(\F_q)
  \qquad(1\le i<j\le 4).
  \]
  It follows that
  \[
  [H_i,H_j]\in \GL_4(\F_q)
  \qquad(1\le i<j\le 6).
  \]
  Hence the parity-check matrix \(H=[\,H_1\ H_2\ H_3\ H_4\ H_5\ H_6\,]\) defines a \((6,4,2)\)
  MDS array code over \(\F_q\).
  
  It remains to construct repair matrices. For \(u,v\in\F_q\), define
  \[
  M(u,v):=
  \begin{bmatrix}
  1&-u&0&0\\
  0&0&v&-1
  \end{bmatrix}.
  \]
  A direct calculation gives
  \[
  \det\!\left(
  M(u,v)
  \begin{bmatrix}
  I_2\\
  W(b,c)
  \end{bmatrix}
  \right)
  =
  -\Phi_{u,v}(b,c),
  \]
  where
  \[
  \Phi_{u,v}(b,c):=(u+b)(c-v)+(1+b).
  \]
  Moreover, whenever \(\Phi_{u,v}(b,c)=0\), the matrix
  \[
  M(u,v)
  \begin{bmatrix}
  I_2\\
  W(b,c)
  \end{bmatrix}
  \]
  has rank \(1\), since its first row is \((1,-u)\).
  
  Now set
  \[
  (u_1,v_1):=
  \left(
  \frac{AB+A+B+s}{s-1},\,
  \frac{(s-1)(s^2-AB)}{s(A+s)(B+s)}
  \right),
  \]
  \[
  (u_2,v_2):=
  \left(
  \frac{B+s}{s-1},\,
  \frac{s-1}{B+s}
  \right),
  \]
  \[
  (u_3,v_3):=
  \left(
  \frac{A+s}{s-1},\,
  \frac{s-1}{A+s}
  \right),
  \]
  and
  \[
  (u_4,v_4):=
  \left(
  s,\,
  \frac1s
  \right).
  \]
  For \(1\le i\le 4\), define
\[
M_i:=M(u_i,v_i).
\]
A direct substitution gives
\[
\Phi_{u_i,v_i}(b_j,c_j)=0
\qquad(1\le i,j\le4,\ i\neq j),
\]
while
\[
\Phi_{u_1,v_1}(b_1,c_1)
=
\frac{AB\Delta}{s(A+s)(B+s)},
\]
\[
\Phi_{u_2,v_2}(b_2,c_2)
=
-\frac{A(A-B)(S-B)}{s(A+s)(B+s)},
\]
\[
\Phi_{u_3,v_3}(b_3,c_3)
=
\frac{B(A-B)(S-A)}{s(A+s)(B+s)},
\]
and
\[
\Phi_{u_4,v_4}(b_4,c_4)
=
\frac{(S-A)(S-B)(S-AB-A-B)}
     {s(A+s)(B+s)(s-1)^2}.
\]
By our choice of \(s,A,B\), all four displayed quantities are non-zero.
Therefore, $M_iH_i\in\GL_2(\F_q)$ for any $i\in [4]$ and
\[
\rank(M_iH_j)=1
\qquad(1\le i,j\le4,\ j\neq i).
\]

It remains to check the two remaining helper nodes. For every \(u,v\in\F_q\),
one has
\[
M(u,v)H_5
=
\begin{bmatrix}
1&-u\\
0&0
\end{bmatrix},
\qquad
M(u,v)H_6
=
\begin{bmatrix}
0&0\\
v&-1
\end{bmatrix}.
\]
Hence
\[
\rank(M(u,v)H_5)=\rank(M(u,v)H_6)=1.
\]
In particular,
\[
\rank(M_iH_5)=\rank(M_iH_6)=1
\qquad(i\in [4]).
\]
Thus each of the first four nodes can be repaired by downloading one symbol
from each of the other five nodes.
  
  For the last two nodes, define
  \[
  M_5:=
  \begin{bmatrix}
  0&1&0&0\\
  1&0&-1&0
  \end{bmatrix},
  \qquad
  M_6:=
  \begin{bmatrix}
  0&0&1&0\\
  1&-1&0&1
  \end{bmatrix}.
  \]
  Then
  \[
  M_5H_5=
  \begin{bmatrix}
  0&1\\
  1&0
  \end{bmatrix}
  \in\GL_2(\F_q),
  \]
  and, for every $i\in [4]$,
  \[
  M_5H_i=
  \begin{bmatrix}
  0&1\\
  0&-b_i
  \end{bmatrix},
  \qquad
  M_5H_6=
  \begin{bmatrix}
  0&0\\
  -1&0
  \end{bmatrix}.
  \]
  Hence
  \[
  \rank(M_5H_j)=1
  \qquad(j\neq 5).
  \]
  Similarly,
  \[
  M_6H_6=I_2\in\GL_2(\F_q),
  \]
  and, for every \(1\le i\le 4\),
  \[
  M_6H_i=
  \begin{bmatrix}
  1&b_i\\
  c_i+1&b_i(c_i+1)
  \end{bmatrix},
  \qquad
  M_6H_5=
  \begin{bmatrix}
  0&0\\
  1&-1
  \end{bmatrix}.
  \]
  Thus
  \[
  \rank(M_6H_j)=1
  \qquad(j\neq 6).
  \]
  
  Consequently every node can be repaired by downloading exactly one symbol from
  each of the other five nodes. Therefore
  \[
  \beta(\mathcal C)=5.
  \]
  \end{proof}

  \begin{proposition}\label{prop:beta-n6-q-3-4}
    One has
    \[
    \beta_{\opt}^{(3,6,2,2)}=6,\qquad
    \beta_{\opt}^{(4,6,2,2)}=5.
    \]
    \end{proposition}
    
    \begin{proof}
    We treat the three cases separately.
    
    \smallskip
    \noindent\textbf{The case \(q=3\).}
    By Theorem~\ref{thm:incidence-multiplicity-bound},
    \[
    \beta(\mathcal C)\ge 2n-q-3=12-3-3=6
    \]
    for every \((6,4,2)\) MDS array code $\mathcal C$ over \(\F_3\). On the other hand, Proposition B.1 in \cite{liu2026linear} shows the existence of a \((6,4,2)\) MDS array code $\mathcal{C}_0$ over \(\F_3\)
    with $\beta(\mathcal{C}_0)=6$. Hence
    \[
    \beta_{\opt}^{(3,6,2,2)}=6.
    \]
    
    \smallskip
    \noindent\textbf{The case \(q=4\).}
    By Theorem~\ref{thm:ZLH-lb},
    $$\beta(\mathcal{C})\ge\left\lceil \frac{5\cdot 6-10}{4}\right\rceil=5$$
    for every \((6,4,2)\) MDS array code $\mathcal C$ over \(\F_4\). Hence it suffices to construct a \((6,4,2)\) MDS array code $\mathcal C$ over \(\F_4\) with $\beta(\mathcal{C})=5$. Let
    \[
    \F_4=\{0,1,\omega,\omega+1\},
    \qquad
    \omega^2+\omega+1=0.
    \]
    Define
    \[
    H=[\,H_1\ H_2\ H_3\ H_4\ H_5\ H_6\,],
    \]
    where
    \[
    H_1=
    \begin{bmatrix}
    0&1\\
    1&0\\
    1&0\\
    0&1
    \end{bmatrix},\qquad
    H_2=
    \begin{bmatrix}
    1&0\\
    1&1\\
    0&1\\
    0&1
    \end{bmatrix},\qquad H_3=
    \begin{bmatrix}
    \omega+1&0\\
    \omega&1\\
    0&1\\
    0&\omega
    \end{bmatrix},
    \]
    \[
    H_4=
    \begin{bmatrix}
    \omega&0\\
    \omega+1&1\\
    0&1\\
    0&\omega+1
    \end{bmatrix},\qquad H_5=
    \begin{bmatrix}
    0&0\\
    1&0\\
    0&0\\
    0&1
    \end{bmatrix},\qquad  H_6=
    \begin{bmatrix}
    1&0\\
    0&0\\
    0&1\\
    0&0
    \end{bmatrix}.
    \]
    A direct calculation shows that
    \[
    [H_i,H_j]\in\GL_4(\F_4)
    \qquad(i\ne j).
    \]
    Thus \(H\) defines a \((6,4,2)\) MDS array code over \(\F_4\).
    
    Now take
    \[
    M_1=
    \begin{bmatrix}
    0&0&1&0\\
    0&0&0&1
    \end{bmatrix},\qquad
    M_2=
    \begin{bmatrix}
    1&0&1&0\\
    0&1&0&1
    \end{bmatrix},\qquad
    M_3=
    \begin{bmatrix}
    \omega+1&0&1&0\\
    0&\omega&0&1
    \end{bmatrix},
    \]
    \[
    M_4=
    \begin{bmatrix}
    \omega&0&1&0\\
    0&\omega+1&0&1
    \end{bmatrix},\qquad 
    M_5=
    \begin{bmatrix}
    0&1&0&0\\
    1&0&0&1
    \end{bmatrix},\qquad
    M_6=
    \begin{bmatrix}
    1&0&0&0\\
    0&1&1&0
    \end{bmatrix}.
    \]
    Again, direct calculation gives
    \[
    M_iH_i\in\GL_2(\F_4)
    \qquad(i\in [6]),
    \]
    and
    \[
    \rank(M_iH_j)=1
    \qquad(i\ne j).
    \]
    Therefore each failed node can be repaired by downloading one symbol from each
    of the other five nodes, and thus $\beta(\mathcal C)=5$. It follows that
    \[
    \beta_{\opt}^{(4,6,2,2)}=5.
    \]
\end{proof}

  \begin{proposition}\label{prop:beta-n5}
    One has
    \[
    \beta_{\opt}^{(q,5,2,2)}
    =
    \begin{cases}
    5,& q=2,\\
    4,& q=3,\ q=4,\text{ or } q\ge 7.
    \end{cases}
    \]
    \end{proposition}
    
    \begin{proof}
    Theorem~\ref{thm:ZLH-lb} implies that for any prime power $q$,
    \[
    \beta(\mathcal C)\ge\left\lceil\frac{5\cdot 5-10}{4}\right\rceil=4
    \]
    for any $(5,3,2)$ MDS array code $\mathcal C$ over $\F_q$.
    
    \smallskip
    \noindent\textbf{The case \(q=2\).}
    It is easy to see that the full Desarguesian spread construction in \cite{liu2026linear} attains the
    incidence-multiplicity bound simultaneously for repair bandwidth and repair I/O. Hence
    \[
    \beta_{\opt}^{(2,5,2,2)}=2\cdot 5-2-3=5.
    \]
    
    \smallskip
    \noindent\textbf{The case \(q=3\).}
    We give an explicit construction attaining the lower bound \(4\).
    Over \(\F_3\), let
    \[
    H=[\,H_1\ H_2\ H_3\ H_4\ H_5\,],
    \]
    where
    \[
    H_1=
    \begin{bmatrix}
    0&0\\
    0&0\\
    1&0\\
    0&1
    \end{bmatrix},\quad
    H_2=
    \begin{bmatrix}
    1&0\\
    0&1\\
    0&0\\
    0&0
    \end{bmatrix},\quad H_3=
    \begin{bmatrix}
    1&0\\
    0&1\\
    0&1\\
    1&0
    \end{bmatrix},\quad
    H_4=
    \begin{bmatrix}
    1&0\\
    0&1\\
    0&2\\
    2&1
    \end{bmatrix},\quad  H_5=
    \begin{bmatrix}
    1&0\\
    0&1\\
    1&1\\
    0&2
    \end{bmatrix}.
    \]
    A direct calculation gives
    \[
    [H_i,H_j]\in\GL_4(\F_3)
    \qquad(i\ne j).
    \]
    Thus \(H\) defines a \((5,3,2)\) MDS array code over \(\F_3\).
    
    Now take
    \[
    M_1=
    \begin{bmatrix}
    0&0&0&1\\
    1&1&2&0
    \end{bmatrix},\qquad
    M_2=
    \begin{bmatrix}
    0&1&0&1\\
    1&0&0&2
    \end{bmatrix},\qquad M_3=
    \begin{bmatrix}
    0&0&1&1\\
    1&0&0&0
    \end{bmatrix},
    \]
    \[
    M_4=
    \begin{bmatrix}
    0&0&1&2\\
    1&2&0&0
    \end{bmatrix},\qquad M_5=
    \begin{bmatrix}
    0&0&1&0\\
    0&1&0&0
    \end{bmatrix}.
    \]
    Again, direct calculation gives
    \[
    M_iH_i\in\GL_2(\F_3)
    \qquad(i\in [5]),\qquad \rank(M_iH_j)=1
    \qquad(i\ne j).
    \]
    Therefore each node can be repaired by downloading exactly one symbol from each
    of the other four nodes, and hence
    \[
    \beta_{\opt}^{(3,5,2,2)}=4.
    \]
    
    \smallskip
    \noindent\textbf{The cases \(q=4\) and \(q\ge 7\).}
    These follow directly from Propositions~\ref{prop:beta-n6-q-3-4} and
    \ref{prop:beta-n6}, respectively, together with
    Lemma~\ref{lem:reduction-by-puncturing}(a).
    \end{proof}

    \begin{proposition}\label{prop:beta-n5,6-5}
      One has
      \[
        \beta_{\opt}^{(5,5,2,2)}=5,\qquad \beta_{\opt}^{(5,6,2,2)}=6.
      \]
      \end{proposition}
 \begin{proof}
  By Lemma~\ref{lem:reduction-by-puncturing}(b), we only need to prove that $\beta_{\opt}^{(5,5,2,2)}=5$.
  By Theorem~\ref{thm:ZLH-lb} and Proposition~\ref{prop:short-bw}, it suffices to prove that no \((5,3,2)\) MDS array code $\mathcal{C}$ over \(\F_5\) satisfies $\beta(\mathcal C)=4$. Suppose, for contradiction, that such a code exists. Choose a parity-check matrix
  \[
  H=[\,H_1\ H_2\ H_3\ H_{4}\ H_{5}\,],
  \qquad
  H_i\in \F_5^{4\times 2}.
  \]
  Since the code is MDS, each \(H_i\) has rank \(2\), and
  \[
  [H_i\ H_j]\in\GL_4(\F_5)
  \qquad(i\neq j).
  \]
  
  Since \([H_{4}\ H_5]\in\GL_4(\F_5)\), after an invertible row operation we may
  assume
  \[
  H_{4}=
  \begin{bmatrix}
  I_2\\
  0
  \end{bmatrix},
  \qquad
  H_5=
  \begin{bmatrix}
  0\\
  I_2
  \end{bmatrix}.
  \]
  Write
  \[
  H_1=
  \begin{bmatrix}
  P\\
  Q
  \end{bmatrix}.
  \]
  The invertibility of \([H_1\ H_{4}]\) and \([H_1\ H_5]\) implies $P,Q\in\GL_2(\F_5)$. Apply the invertible row operation
  \[
  \begin{bmatrix}
  P^{-1}&0\\
  0&Q^{-1}
  \end{bmatrix}
  \]
  to the whole parity-check matrix $H$, which sends \(H_1\) to
  \[
  \begin{bmatrix}
  I_2\\
  I_2
  \end{bmatrix},
  \]
  and sends \(H_{4}\) and \(H_5\) to
  \[
  \begin{bmatrix}
  P^{-1}\\
  0
  \end{bmatrix},
  \qquad
  \begin{bmatrix}
  0\\
  Q^{-1}
  \end{bmatrix},
  \]
  respectively. Applying the block-wise invertible column operations \(H_{4}\mapsto H_{4}P\) and
  \(H_5\mapsto H_5Q\), we may therefore assume simultaneously that
  \[
  H_{4}=
  \begin{bmatrix}
  I_2\\
  0
  \end{bmatrix},
  \qquad
  H_5=
  \begin{bmatrix}
  0\\
  I_2
  \end{bmatrix},
  \qquad
  H_1=
  \begin{bmatrix}
  I_2\\
  I_2
  \end{bmatrix}.
  \]
  
  Next write
  \[
  H_2=
  \begin{bmatrix}
  P_2\\
  Q_2
  \end{bmatrix},\qquad
  H_3=
  \begin{bmatrix}
  P_3\\
  Q_3
  \end{bmatrix}
  \]
  For $i=2,3$, since \([H_i\ H_5]\) is invertible, \(P_i\) is invertible. Applying the block-wise invertible column operations
  \(H_2\mapsto H_2P_2^{-1}\) and \(H_3\mapsto H_3P_3^{-1}\), we may write
  \[
  H_2=
  \begin{bmatrix}
  I_2\\
  A
  \end{bmatrix},\qquad
  H_3=
  \begin{bmatrix}
  I_2\\
  B
  \end{bmatrix}
  \]
  with each \(A,B\in \F_5^{2\times 2}\). The remaining MDS conditions are exactly
  \[
  A,\quad B,\quad A-I_2,\quad B-I_2,\quad A-B
  \in\GL_2(\F_5).
  \]
  A \(2\)-subset $\{A,B\}\subseteq\F_5^{2\times2}$ satisfying the above conditions is called \emph{admissible}.

  For any admissible subset $\{A,B\}$, let $\mathcal{C}_{A,B}$ be the $(5,3,2)$ MDS array code over $\F_5$ with parity-check matrix $H=[\,H_1\ H_2\ H_3\ H_{4}\ H_5\,]$, where
\[
H_1=
\begin{bmatrix}
I_2\\
I_2
\end{bmatrix},
\qquad
H_2=
\begin{bmatrix}
I_2\\
A
\end{bmatrix},
\qquad
H_3=
\begin{bmatrix}
I_2\\
B
\end{bmatrix},
\qquad
H_{4}=
\begin{bmatrix}
I_2\\
0
\end{bmatrix},
\qquad
H_5=
\begin{bmatrix}
0\\
I_2
\end{bmatrix}.
\]
Then any counterexample $\mathcal C$ with \(\beta(\mathcal C)=4\) gives an admissible subset
\(\{A,B\}\) such that $\beta(\mathcal C_{A,B})=4$.

We can perform an exhaustive finite check over \(\F_5\). We enumerate all admissible subsets $\{A,B\}$ of $\F_5^{2\times2}$. By the
repair-matrix normalization discussed above, when computing
\(\beta_i(\mathcal C_{A,B})\), the condition \(MH_i\in\GL_2(\F_5)\) may
be replaced by \(MH_i=I_2\). Thus, for each admissible subset \(\{A,B\}\)
and each \(i\in[5]\), we compute
\[
\beta_i(\mathcal C_{A,B})
=
\min_{\substack{M\in\F_5^{2\times4}\\ MH_i=I_2}}
\sum_{j\neq i}\rank(MH_j),
\]
and then
\[
\beta(\mathcal C_{A,B})
=
\max_{i\in[5]}\beta_i(\mathcal C_{A,B}).
\]
The enumeration gives $\beta(\mathcal C_{A,B})\ge 5$ for every admissible subset \(\{A,B\}\), contradicting the existence of the
counterexample. Therefore \[ \beta_{\opt}^{(5,5,2,2)}=5,\qquad\beta_{\opt}^{(5,6,2,2)}=6.\]
 \end{proof}

\refstepcounter{subsubsection}\label{app:beta-opt-9-10-2-2}
\subsubsection*{\thesubsubsection\quad The Exceptional Cases \texorpdfstring{$n=9,10$}{n=9,10}}
\addcontentsline{toc}{appsubsubsection}{\protect\numberline{\thesubsubsection}The Exceptional Cases \texorpdfstring{$n=9,10$}{n=9,10}}

We first record a general template that will be used repeatedly below. The
template reduces the construction of the desired \((10,8,2)\) and
\((9,7,2)\) codes to verifying explicit conditions on five \(2\times2\)
matrices \(C,T_1,T_2,T_3,T_4\).

\begin{lemma}\label{lem:template-10-node-construction}
Let \(q\) be a prime power and let
\[
  C,T_1,T_2,T_3,T_4\in \F_q^{2\times 2}
\]
be invertible matrices. For each \(a\in[4]\), define
\[
  f_a(x):=\det(I_2+xCT_a).
\]
Suppose that the following conditions hold.

\begin{enumerate}[label=(\roman*)]
  \item For each \(a\in[4]\), the polynomial \(f_a(x)\) has two distinct
  non-zero roots in \(\F_q\). Denote its root set by $\Lambda_a\subseteq\F_q^\times$.
  \item For every \(1\le a<b\le4\), let
  \[
    G_{ab}(x,y):=\det(yT_b-xT_a)
  \]
  and define the iterated resultant
  \[
    R_{ab}:=
    \operatorname{Res}_x
    \left(
      f_a(x),
      \operatorname{Res}_y\bigl(f_b(y),G_{ab}(x,y)\bigr)
    \right).
  \]
  Then
  \[
    R_{ab}\neq0
    \qquad(1\le a<b\le4).
  \]

  \item There exist non-zero row vectors
  \[
    u_1,\dots,u_4,\quad v_1,\dots,v_4\in\F_q^{1\times2}
  \]
  such that, for every \(a,b\in[4]\),
  \[
    v_aT_b\in\langle u_a\rangle
    \qquad(b\neq a),
  \]
  while
  \[
    v_aT_a\notin\langle u_a\rangle.
  \]
\end{enumerate}
Then 
\[
  \beta_{\opt}^{(q,10,2,2)}=10,\qquad\qquad \beta_{\opt}^{(q,9,2,2)}=9.
\]
\end{lemma}

\begin{proof}
  For \(a\in[4]\) and \(\lambda\in\Lambda_a\), define
  \[
    H_{a,\lambda}:=
    \begin{bmatrix}
    I_2\\
    \lambda T_a
    \end{bmatrix}.
  \]
  Also set
  \[
    H_9=
    \begin{bmatrix}
    I_2\\
    0
    \end{bmatrix},
    \qquad
    H_{10}=
    \begin{bmatrix}
    0\\
    I_2
    \end{bmatrix}.
  \]
  
  We first verify the MDS property. If \(a=b\) and
  \(\lambda,\mu\in\Lambda_a\) are distinct, then
  \[
    \det[H_{a,\lambda},H_{a,\mu}]
    =
    \det((\mu-\lambda)T_a)\neq0.
  \]
  If \(a\neq b\), then
  \[
    \det[H_{a,\lambda},H_{b,\mu}]
    =
    \det(\mu T_b-\lambda T_a).
  \]
  We claim that this determinant is non-zero for every
  \(\lambda\in\Lambda_a\) and \(\mu\in\Lambda_b\). Indeed, set
  \[
    Q_{ab}(x):=\operatorname{Res}_y\bigl(f_b(y),G_{ab}(x,y)\bigr).
  \]
  Since \(f_b\) splits into two distinct roots in \(\F_q\), say
  \[
    f_b(y)=\tau_b\prod_{\mu\in\Lambda_b}(y-\mu),
    \qquad \tau_b\in\F_q^\times,
  \]
  the basic product formula for resultants gives
  \[
    Q_{ab}(x)
    =
    \gamma_b
    \prod_{\mu\in\Lambda_b}
    G_{ab}(x,\mu)
  \]
  for some \(\gamma_b\in\F_q^\times\). Similarly, since
  \[
    f_a(X)=\tau_a\prod_{\lambda\in\Lambda_a}(X-\lambda),
    \qquad \tau_a\in\F_q^\times,
  \]
  we have
  \[
    R_{ab}
    =
    \operatorname{Res}_x(f_a(X),Q_{ab}(X))
    =
    \kappa_{ab}
    \prod_{\lambda\in\Lambda_a}
    \prod_{\mu\in\Lambda_b}
    G_{ab}(\lambda,\mu)
  \]
  for some \(\kappa_{ab}\in\F_q^\times\). Since
  \[
    G_{ab}(\lambda,\mu)=\det(\mu T_b-\lambda T_a),
  \]
  we obtain
  \[
    R_{ab}
    =
    \kappa_{ab}
    \prod_{\lambda\in\Lambda_a}
    \prod_{\mu\in\Lambda_b}
    \det(\mu T_b-\lambda T_a).
  \]
  Therefore condition \emph{(ii)} implies
  that every factor in the product is non-zero. Moreover, for every \(a\in[4]\) and \(\lambda\in\Lambda_a\),
  \[
    \det[H_{a,\lambda},H_9]=\det(\lambda T_a)\neq0,
  \]
  and
  \[
    \det[H_{a,\lambda},H_{10}]=1.
  \]
  Finally, $[H_9,H_{10}]=I_4$. Thus these ten parity-check blocks define a \((10,8,2)\) MDS array code $\mathcal{C}$.
  
  For \(a\in[4]\), define
  \[
    M^{(a)}:=
    \begin{bmatrix}
    u_a&0\\
    0&v_a
    \end{bmatrix}.
  \]
  Then for \(b\in[4]\) and \(\lambda\in\Lambda_b\),
  \[
    M^{(a)}H_{b,\lambda}
    =
    \begin{bmatrix}
    u_a\\
    \lambda v_aT_b
    \end{bmatrix}.
  \]
  Since \(\lambda\neq0\), condition \emph{(iii)} gives
  \[
    \rank(M^{(a)}H_{b,\lambda})
    =
    \begin{cases}
    2,& b=a,\\
    1,& b\neq a.
    \end{cases}
  \]
  Also,
  \[
    \rank(M^{(a)}H_9)=1,
    \qquad
    \rank(M^{(a)}H_{10})=1.
  \]
  Therefore, for each fixed \(a\in[4]\), either of the two nodes corresponding to
  the blocks
  \[
    H_{a,\lambda},\qquad \lambda\in\Lambda_a,
  \]
  can be repaired using \(M^{(a)}\) with bandwidth
  \[
    2+6+1+1=10.
  \]
  
  For the last two nodes, take
  \[
    M_9=\begin{bmatrix}I_2&C\end{bmatrix},
    \qquad
    M_{10}=\begin{bmatrix}C^{-1}&I_2\end{bmatrix}.
  \]
  For \(a\in[4]\) and \(\lambda\in\Lambda_a\), since \(\lambda\) is a root of
  \(f_a(x)=\det(I_2+xCT_a)\), we have
  \[
    \det(I_2+\lambda CT_a)=0.
  \]
  Moreover, \(I_2+\lambda CT_a\neq0\). Otherwise \(CT_a=-\lambda^{-1}I_2\), and
  then \(f_a(X)\) would be the square of a linear polynomial, contradicting that
  \(f_a\) has two distinct roots. Hence
  \[
    \rank(I_2+\lambda CT_a)=1.
  \]
  Thus
  \[
    \rank(M_9H_{a,\lambda})=1.
  \]
  Similarly,
  \[
    M_{10}H_{a,\lambda}
    =
    C^{-1}(I_2+\lambda CT_a),
  \]
  so
  \[
    \rank(M_{10}H_{a,\lambda})=1.
  \]
  Finally,
  \[
    M_9H_9=I_2,\qquad M_9H_{10}=C,
  \]
  and
  \[
    M_{10}H_{10}=I_2,\qquad M_{10}H_9=C^{-1}.
  \]
  Thus the last two nodes can be repaired with bandwidth
  \[
    8+2=10.
  \]
  
  Therefore the constructed code $\mathcal{C}$ satisfies
  \[
    \beta(\mathcal C)\le10.
  \]
  By Theorem~\ref{thm:ZLH-lb},
  \[
    \beta(\mathcal C)\ge
    \left\lceil\frac{5\cdot10-10}{4}\right\rceil=10.
  \]
  Hence $\beta(\mathcal C)=10$ and consequently
  \[
    \beta_{\opt}^{(q,10,2,2)}=10.
  \]
  By Lemma~\ref{lem:reduction-by-puncturing}(a), one has $\beta_{\opt}^{(q,9,2,2)}=9$.
  \end{proof}

We begin by handling the case where \(q\) is an even prime power.

\begin{lemma}\label{lem:even-one-parameter-construction}
  Let \(q\) be an even prime power. Suppose that there exists \(s\in\F_q\) such that
  \[
    s(s+1)(s^3+s+1)\neq0
  \]
  and
  \[
    \Tr_{\F_q/\F_2}\left(\frac1s\right)=0,
    \qquad
    \Tr_{\F_q/\F_2}\left(\frac1{s(s+1)}\right)=0,
  \]
  \[
    \Tr_{\F_q/\F_2}\left(\frac s{(s+1)^3}\right)=0,
    \qquad
    \Tr_{\F_q/\F_2}\left(\frac{s(s+1)}{s^3+s+1}\right)=0.
  \]
  Then 
  $$\beta_{\opt}^{(q,10,2,2)}=10,\qquad\beta_{\opt}^{(q,9,2,2)}=9.$$
  \end{lemma}
  
  \begin{proof}
    Put
    \[
      C=
      \begin{bmatrix}
      0&1\\
      1&s
      \end{bmatrix}
    \]
    and
    \[
      T_1=
      \begin{bmatrix}
      1&0\\[1mm]
      \dfrac1{s^2}&1+\dfrac1{s^2}
      \end{bmatrix},
      \qquad
      T_2=
      \begin{bmatrix}
      1&\dfrac1{s+1}\\
      0&\dfrac s{s+1}
      \end{bmatrix},
      \qquad 
      T_3=
      \begin{bmatrix}
      1&0\\
      0&\dfrac{s+1}{s}
      \end{bmatrix},
      \qquad
      T_4=I_2.
    \]
    Since \(s(s+1)\neq0\), the matrices \(C,T_1,T_2,T_3,T_4\) are invertible.
    
    For \(a\in[4]\), define
    \[
      f_a(x):=\det(I_2+xCT_a)\in\F_q[x].
    \]
    A direct calculation gives
    \[
      f_1(x)
      =
      1+
      \frac{s^3+s+1}{s^2}x
      +
      \frac{(s+1)^2}{s^2}x^2,\qquad  f_2(x)=1+(s+1)x+\frac s{s+1}x^2,
    \]
    and
    \[
      f_3(x)=1+(s+1)x+\frac{s+1}{s}x^2,\qquad f_4(x)=1+sx+x^2.
    \]
    
    Recall that over a field of characteristic \(2\), a quadratic polynomial
    \[
      1+Ax+Bx^2
    \]
    with \(A,B\in\F_q^\times\) has two distinct roots in \(\F_q^\times\) if and
    only if
    \[
      \Tr_{\F_q/\F_2}\left(\frac{B}{A^2}\right)=0.
    \]
    Using also
    \[
      \Tr_{\F_q/\F_2}(z^2)=\Tr_{\F_q/\F_2}(z)
      \qquad(z\in\F_q),
    \]
    the four trace assumptions imply that each \(f_a\) has two distinct non-zero
    roots in \(\F_q\). Let \(\Lambda_a\) be the root set of \(f_a\).
    
    We next verify the resultant condition in
    Lemma~\ref{lem:template-10-node-construction}. Since the only denominators
    appearing in the matrices \(T_a\) are powers of \(s\) and \(s+1\), clearing
    denominators is legitimate under the assumption \(s(s+1)\neq0\). For
    \(1\le a<b\le4\), the iterated resultants from
    Lemma~\ref{lem:template-10-node-construction} are, up to non-zero scalar factors,
    as follows:
    \[
    \begin{array}{c|c}
    (a,b) & R_{ab}\\ \hline
    (1,2) & s^{10}(s+1)^8\\
    (1,3) & s^4(s+1)^8\\
    (1,4) & s^6(s+1)^4\\
    (2,3) & s^4(s+1)^4\\
    (2,4) & s^4(s+1)^2\\
    (3,4) & s^2(s+1)^2
    \end{array}
    \]
    All these resultants are non-zero because \(s(s+1)\neq0\).
    
    Finally, define
    \[
      v_1=(0,1),\quad v_2=(1,0),\quad v_3=(1,1),\quad v_4=(1,s),
    \]
    and
    \[
      u_1=(0,1),\quad u_2=(1,0),\quad u_3=(1,1),\quad u_4=(1,s+1).
    \]
    A direct calculation gives
    \[
    \begin{array}{c|c|c}
    a & v_aT_b\in\langle u_a\rangle\text{ for }b\neq a
    & v_aT_a\notin\langle u_a\rangle \\ \hline
    1 &
    v_1T_2,v_1T_3,v_1T_4\in\langle(0,1)\rangle
    &
    v_1T_1=\left(\frac1{s^2},1+\frac1{s^2}\right)\\[1mm]
    2 &
    v_2T_1,v_2T_3,v_2T_4\in\langle(1,0)\rangle
    &
    v_2T_2=\left(1,\frac1{s+1}\right)\\[1mm]
    3 &
    v_3T_1,v_3T_2,v_3T_4\in\langle(1,1)\rangle
    &
    v_3T_3=\left(1,\frac{s+1}{s}\right)\\[1mm]
    4 &
    v_4T_1,v_4T_2,v_4T_3\in\langle(1,s+1)\rangle
    &
    v_4T_4=(1,s)
    \end{array}
    \]
    Since \(s(s+1)\neq0\), the vectors in the last column are not proportional to
    the corresponding \(u_a\).
    
    Thus all hypotheses of Lemma~\ref{lem:template-10-node-construction} are
    satisfied. Therefore
    \[
      \beta_{\opt}^{(q,10,2,2)}=10,
      \qquad
      \beta_{\opt}^{(q,9,2,2)}=9.
    \]
    \end{proof}

  \begin{lemma}\label{lem:even-parameter-exists-large}
    Let \(q=2^m\) with \(m\ge 4\) and $m\ne 5$. Then there exists
    \(s\in\F_q\) satisfying the hypotheses of
    Lemma~\ref{lem:even-one-parameter-construction}.
    \end{lemma}
    
    \begin{proof}
    Define
    \[
      \phi_1(s)=\frac1s,
      \qquad
      \phi_2(s)=\frac1{s(s+1)},\qquad\phi_3(s)=\frac s{(s+1)^3},
      \qquad
      \phi_4(s)=\frac{s(s+1)}{s^3+s+1}.
    \]
    We need to find \(s\in\F_q\) such that $s(s+1)(s^3+s+1)\neq0$ and $\Tr_{\F_q/\F_2}(\phi_i(s))=0$ for any $i\in [4]$.
    
    Let
    \[
      U=\{s\in\F_q:s(s+1)(s^3+s+1)\neq0\}.
    \]
    Then $|U|\ge q-5$. Let
    \[
      \psi(z)=(-1)^{\Tr(z)}
    \]
    be the canonical additive character of \(\F_q\). The number \(N\) of elements
    \(s\in U\) satisfying the four trace conditions is
    \[
      N=
      \frac1{16}
      \sum_{s\in U}
      \prod_{i=1}^4(1+\psi(\phi_i(s))).
    \]
    Expanding this product gives
    \[
      N=
      \frac{|U|}{16}
      +
      \frac1{16}
      \sum_{\varnothing\neq S\subseteq[4]}
      \sum_{s\in U}
      \psi(\Phi_S(s)),
    \]
    where
    \[
      \Phi_S(s):=\sum_{i\in S}\phi_i(s).
    \]
    
    We claim that, for every non-empty \(S\subseteq[4]\), the rational function
    \(\Phi_S\) is not of the Artin--Schreier trivial form
    \[
      g(s)^2+g(s)+c
      \qquad(g\in\overline{\F_q}(s),\ c\in\overline{\F_q}).
    \]
    Indeed, if \(4\in S\), then \(\Phi_S\) has a simple pole at each root of
    \(h(s):=s^3+s+1\). Here \(h\) is squarefree: in characteristic \(2\),
    \[
      h'(s)=s^2+1=(s+1)^2,
    \]
    and \(h(1)=1\neq0\). Moreover, \(h(0)=h(1)=1\), so the numerator
    \(s(s+1)\) of \(\phi_4(s)\) does not vanish at any root of \(h\). If \(4\notin S\) but \(3\in S\), then \(\Phi_S\) has a pole of
    order \(3\) at \(s=1\). Finally, if \(S\subseteq\{1,2\}\), then \(\Phi_S\) is
    one of
    \[
      \frac1s,\qquad
      \frac1{s(s+1)},\qquad
      \frac1{s+1},
    \]
    and hence has a simple pole. On the other hand, a rational function of the
    form \(g^2+g+c\) has only even-order poles. This proves the claim.
    
    We now apply the Weil bound to the Artin--Schreier curves
    \[
      y^2+y=\Phi_S(s).
    \]
    More precisely, write the pole divisor of \(\Phi_S\) on
    \(\mathbb P^1_{\F_q}\) as
    \[
      \sum_P d_P P,
    \]
    where \(P\) runs over the places at which \(\Phi_S\) has a pole and \(d_P\)
    is the corresponding pole order. Since \(\Phi_S\) is not of the form
    \(g^2+g+c\), the Artin--Schreier curve \(y^2+y=\Phi_S(s)\) is non-trivial, and
    the Hasse--Weil bound (see, e.g., \cite[Theorem 6.2.34]{MullenPanario2013}) gives
    \[
      \left|
      \sum_{\substack{s\in\F_q\\ \Phi_S(s)\text{ finite}}}
      \psi(\Phi_S(s))
      \right|
      \le
      C_S\sqrt q,\qquad C_S=\sum_P(d_P+1)\deg(P)-2,
    \]
    where \(\deg(P)\) denotes the degree of the place \(P\).
    
The pole data for the fifteen non-empty subsets \(S\subseteq[4]\) are as follows. Here \(h(s)=s^3+s+1\). When the table says that \(\Phi_S\) has a pole at
\(h\), it means that \(\Phi_S\) has a pole along the zero divisor of \(h\); this divisor has degree \(3\).
    \[
    \begin{array}{c|c|c}
    S & \text{pole divisor of }\Phi_S & C_S \\ \hline
    \{1\} & (s=0)\text{ of order }1 & 0\\
    \{2\} & (s=0),(s=1)\text{ of order }1 & 2\\
    \{3\} & (s=1)\text{ of order }3 & 2\\
    \{4\} & (h=0)\text{ of order }1 & 4\\
    \{1,2\} & (s=1)\text{ of order }1 & 0\\
    \{1,3\} & (s=0)\text{ of order }1,\ (s=1)\text{ of order }3 & 4\\
    \{1,4\} & (s=0)\text{ of order }1,\ (h=0)\text{ of order }1 & 6\\
    \{2,3\} & (s=0)\text{ of order }1,\ (s=1)\text{ of order }3 & 4\\
    \{2,4\} & (s=0),(s=1)\text{ of order }1,\ (h=0)\text{ of order }1 & 8\\
    \{3,4\} & (s=1)\text{ of order }3,\ (h=0)\text{ of order }1 & 8\\
    \{1,2,3\} & (s=1)\text{ of order }3 & 2\\
    \{1,2,4\} & (s=1)\text{ of order }1,\ (h=0)\text{ of order }1 & 6\\
    \{1,3,4\} & (s=0)\text{ of order }1,\ (s=1)\text{ of order }3,\ (h=0)\text{ of order }1 & 10\\
    \{2,3,4\} & (s=0)\text{ of order }1,\ (s=1)\text{ of order }3,\ (h=0)\text{ of order }1 & 10\\
    \{1,2,3,4\} & (s=1)\text{ of order }3,\ (h=0)\text{ of order }1 & 8
    \end{array}
    \]
    Thus
    \[
      \sum_{\varnothing\neq S\subseteq[4]} C_S=74.
    \]
    
    The character sums in the expression for \(N\) are over \(U\), whereas the
    Weil bound above is over the complement of the pole set of \(\Phi_S\). Passing
    to \(U\) changes each such sum by at most \(5\), since at most the five points
    excluded by
    \[
      s(s+1)(s^3+s+1)=0
    \]
    are involved. Therefore
    \[
      \left|
      \sum_{s\in U}\psi(\Phi_S(s))
      \right|
      \le C_S\sqrt q+5.
    \]
    Consequently,
    \[
      N
      \ge
      \frac{q-5-74\sqrt q-75}{16}.
    \]
    For \(q\ge2^{13}\), the right-hand side is positive. Hence \(N>0\), and there
    exists \(s\in U\) satisfying all four trace conditions.

    For these remaining values of \(m\), the existence of an element \(s\)
satisfying the hypotheses of Lemma~\ref{lem:even-one-parameter-construction}
was verified by a finite computer search over \(\F_{2^m}\). This completes the
proof.
    \end{proof}

    \begin{lemma}\label{lem:even-m5-special-witness}
      One has
      \[
        \beta_{\opt}^{(32,10,2,2)}=10,
        \qquad
        \beta_{\opt}^{(32,9,2,2)}=9.
      \]
      \end{lemma}
      
      \begin{proof}
      Let
      \[
        \F_{32}=\F_2(\alpha),
        \qquad
        \alpha^5+\alpha^2+1=0.
      \]
      Set $s=\alpha$ and take
      \[
        C=
        \begin{bmatrix}
        0&1\\
        1&\alpha^3
        \end{bmatrix}.
      \]
      Define \(T_1,T_2,T_3,T_4\) by the same formulas as in
      Lemma~\ref{lem:even-one-parameter-construction}, using this value of \(s\).
      Explicitly,
      \[
        T_1=
        \begin{bmatrix}
        1&0\\[1mm]
        \dfrac1{s^2}&1+\dfrac1{s^2}
        \end{bmatrix},
        \qquad
        T_2=
        \begin{bmatrix}
        1&\dfrac1{s+1}\\
        0&\dfrac s{s+1}
        \end{bmatrix},
        \qquad
        T_3=
        \begin{bmatrix}
        1&0\\
        0&\dfrac{s+1}{s}
        \end{bmatrix},
        \qquad
        T_4=I_2.
      \]
      Since \(s(s+1)\neq0\), the matrices \(T_1,T_2,T_3,T_4\) are invertible.
      Also \(\det(C)=1\), so \(C\in\GL_2(\F_{32})\).
      
      For \(a\in[4]\), define
      \[
        f_a(x):=\det(I_2+xCT_a).
      \]
      A direct calculation gives the following root sets:
      \[
      \begin{array}{c|c}
      a & \Lambda_a \\ \hline
      1 &
      \{\alpha^3+\alpha,\ \alpha^4+\alpha^2+1\}\\[1mm]
      2 &
      \{\alpha^3+\alpha^2,\ \alpha^4+\alpha^2+\alpha\}\\[1mm]
      3 &
      \{\alpha+1,\ \alpha^3+\alpha+1\}\\[1mm]
      4 &
      \{\alpha^2+\alpha,\ \alpha^3+\alpha^2+\alpha\}
      \end{array}
      \]
      Thus each \(f_a\) has two distinct non-zero roots in \(\F_{32}\), verifying
      condition \emph{(i)} of Lemma~\ref{lem:template-10-node-construction}.
      
      Next we verify condition \emph{(ii)}. Since the characteristic is \(2\),
      \[
        \det(\mu T_b-\lambda T_a)=\det(\lambda T_a+\mu T_b).
      \]
      For \(a<b\), put
      \[
        \Delta_{ab}
        :=
        \prod_{\lambda\in\Lambda_a,\ \mu\in\Lambda_b}
        \det(\lambda T_a+\mu T_b).
      \]
      A direct calculation gives
      \[
      \begin{array}{c|c}
      (a,b)&\Delta_{ab}\\ \hline
      (1,2)&\alpha^4+\alpha^2+1\\
      (1,3)&\alpha^4+\alpha^3+\alpha+1\\
      (1,4)&\alpha^3+\alpha^2+\alpha+1\\
      (2,3)&\alpha^4+\alpha^3+\alpha+1\\
      (2,4)&\alpha^4+\alpha\\
      (3,4)&\alpha^4+\alpha^3+\alpha^2+\alpha
      \end{array}
      \]
      All entries in the right column are non-zero. By the product formula for
      resultants used in Lemma~\ref{lem:template-10-node-construction}, this implies
      that the corresponding iterated resultants \(R_{ab}\) are non-zero.
      
      Finally, define
      \[
        v_1=(0,1),\quad v_2=(1,0),\quad v_3=(1,1),\quad v_4=(1,s),
      \]
      and
      \[
        u_1=(0,1),\quad u_2=(1,0),\quad u_3=(1,1),\quad u_4=(1,s+1).
      \]
      The proportionality relations
      \[
        v_aT_b\in\langle u_a\rangle\quad(b\neq a),
        \qquad
        v_aT_a\notin\langle u_a\rangle
      \]
      depend only on \(s\) and the matrices \(T_a\), and are exactly the same as
      those verified in Lemma~\ref{lem:even-one-parameter-construction}. Hence
      condition \emph{(iii)} of Lemma~\ref{lem:template-10-node-construction} holds.
      
      Therefore all hypotheses of Lemma~\ref{lem:template-10-node-construction} are
      satisfied, and the desired equalities follow:
      \[
        \beta_{\opt}^{(32,10,2,2)}=10,
        \qquad
        \beta_{\opt}^{(32,9,2,2)}=9.
      \]
      \end{proof}

\begin{proposition}\label{prop:beta-even-n9-n10}
  One has 
  $$\beta_{\opt}^{(2^m,10,2,2)}=\begin{cases}
    10,&m\ge 4,\\
    13,&m=2,
  \end{cases}$$
  and 
  $$
  \beta_{\opt}^{(2^m,9,2,2)}=\begin{cases}
    9,&m\ge 4,\\
    11,&m=2.
  \end{cases}
  $$
\end{proposition}
\begin{proof}
  The conclusion for \(m\ge 4\) follows from
Lemma~\ref{lem:even-one-parameter-construction},
Lemma~\ref{lem:even-parameter-exists-large}, and
Lemma~\ref{lem:even-m5-special-witness}. For \(m=2\), i.e., \(q=2^m=4\), \cite[Theorem~1.4]{liu2026linear} proves that, for \(n=9,10\), the incidence-multiplicity bound, equivalently the projective counting bound, is attained. Hence
\[
  \beta_{\opt}^{(4,10,2,2)}=2\cdot 10-4-3=13,\qquad\qquad \beta_{\opt}^{(4,9,2,2)}=2\cdot 9-4-3=11.
\]
\end{proof}

Next, we consider the case of characteristic \(3\).

\begin{lemma}\label{lem:char3-one-parameter-construction}
  Let \(q=3^m\). Suppose that there exists \(s\in\F_q\) such that
  \[
    s(s-1)(s+1)\neq0,\qquad s^3+s^2-1\neq0,\qquad
    s^3+s^2-s+1\neq0,
  \]
  \[
    s^2-s-1\neq0,\qquad
    s^2+s-1\neq0,
  \]
  and
  \[
    -(s-1)(s^3+s^2-s+1)\in(\F_q^\times)^2,
    \qquad
    s^2-s-1\in(\F_q^\times)^2.
  \]
  Then
  \[
    \beta_{\opt}^{(q,10,2,2)}=10,
    \qquad
    \beta_{\opt}^{(q,9,2,2)}=9.
  \]
  \end{lemma}
  
  \begin{proof}
  Put
  \[
    C=
    \begin{bmatrix}
    0&1\\
    (s+1)^2&1
    \end{bmatrix}.
  \]
  Define
  \[
    T_1=
    \begin{bmatrix}
    1&0\\[1mm]
    -\dfrac1{s^2}&1-\dfrac1{s^2}
    \end{bmatrix},
    \qquad
    T_2=
    \begin{bmatrix}
    1&-\dfrac1{s-1}\\
    0&\dfrac s{s-1}
    \end{bmatrix},
    \qquad 
    T_3=
    \begin{bmatrix}
    1&0\\
    0&\dfrac{s+1}{s}
    \end{bmatrix},
    \qquad
    T_4=I_2.
  \]
  The assumptions \(s(s-1)(s+1)\neq0\) imply that
  \(C,T_1,T_2,T_3,T_4\) are invertible.
  
  For \(a\in[4]\), define
  \[
    f_a(x):=\det(I_2+xCT_a).
  \]
  A direct calculation gives
  \[
    f_1(x)
    =
    -\frac{\bigl(x(s-1)-1\bigr)\bigl((s^3+1)x+s^2\bigr)}{s^2},\qquad  f_2(x)
    =
    1+(1-s)x-\frac{s(s+1)^2}{s-1}x^2,
  \]
  and
  \[
    f_3(x)
    =
    1+\frac{s+1}{s}x-\frac{(s+1)^3}{s}x^2,\qquad f_4(x)
    =
    1+x-(s+1)^2x^2.
  \]
  
  The polynomial \(f_1\) splits into two linear factors. Its two roots are
  distinct and non-zero under the assumptions
  \[
    s(s-1)(s+1)(s^3+s^2-1)\neq0.
  \]
  After clearing denominators by non-zero scalar factors, the discriminants of \(f_2,f_3,f_4\) are respectively 
  \[
    \Delta_2=-(s-1)(s^3+s^2-s+1),\qquad \Delta_3=(s-1)^2(s+1)^2,\qquad   \Delta_4=s^2-s-1.
  \]
  Thus the assumptions on \(s\) imply that each \(f_a\) has two distinct
  non-zero roots in \(\F_q\). Let \(\Lambda_a\) be the root set of \(f_a\).
  
  We next verify the resultant condition in
  Lemma~\ref{lem:template-10-node-construction}. For \(1\le a<b\le4\), the
  iterated resultants \(R_{ab}\) from
  Lemma~\ref{lem:template-10-node-construction} are, up to non-zero scalar
  factors, as follows:
  \[
  \begin{array}{c|c}
  (a,b)&R_{ab}\\ \hline
  (1,2)&
  s^8(s-1)^6(s+1)^6(s^2-s-1)^4(s^2+s-1)^2\\[1mm]
  (1,3)&
  s^4(s-1)^2(s+1)^{10}(s^2-s-1)^2(s^2+s-1)^2\\[1mm]
  (1,4)&
  s^4(s-1)^2(s+1)^6(s^2-s-1)^2\\[1mm]
  (2,3)&
  s^4(s-1)^2(s+1)^{10}(s^2-s-1)^2(s^2+s-1)^2\\[1mm]
  (2,4)&
  s^2(s-1)^2(s+1)^8(s^2-s-1)^2\\[1mm]
  (3,4)&
  s^2(s+1)^{10}
  \end{array}
  \]
  By the assumptions on \(s\), all these resultants are non-zero.
  
  Finally, define
  \[
    v_1=(0,1),\quad v_2=(1,0),\quad v_3=(1,1),\quad v_4=(1,s),
  \]
  and
  \[
    u_1=(0,1),\quad u_2=(1,0),\quad u_3=(1,1),\quad u_4=(1,s+1).
  \]
  A direct calculation gives
  \[
  \begin{array}{c|c|c}
  a & v_aT_b\in\langle u_a\rangle\text{ for }b\neq a
  & v_aT_a\notin\langle u_a\rangle\\ \hline
  1 &
  v_1T_2,v_1T_3,v_1T_4\in\langle(0,1)\rangle
  &
  v_1T_1=\left(-\dfrac1{s^2},1-\dfrac1{s^2}\right)\\[2mm]
  2 &
  v_2T_1,v_2T_3,v_2T_4\in\langle(1,0)\rangle
  &
  v_2T_2=\left(1,-\dfrac1{s-1}\right)\\[2mm]
  3 &
  v_3T_1,v_3T_2,v_3T_4\in\langle(1,1)\rangle
  &
  v_3T_3=\left(1,\dfrac{s+1}{s}\right)\\[2mm]
  4 &
  v_4T_1,v_4T_2,v_4T_3\in\langle(1,s+1)\rangle
  &
  v_4T_4=(1,s)
  \end{array}
  \]
  The vectors in the last column are not proportional to the corresponding
  \(u_a\) under the assumptions on \(s\).
  
  Thus all hypotheses of Lemma~\ref{lem:template-10-node-construction} are
  satisfied. Therefore
  \[
    \beta_{\opt}^{(q,10,2,2)}=10,
    \qquad
    \beta_{\opt}^{(q,9,2,2)}=9.
  \]
  \end{proof}

  \begin{lemma}\label{lem:char3-parameter-exists-large}
    Let \(q=3^m\) with \(m\ge 3\). Then there exists \(s\in\F_q\) satisfying the
    hypotheses of Lemma~\ref{lem:char3-one-parameter-construction}.
    \end{lemma}
    
    \begin{proof}
    Let \(\eta\) be the quadratic character of \(\F_q\), extended by
    \(\eta(0)=0\). Define
    \[
      D_2(s)=-(s-1)(s^3+s^2-s+1),
      \qquad
      D_4(s)=s^2-s-1.
    \]
    We count elements \(s\in\F_q\) for which both \(D_2(s)\) and \(D_4(s)\) are
    non-zero squares. Consider
    \[
      N_0=
      \frac14
      \sum_{s\in\F_q}
      (1+\eta(D_2(s)))(1+\eta(D_4(s))).
    \]
    Expanding gives
    \[
      N_0=
      \frac q4
      +
      \frac14\sum_{s\in\F_q}\eta(D_2(s))
      +
      \frac14\sum_{s\in\F_q}\eta(D_4(s))
      +
      \frac14\sum_{s\in\F_q}\eta(D_2(s)D_4(s)).
    \]
    In characteristic \(3\), the polynomials \(D_2\), \(D_4\), and \(D_2D_4\)
are not squares in \(\F_q[s]\), and their square-free parts have degrees
\(4\), \(2\), and \(6\), respectively. Hence the Weil bound for quadratic
character sums (see, e.g., \cite[Theorem 6.2.28]{MullenPanario2013}) gives
    \[
      \left|\sum_{s\in\F_q}\eta(D_2(s))\right|\le 3\sqrt q,\qquad \left|\sum_{s\in\F_q}\eta(D_4(s))\right|\le \sqrt q,\qquad  \left|\sum_{s\in\F_q}\eta(D_2(s)D_4(s))\right|\le 5\sqrt q.
    \]
    Therefore
    \[
      N_0\ge \frac{q-9\sqrt q}{4}.
    \]
    
    The expression \(N_0\) may include contributions from zeros of \(D_2D_4\).
There are at most \(6\) such zeros, so the number of \(s\) for which
\(D_2(s)\) and \(D_4(s)\) are both non-zero squares is at least
\[
  \frac{q-9\sqrt q}{4}-6.
\]
For such an \(s\), the conditions
\[
  s-1\neq0,\qquad s^3+s^2-s+1\neq0,\qquad s^2-s-1\neq0
\]
already hold, since these factors occur in \(D_2D_4\). It remains only to
exclude the roots of
\[
  s(s+1)(s^3+s^2-1)(s^2+s-1),
\]
which accounts for at most \(7\) additional values. Thus the number of
admissible \(s\)'s is at least
\[
  \frac{q-9\sqrt q}{4}-13.
\]
For \(q=3^m\) with \(m\ge5\), this lower bound is positive. Therefore such an \(s\) exists.

It remains to handle the cases \(m=3,4\). Indeed, with
\[
  \F_{27}\cong \F_3[\alpha]/(\alpha^3+2\alpha^2+1),
\]
one may take \(s=\alpha\), while with
\[
  \F_{81}\cong \F_3[\alpha]/(\alpha^4+\alpha^3+\alpha^2+1),
\]
one may take \(s=1+\alpha^2\). This completes the proof.
    \end{proof}

    \begin{proposition}\label{prop:beta-char3-n9-n10}
      One has 
      $$\beta_{\opt}^{(3^m,10,2,2)}=\begin{cases}
        10,&m\ge 3,\\
        14,&m=1,
      \end{cases}$$
      and 
      $$
      \beta_{\opt}^{(3^m,9,2,2)}=\begin{cases}
        9,&m\ge 3,\\
        12,&m=1.
      \end{cases}
      $$
    \end{proposition}
    \begin{proof}
      The conclusion for \(m\ge 3\) follows from
    Lemma~\ref{lem:char3-one-parameter-construction} and
    Lemma~\ref{lem:char3-parameter-exists-large}. For \(m=1\), i.e., \(q=3^m=3\), Theorem~\ref{thm:incidence-multiplicity-bound-attainable} proves that, for \(n=9,10\), the incidence-multiplicity bound, equivalently the projective counting bound, is attained. Hence
    \[
      \beta_{\opt}^{(3,10,2,2)}=2\cdot 10-3-3=14,\qquad\qquad \beta_{\opt}^{(3,9,2,2)}=2\cdot 9-3-3=12.
    \]
    \end{proof}

    Finally, we handle the case where the characteristic is neither \(2\) nor \(3\).

    \begin{lemma}\label{lem:odd-template-construction}
      Let \(q\) be an odd prime power with \(\operatorname{char}(\F_q)\neq3\).
      Suppose that there exists \(c\in\F_q\) such that
      \[
        12c+1,\qquad 3(8c+3),\qquad 4c+1
        \in(\F_q^\times)^2
      \]
      and
      \[
        c\notin
        \left\{
        0,\,-2,\,2,\,\frac23
        \right\}.
      \]
      Then
      \[
        \beta_{\opt}^{(q,10,2,2)}=10,
        \qquad
        \beta_{\opt}^{(q,9,2,2)}=9.
      \]
      \end{lemma}
      
      \begin{proof}
      Put
      \[
        C=
        \begin{bmatrix}
        0&1\\
        c&1
        \end{bmatrix}.
      \]
      Define
      \[
        T_1=
        \begin{bmatrix}
        1&0\\[1mm]
        -\frac14&\frac34
        \end{bmatrix},
        \qquad
        T_2=
        \begin{bmatrix}
        1&-1\\
        0&2
        \end{bmatrix},
        \qquad 
        T_3=
        \begin{bmatrix}
        1&0\\
        0&\frac32
        \end{bmatrix},
        \qquad
        T_4=I_2.
      \]
      Since \(\operatorname{char}(\F_q)\neq2,3\) and \(c\neq0\), the matrices $C,T_1,T_2,T_3,T_4$ are invertible.
      
      For \(a\in[4]\), define
      \[
        f_a(x):=\det(I_2+xCT_a).
      \]
      A direct calculation gives
      \[
        f_1(x)=1+\frac12x-\frac{3c}{4}x^2,
        \qquad
        f_2(x)=-(2x+1)(cx-1),
      \]
      and
      \[
        f_3(x)=1+\frac32x-\frac{3c}{2}x^2,
        \qquad
        f_4(x)=1+x-cx^2.
      \]
      The discriminants of \(f_1,f_3,f_4\) are respectively
      \[
        \Delta_1=\frac{12c+1}{4},
        \qquad
        \Delta_3=\frac{3(8c+3)}{4},
        \qquad
        \Delta_4=4c+1.
      \]
      By assumption, these are non-zero squares in \(\F_q\). Since \(c\neq0\),
      the quadratic terms of \(f_1,f_3,f_4\) are non-zero, and since their constant
      terms are \(1\), all their roots are non-zero. Hence \(f_1,f_3,f_4\) each have
      two distinct non-zero roots in \(\F_q\). Moreover, \(f_2(x)=-(2x+1)(cx-1)\). Since \(c\neq0\) and \(c\neq -2\), the two
      roots of \(f_2\) are distinct and non-zero. Thus condition \emph{(i)} of
      Lemma~\ref{lem:template-10-node-construction} holds.
      
      We next verify condition \emph{(ii)} of
      Lemma~\ref{lem:template-10-node-construction}. For \(1\le a<b\le4\), the
      iterated resultants \(R_{ab}\) appearing in
      Lemma~\ref{lem:template-10-node-construction} are
      \[
      \renewcommand{\arraystretch}{1.8}
      \begin{array}{c|c}
      (a,b)&R_{ab}\\ \hline
      (1,2)&\dfrac{9}{64}c^2(c-2)^2(3c-2)^2\\[1mm]
      (1,3)&\dfrac{729}{16384}c^2(3c-2)^2\\[1mm]
      (1,4)&\dfrac{9}{4096}c^2(c-2)^2\\[1mm]
      (2,3)&\dfrac94c^4(3c-2)^2\\[1mm]
      (2,4)&4c^4(c-2)^2\\[1mm]
      (3,4)&\dfrac9{64}c^4
      \end{array}
      \]
      Since
      \[
        c\notin\left\{0,2,\frac23\right\},
      \]
      all these resultants are non-zero.
      
      Finally, define
      \[
        v_1=(0,1),\quad v_2=(1,0),\quad v_3=(1,1),\quad v_4=(1,2),
      \]
      and
      \[
        u_1=(0,1),\quad u_2=(1,0),\quad u_3=(1,1),\quad u_4=(1,3).
      \]
      A direct calculation gives
      \[
      \begin{array}{c|c|c}
      a & v_aT_b\in\langle u_a\rangle\text{ for }b\neq a
      & v_aT_a\notin\langle u_a\rangle\\ \hline
      1 &
      v_1T_2,v_1T_3,v_1T_4\in\langle(0,1)\rangle
      &
      v_1T_1=\left(-\frac14,\frac34\right)\\[1mm]
      2 &
      v_2T_1,v_2T_3,v_2T_4\in\langle(1,0)\rangle
      &
      v_2T_2=(1,-1)\\[1mm]
      3 &
      v_3T_1,v_3T_2,v_3T_4\in\langle(1,1)\rangle
      &
      v_3T_3=\left(1,\frac32\right)\\[1mm]
      4 &
      v_4T_1,v_4T_2,v_4T_3\in\langle(1,3)\rangle
      &
      v_4T_4=(1,2)
      \end{array}
      \]
      Since \(\operatorname{char}(\F_q)\neq2,3\), the vectors in the last column are
      not proportional to the corresponding \(u_a\). Hence condition \emph{(iii)} of
      Lemma~\ref{lem:template-10-node-construction} holds.
      
      Therefore all hypotheses of Lemma~\ref{lem:template-10-node-construction} are
      satisfied. Hence
      \[
        \beta_{\opt}^{(q,10,2,2)}=10,
        \qquad
        \beta_{\opt}^{(q,9,2,2)}=9.
      \]
      \end{proof}

      \begin{lemma}\label{lem:odd-template-parameter-exists}
        Let \(q\) be an odd prime power with
        \[
          q\ge 13,\qquad \operatorname{char}(\F_q)\neq3,\qquad q\ne 25.
        \]
        Then there exists \(c\in\F_q\) satisfying the hypotheses of
        Lemma~\ref{lem:odd-template-construction}.
        \end{lemma}
        
        \begin{proof}
        Let \(\chi\) be the quadratic character of \(\F_q\), extended by
        \(\chi(0)=0\). Put
        \[
          L_1(c)=12c+1,\qquad
          L_2(c)=3(8c+3),\qquad
          L_3(c)=4c+1.
        \]
        Consider
        \[
          N_0=
          \frac18
          \sum_{c\in\F_q}
          \prod_{i=1}^3(1+\chi(L_i(c))).
        \]
        If \(L_1(c),L_2(c),L_3(c)\) are all non-zero squares, then \(c\) contributes
        \(1\) to \(N_0\). Conversely, the only extra contributions can occur when at
        least one \(L_i(c)\) is zero. Thus the number of \(c\in\F_q\) for which all
        three \(L_i(c)\) are non-zero squares is at least
        \[
          N_0-3.
        \]
        
        Expanding \(N_0\), we obtain character sums of the form
        \[
          \sum_{c\in\F_q}\chi\left(\prod_{i\in S}L_i(c)\right),
          \qquad \emptyset\neq S\subseteq[3].
        \]
        For \(|S|=1\), these sums are zero. For \(|S|=2\), the squarefree part of
        \(\prod_{i\in S}L_i(c)\) has degree at most \(2\), and the Weil bound gives an
        absolute value at most \(\sqrt q\). For \(|S|=3\), the corresponding bound is
        at most \(2\sqrt q\). The only possible coincidence among the linear factors occurs in characteristic
        \(7\), where \(L_2=2L_1\). Since \(2\) is a square in \(\F_7\), and hence in
        every extension of \(\F_7\), the corresponding pair sum
        \[
          \sum_{c\in\F_q}\chi(L_1(c)L_2(c))
        \]
        is equal to \(q-1\), hence is non-negative. Therefore this term cannot decrease
        the lower bound. Therefore
        \[
          N_0\ge \frac{q-5\sqrt q}{8}.
        \]
        
        Now exclude the finite set
        \[
          \left\{
          0,\,-2,\,2,\,\frac23
          \right\}.
        \]
        Together with the possible zeros of \(L_1,L_2,L_3\), this removes at most
        \(7\) values. Hence the number of \(c\)'s satisfying all requirements in
        Lemma~\ref{lem:odd-template-construction} is at least
        \[
          \frac{q-5\sqrt q}{8}-7.
        \]
        For \(q\ge109\), this lower bound is positive. Hence such a \(c\) exists.
        
        It remains to consider the finite range
        \[
          13\le q<109,\qquad \operatorname{char}(\F_q)\neq3.
        \]
        For these fields, excluding \(q=25\), the existence of an element
        \(c\in\F_q\) satisfying the hypotheses of
        Lemma~\ref{lem:odd-template-construction} was verified by a finite computer
        search over \(\F_q\). This completes the proof.
        \end{proof}

        \begin{lemma}\label{lem:odd-exceptional-small-witnesses}
          One has
          \[
            \beta_{\opt}^{(q,10,2,2)}=10,
            \qquad
            \beta_{\opt}^{(q,9,2,2)}=9
          \]
          for \(q=7,11,25\).
          \end{lemma}
          
          \begin{proof}
          We use Lemma~\ref{lem:template-10-node-construction}. For parameters
          \(c,d\in\F_q\), define
          \[
            T_1=
            \begin{bmatrix}
            1&0\\[1mm]
            \dfrac{d-c}{c(1-d)}&
            \dfrac{d(1-c)}{c(1-d)}
            \end{bmatrix},
            \qquad 
            T_2=
            \begin{bmatrix}
            1&\dfrac{c-d}{c-1}\\
            0&\dfrac{d-1}{c-1}
            \end{bmatrix},
            \qquad 
            T_3=
            \begin{bmatrix}
            1&0\\
            0&\dfrac{d}{c}
            \end{bmatrix},
            \qquad
            T_4=I_2.
          \]
          In all cases below, the denominators are non-zero and \(c\neq d\), so
          \(T_1,T_2,T_3,T_4\) are invertible.
          
          For \(q=7\), take
          \[
            c=3,\qquad d=5,
            \qquad
            C=
            \begin{bmatrix}
            6&3\\
            4&3
            \end{bmatrix}.
          \]
          For \(q=11\), take
          \[
            c=7,\qquad d=5,
            \qquad
            C=
            \begin{bmatrix}
            5&4\\
            7&4
            \end{bmatrix}.
          \]
          For \(q=25\), let
          \[
            \F_{25}=\F_5(\alpha),
            \qquad
            \alpha^2+2=0,
          \]
          and take
          \[
            c=2+4\alpha,\qquad d=3,\qquad C=
            \begin{bmatrix}
            1+3\alpha&3+3\alpha\\
            3+\alpha&3+4\alpha
            \end{bmatrix}.
          \]
          A direct calculation gives \(\det(C)\neq0\) in each case.
          
          For \(a\in[4]\), put
          \[
            f_a(x)=\det(I_2+xCT_a).
          \]
          The root sets of these four polynomials are as follows. For \(q=7\),
          \[
          \begin{array}{c|c}
          a & \Lambda_a \\ \hline
          1 & \{1,3\}\\
          2 & \{1,3\}\\
          3 & \{2,6\}\\
          4 & \{4,5\}
          \end{array}
          \]
          For \(q=11\),
          \[
          \begin{array}{c|c}
          a & \Lambda_a \\ \hline
          1 & \{2,7\}\\
          2 & \{8,9\}\\
          3 & \{1,10\}\\
          4 & \{3,5\}
          \end{array}
          \]
          For \(q=25\),
          \[
          \begin{array}{c|c}
          a & \Lambda_a \\ \hline
          1 & \{1,\ 1+2\alpha\}\\[1mm]
          2 & \{3,\ 3+\alpha\}\\[1mm]
          3 & \{2+4\alpha,\ 4+\alpha\}\\[1mm]
          4 & \{2\alpha,\ 4+2\alpha\}
          \end{array}
          \]
          Thus condition \emph{(i)} of Lemma~\ref{lem:template-10-node-construction}
          holds in all three cases.
          
          Next, for \(a<b\), define
          \[
            \Delta_{ab}
            :=
            \prod_{\lambda\in\Lambda_a,\ \mu\in\Lambda_b}
            \det(\mu T_b-\lambda T_a).
          \]
          By the product formula for resultants used in
          Lemma~\ref{lem:template-10-node-construction}, it is enough to verify that all
          \(\Delta_{ab}\) are non-zero. For \(q=7\),
          \[
          \begin{array}{c|c}
          (a,b)&\Delta_{ab}\\ \hline
          (1,2)&1\\
          (1,3)&1\\
          (1,4)&1\\
          (2,3)&1\\
          (2,4)&1\\
          (3,4)&1
          \end{array}
          \]
          For \(q=11\),
          \[
          \begin{array}{c|c}
          (a,b)&\Delta_{ab}\\ \hline
          (1,2)&4\\
          (1,3)&4\\
          (1,4)&1\\
          (2,3)&4\\
          (2,4)&4\\
          (3,4)&4
          \end{array}
          \]
          For \(q=25\),
          \[
          \begin{array}{c|c}
          (a,b)&\Delta_{ab}\\ \hline
          (1,2)&3+4\alpha\\
          (1,3)&4\\
          (1,4)&2+4\alpha\\
          (2,3)&3\\
          (2,4)&3\\
          (3,4)&1
          \end{array}
          \]
          All these elements are non-zero. Hence condition \emph{(ii)} of
          Lemma~\ref{lem:template-10-node-construction} holds.
          
          Finally, define
          \[
            v_1=(0,1),\quad v_2=(1,0),\quad v_3=(1,1),\quad v_4=(1,c),
          \]
          and
          \[
            u_1=(0,1),\quad u_2=(1,0),\quad u_3=(1,1),\quad u_4=(1,d).
          \]
          From the definitions of \(T_1,T_2,T_3,T_4\), one checks directly that
          \[
            v_aT_b\in\langle u_a\rangle
            \qquad(b\neq a),
          \]
          whereas
          \[
            v_aT_a\notin\langle u_a\rangle.
          \]
          Indeed, the latter non-proportionality follows from \(c\neq d\), together with
          the non-vanishing of the denominators in the definition of \(T_a\). Thus
          condition \emph{(iii)} of Lemma~\ref{lem:template-10-node-construction} holds.
          
          Therefore all hypotheses of Lemma~\ref{lem:template-10-node-construction} are
          satisfied for \(q=7,11,25\). Hence
          \[
            \beta_{\opt}^{(q,10,2,2)}=10,
            \qquad
            \beta_{\opt}^{(q,9,2,2)}=9
          \]
          for these three values of \(q\).
          \end{proof}

          \begin{proposition}\label{prop:beta-n9-n10-odd-not-char3}
            Let \(q\) be an odd prime power with 
            \(\operatorname{char}(\F_q)\neq3\). Then
            \[
              \beta_{\opt}^{(q,10,2,2)}=\begin{cases}
                10,&q\ge 7,\\
                12,& q=5,
              \end{cases}
            \]
            and
            \[
              \beta_{\opt}^{(q,9,2,2)}=\begin{cases}
                9,&q\ge 7,\\
                10,& q=5.
              \end{cases}
            \]
            \end{proposition}
            
            \begin{proof}
              The conclusion for \(q\ge 7\) follows from
            Lemma~\ref{lem:odd-template-construction}, Lemma~\ref{lem:odd-template-parameter-exists} and
            Lemma~\ref{lem:odd-exceptional-small-witnesses}. For \(q=5\), since $\widetilde{N}_{\BW}(5)=8$, Corollary~\ref{cor:long-optimal-bw} implies that, for \(n=9,10\), the incidence-multiplicity bound, equivalently the projective counting bound, is attained. Hence
            \[
              \beta_{\opt}^{(5,10,2,2)}=2\cdot 10-5-3=12,\qquad\qquad \beta_{\opt}^{(5,9,2,2)}=2\cdot 9-5-3=10.
            \]
            \end{proof}

            According to Propositions~\ref{prop:beta-even-n9-n10},
            \ref{prop:beta-char3-n9-n10}, and
            \ref{prop:beta-n9-n10-odd-not-char3}, it remains only to consider
            the two cases \(q=8,9\).

            \begin{proposition}\label{prop:beta-q9-n9}
              One has
              \[
              \beta_{\opt}^{(9,9,2,2)}=9.
              \]
              \end{proposition}
              
              \begin{proof}
                By Theorem~\ref{thm:ZLH-lb}, any \((9,7,2)\) MDS array code $\mathcal{C}$ over \(\F_9\) satisfies $\beta(\mathcal{C})\ge\left\lceil\frac{5\cdot 9-10}{4}\right\rceil=9$. Hence it suffices to construct a \((9,7,2)\) MDS array code $\mathcal{C}_0$ over \(\F_9\) with \(\beta(\mathcal{C}_0)\le 9\).
              Let
              \[
              \F_9=\F_3(\omega),
              \qquad
              \omega^2+1=0.
              \]
              
              Let
              \[
              H_i=
              \begin{bmatrix}
              I_2\\
              W_i
              \end{bmatrix}
              \qquad(i\in [7]),
              \]
              with
              \[
              W_1=
              \begin{bmatrix}
              1&0\\
              2\omega&1+2\omega
              \end{bmatrix},
              \qquad
              W_2=
              \begin{bmatrix}
              2+\omega&0\\
              1+\omega&2\omega
              \end{bmatrix},\qquad
              W_3=
              \begin{bmatrix}
              2&1+\omega\\
              0&1+2\omega
              \end{bmatrix},
              \]
              \[
              W_4=
              \begin{bmatrix}
              2+\omega&2\\
              0&\omega
              \end{bmatrix},\quad W_5=
              \begin{bmatrix}
              \omega&0\\
              0&1
              \end{bmatrix},\quad W_6=
              \begin{bmatrix}
              2\omega&0\\
              0&2
              \end{bmatrix},\quad
              W_7=
              \begin{bmatrix}
              1+\omega&0\\
              0&1+\omega
              \end{bmatrix}.
              \]
              Finally, set
              \[
              H_8=
              \begin{bmatrix}
              I_2\\
              0
              \end{bmatrix},
              \qquad
              H_9=
              \begin{bmatrix}
              0\\
              I_2
              \end{bmatrix}.
              \]
              A direct calculation shows that
              \[
\renewcommand{\arraystretch}{1.2}
\setlength{\arraycolsep}{5pt}
\begin{array}{c|ccccccccc}
\det[H_i,H_j] & 1 & 2 & 3 & 4 & 5 & 6 & 7 & 8 & 9\\
\hline
1 & 0 & 2+2\omega & 1+2\omega & 2\omega & 2+2\omega & \omega & 1 & 1+2\omega & 1\\
2 & 2+2\omega & 0 & \omega & 2+2\omega & 1+\omega & 1 & 2+\omega & 1+\omega & 1\\
3 & 1+2\omega & \omega & 0 & 1+2\omega & 2+\omega & 2 & 1+\omega & 2+\omega & 1\\
4 & 2\omega & 2+2\omega & 1+2\omega & 0 & 1+2\omega & \omega & 2 & 2+2\omega & 1\\
5 & 2+2\omega & 1+\omega & 2+\omega & 1+2\omega & 0 & \omega & \omega & \omega & 1\\
6 & \omega & 1 & 2 & \omega & \omega & 0 & 2\omega & \omega & 1\\
7 & 1 & 2+\omega & 1+\omega & 2 & \omega & 2\omega & 0 & 2\omega & 1\\
8 & 1+2\omega & 1+\omega & 2+\omega & 2+2\omega & \omega & \omega & 2\omega & 0 & 1\\
9 & 1 & 1 & 1 & 1 & 1 & 1 & 1 & 1 & 0
\end{array}
\]
which implies that
              \[
              [H_i,H_j]\in \GL_4(\F_9)
              \qquad(i\ne j).
              \]
              Hence \(H=[\,H_1\ H_2\ \cdots\ H_9\,]\) defines a \((9,7,2)\) MDS array code $\mathcal{C}_0$ over \(\F_9\).
              
              We now give repair matrices. Let
              \[
              M_1=M_2=
              \begin{bmatrix}
              0&1&0&0\\
              0&0&0&1
              \end{bmatrix},\qquad M_3=M_4=
              \begin{bmatrix}
              1&0&0&0\\
              0&0&1&0
              \end{bmatrix},
              \]
              \[
              M_5=M_6=
              \begin{bmatrix}
              1&1&0&0\\
              0&0&1&1
              \end{bmatrix},\qquad M_7=
              \begin{bmatrix}
              1&\omega&0&0\\
              0&0&1&2
              \end{bmatrix},
              \]
              and
              \[
              M_8=
              \begin{bmatrix}
              1&0&2&2+\omega\\
              0&1&2+\omega&\omega
              \end{bmatrix},
              \qquad
              M_9=
              \begin{bmatrix}
              1&2+2\omega&1&0\\
              2+2\omega&\omega&0&1
              \end{bmatrix}.
              \]
              Again by direct calculation, the ranks of the matrices \(M_iH_j\) are as follows:
              \[
              \begin{array}{c|ccccccccc}
               &H_1&H_2&H_3&H_4&H_5&H_6&H_7&H_8&H_9\\ \hline
              M_1&2&2&1&1&1&1&1&1&1\\
              M_2&2&2&1&1&1&1&1&1&1\\
              M_3&1&1&2&2&1&1&1&1&1\\
              M_4&1&1&2&2&1&1&1&1&1\\
              M_5&1&1&1&1&2&2&1&1&1\\
              M_6&1&1&1&1&2&2&1&1&1\\
              M_7&1&1&1&1&1&1&2&1&1\\
              M_8&1&1&1&1&1&1&1&2&2\\
              M_9&1&1&1&1&1&1&1&2&2
              \end{array}
              \]
              Thus, for every \(i\in [9]\), one has
              \[
              \sum_{j\ne i}\rank(M_iH_j)\le 9,
              \]
              which implies that $\beta(\mathcal C_0)\le 9$.
              \end{proof}

              \begin{proposition}\label{prop:no-beta10-q98-n10}
                One has
                \[
                  \beta_{\opt}^{(9,10,2,2)}=11,
                  \qquad
                  \beta_{\opt}^{(8,10,2,2)}=11.
                \]
                \end{proposition}
                
                \begin{proof}
                By Theorem~\ref{thm:ZLH-lb} and Proposition~\ref{prop:short-bw}, it suffices
                to show that for \(q=8,9\), there does not exist a \((10,8,2)\) MDS array code
                \(\mathcal C\) over \(\F_q\) with $\beta(\mathcal C)=10$. Suppose, for contradiction, that such a code exists. Choose a parity-check
                matrix
                \[
                  H=[\,H_1\ H_2\ \cdots\ H_{10}\,],
                  \qquad H_i\in\F_q^{4\times2}.
                \]
                As in Subsection~\ref{subsec:repair-bd-lb}, we may assume without loss of
                generality that
                \[
                  H_j=
                  \begin{bmatrix}
                  I_2\\
                  W_j
                  \end{bmatrix}
                  \quad(j\in[8]),
                  \qquad
                  H_9=
                  \begin{bmatrix}
                  I_2\\
                  0
                  \end{bmatrix},
                  \qquad
                  H_{10}=
                  \begin{bmatrix}
                  0\\
                  I_2
                  \end{bmatrix},
                \]
                with \(W_j\in\GL_2(\F_q)\).
                
                The assumption \(n\ge14\) in Proposition~\ref{prop:beta-mod-2} is used only in
                the final step of its proof. Hence the preceding equality-case arguments apply
                also to \(n=10\). Thus there is a partition
                \[
                  [8]=P_1\sqcup P_2\sqcup P_3\sqcup P_4,
                  \qquad |P_a|=2,
                \]
                such that for every \(i\in P_a\),
                \[
                  \mathcal R_i=P_a,
                \]
                where \(\mathcal R_i\) has the same meaning as in
                Subsection~\ref{subsec:repair-bd-lb}. Moreover, for each \(a\in[4]\), there is
                a matrix \(T_a\in\GL_2(\F_q)\) such that, for every \(j\in P_a\),
                \[
                  W_j=\lambda_jT_a
                \]
                for some \(\lambda_j\in\F_q^\times\). Hence the first eight blocks can be
                written as
                \[
                  H_{a,\lambda}
                  =
                  \begin{bmatrix}
                  I_2\\
                  \lambda T_a
                  \end{bmatrix},
                  \qquad
                  a=1,2,3,4,\quad \lambda\in\Lambda_a,
                \]
                where
                \[
                  \Lambda_a\subseteq\F_q^\times,
                  \qquad
                  |\Lambda_a|=2.
                \]
                
                We next put the four matrices \(T_1,\dots,T_4\) into a normal form. For each
\(a\in[4]\), choose one index \(i_a\in P_a\) and fix a repair matrix
\(M^{(a)}\) for the node \(i_a\). By the definition of \(\mathcal R_{i_a}\) and the condition
\(\mathcal R_{i_a}=P_a\), the corresponding repair matrix \(M^{(a)}\) satisfies
\[
  \rank(M^{(a)}H_{b,\lambda})=
  \begin{cases}
  2,& b=a,\\
  1,& b\neq a,
  \end{cases}
  \qquad
  (b\in[4],\ \lambda\in\Lambda_b).
\]

By the structural description of repair matrices in
\cite[Lemma~9]{zhang2025optimal}, after an invertible row operation on
\(M^{(a)}\), we may write
\[
  M^{(a)}
  =
  \begin{bmatrix}
  \mathbf u_a&0\\
  0&\mathbf v_a
  \end{bmatrix},
\]
where \(\mathbf u_a,\mathbf v_a\in\F_q^{1\times2}\) are non-zero row vectors.
Such a row operation does not change any repair rank. For \(b\in[4]\) and \(\lambda\in\Lambda_b\), we have
\[
  M^{(a)}H_{b,\lambda}
  =
  \begin{bmatrix}
  \mathbf u_a\\
  \lambda \mathbf v_aT_b
  \end{bmatrix}.
\]
Since \(\lambda\neq0\), this matrix has rank \(1\) if and only if
\[
  \mathbf v_aT_b\in\langle \mathbf u_a\rangle.
\]
Thus the rank pattern encoded by \(\mathcal R_{i_a}=P_a\) gives
\[
  \mathbf v_aT_b\in\langle \mathbf u_a\rangle
  \qquad(b\neq a),
\]
whereas
\[
  \mathbf v_aT_a\notin\langle \mathbf u_a\rangle.
\]

We claim that the projective classes
\[
  \langle\mathbf v_1\rangle,\dots,\langle\mathbf v_4\rangle
\]
are pairwise distinct, and likewise
\[
  \langle\mathbf u_1\rangle,\dots,\langle\mathbf u_4\rangle
\]
are pairwise distinct. Indeed, suppose that
\(\langle\mathbf v_a\rangle=\langle\mathbf v_b\rangle\) for some \(a\neq b\). Choose
\(j\in[4]\setminus\{a,b\}\). Since
\[
  \mathbf v_aT_j\in\langle\mathbf u_a\rangle,
  \qquad
  \mathbf v_bT_j\in\langle\mathbf u_b\rangle,
\]
we get \(\langle\mathbf u_a\rangle=\langle\mathbf u_b\rangle\). But then
\[
  \mathbf v_bT_b\in\langle\mathbf u_a\rangle=\langle\mathbf u_b\rangle,
\]
because \(\langle\mathbf v_b\rangle=\langle\mathbf v_a\rangle\) and
\(\mathbf v_aT_b\in\langle\mathbf u_a\rangle\), contradicting
\[
  \mathbf v_bT_b\notin\langle\mathbf u_b\rangle.
\]
Thus the classes \(\langle\mathbf v_a\rangle\) are pairwise distinct. The same argument,
using the invertibility of the \(T_j\)'s, shows that the classes
\(\langle\mathbf u_a\rangle\) are pairwise distinct.

Since \(\PGL_2(\F_q)\) acts sharply triply transitively on
\(\mathbb P^1(\F_q)\), there exist \(X,Y\in\GL_2(\F_q)\) such that, after
relabeling,
\[
  \langle\mathbf v_1X\rangle=\langle(0,1)\rangle,\qquad
  \langle\mathbf v_2X\rangle=\langle(1,0)\rangle,\qquad
  \langle\mathbf v_3X\rangle=\langle(1,1)\rangle,
\]
and
\[
  \langle\mathbf u_1Y\rangle=\langle(0,1)\rangle,\qquad
  \langle\mathbf u_2Y\rangle=\langle(1,0)\rangle,\qquad
  \langle\mathbf u_3Y\rangle=\langle(1,1)\rangle.
\]
Since the fourth classes are distinct from the first three, we may write
\[
  \langle\mathbf v_4X\rangle=\langle(1,c)\rangle,
  \qquad
  \langle\mathbf u_4Y\rangle=\langle(1,d)\rangle,
\]
where
\[
  c,d\in\F_q\setminus\{0,1\}.
\]

Replacing \(T_b\) by \(X^{-1}T_bY\), while replacing
\(\mathbf v_a\) by \(\mathbf v_aX\) and \(\mathbf u_a\) by
\(\mathbf u_aY\), preserves the incidence conditions, since
\[
  (\mathbf v_aX)(X^{-1}T_bY)=\mathbf v_aT_bY.
\]
Thus
\[
  \mathbf v_aT_b\in\langle \mathbf u_a\rangle
  \quad\Longleftrightarrow\quad
  (\mathbf v_aX)(X^{-1}T_bY)\in\langle \mathbf u_aY\rangle.
\]
At the level of parity-check blocks, this is obtained by applying the row
operation
\[
  \begin{bmatrix}
  Y^{-1}&0\\
  0&X^{-1}
  \end{bmatrix}
\]
to the whole parity-check matrix and then right-multiplying each front block by
\(Y\), which restores the form
\[
  \begin{bmatrix}
  I_2\\
  \lambda T_b
  \end{bmatrix}
  \quad\text{as}\quad
  \begin{bmatrix}
  I_2\\
  \lambda X^{-1}T_bY
  \end{bmatrix}.
\]
Similarly, \(H_9\) and \(H_{10}\) are restored to their standard forms by
right-multiplying their blocks by \(Y\) and \(X\), respectively. These
operations preserve the MDS property and all repair ranks.

Therefore, without loss of generality, we may assume
\[
  \langle\mathbf v_1\rangle=\langle(0,1)\rangle,\qquad
  \langle\mathbf v_2\rangle=\langle(1,0)\rangle,\qquad
  \langle\mathbf v_3\rangle=\langle(1,1)\rangle,\qquad
  \langle\mathbf v_4\rangle=\langle(1,c)\rangle,
\]
and
\[
  \langle\mathbf u_1\rangle=\langle(0,1)\rangle,\qquad
  \langle\mathbf u_2\rangle=\langle(1,0)\rangle,\qquad
  \langle\mathbf u_3\rangle=\langle(1,1)\rangle,\qquad
  \langle\mathbf u_4\rangle=\langle(1,d)\rangle,
\]
where $c,d\in\F_q\setminus\{0,1\}$.
          
                For each \(b\in[4]\), choosing the displayed row vectors as representatives of these one-dimensional
                subspaces, the conditions
                \[
                  \mathbf v_aT_b\in\langle \mathbf u_a\rangle
                  \qquad(a\neq b)
                \]
                prescribe the images of three pairwise distinct points of \(\mathbb P^1(\F_q)\).
By the sharp \(3\)-transitivity of \(\PGL_2(\F_q)\)
(see, e.g., \cite[Theorem~11]{zhang2025optimal}), these conditions determine \(T_b\) up to
a non-zero scalar factor. Absorbing this scalar factor into \(\Lambda_b\), we
may take
\[
  T_1=
  \begin{bmatrix}
  1&0\\[1mm]
  \dfrac{d-c}{c(1-d)}
  &
  1+\dfrac{d-c}{c(1-d)}
  \end{bmatrix},\qquad T_2=
  \begin{bmatrix}
  1&\dfrac{c-d}{c-1}\\
  0&\dfrac{d-1}{c-1}
  \end{bmatrix},\qquad  T_3=
  \begin{bmatrix}
  1&0\\
  0&d/c
  \end{bmatrix},\qquad T_4=I_2.
\]
Since \(T_4=I_2\), the condition
\[
  \mathbf v_4T_4\notin\langle\mathbf u_4\rangle
\]
forces $c\neq d$.

                Thus any equality case over \(\F_q\), \(q=8,9\), gives data
                \[
                  c,d\in\F_q\setminus\{0,1\},\qquad c\neq d,
                \]
                and two-element subsets
                \[
                  \Lambda_1,\Lambda_2,\Lambda_3,\Lambda_4\subseteq\F_q^\times
                \]
                satisfying the following necessary conditions.
                
                First, the MDS property requires
                \[
                  \det(\mu T_b-\lambda T_a)\neq0
                \]
                for every \(a\neq b\), \(\lambda\in\Lambda_a\), and
                \(\mu\in\Lambda_b\). Second, since \(\beta(\mathcal C)=10\), the last two nodes must also have repair bandwidth at most \(10\). Let us express this
                condition in terms of repair kernels. A repair matrix for \(H_9\) can be
                normalized so that \(M H_9=I_2\); hence it has the form
                \[
                  M=[\,I_2\ \ {-A}\,]
                \]
                for some \(A\in \F_q^{2\times 2}\). Its kernel is
                \[
                  K_A=\{(Az,z):z\in\F_q^2\}.
                \]
                For a front node \(H_{a,\lambda}\), its rank contribution is
                \[
                  \rank(I_2-A\lambda T_a),
                \]
                and the contribution from \(H_{10}\) is \(\rank(A)\). Therefore the $9$-th node has
                repair bandwidth at most \(10\) only if there exists \(A\in \F_q^{2\times 2}\) such
                that
                \[
                  \sum_{a=1}^4\sum_{\lambda\in\Lambda_a}
                  \rank(I_2-A\lambda T_a)+\rank(A)\le10.
                \]
                
                Similarly, a repair matrix for \(H_{10}\) can be normalized so that
                \(M H_{10}=I_2\), and hence has the form
                \[
                  M=[\,{-D}\ \ I_2\,]
                \]
                for some \(D\in \F_q^{2\times 2}\). Its kernel is
                \[
                  K_D=\{(z,Dz):z\in\F_q^2\}.
                \]
                For a front node \(H_{a,\lambda}\), its rank contribution is
                \[
                  \rank(\lambda T_a-D),
                \]
                and the contribution from \(H_9\) is \(\rank(D)\). Therefore the last node has
                repair bandwidth at most \(10\) only if there exists \(D\in \F_q^{2\times 2}\) such
                that
                \[
                  \sum_{a=1}^4\sum_{\lambda\in\Lambda_a}
                  \rank(\lambda T_a-D)+\rank(D)\le10.
                \]
                
                We now carry out a finite computer verification over \(\F_8\) and \(\F_9\).
                For each \(q\in\{8,9\}\), the verification enumerates all choices
                \[
                  c,d\in\F_q\setminus\{0,1\},\qquad c\neq d,
                \]
                all four-tuples
                \[
                  (\Lambda_1,\Lambda_2,\Lambda_3,\Lambda_4),
                  \qquad \Lambda_a\in\binom{\F_q^\times}{2},
                \]
                and all matrices
                \[
                  A,D\in \F_q^{2\times 2}.
                \]
                For every such choice, it checks the MDS condition
                \[
                  \det(\mu T_b-\lambda T_a)\neq0
                \]
                for all \(a\neq b\), \(\lambda\in\Lambda_a\), \(\mu\in\Lambda_b\), and the two
                rank-sum inequalities
                \[
                  \sum_{a=1}^4\sum_{\lambda\in\Lambda_a}
                  \rank(I_2-A\lambda T_a)+\rank(A)\le10,
                \]
                \[
                  \sum_{a=1}^4\sum_{\lambda\in\Lambda_a}
                  \rank(\lambda T_a-D)+\rank(D)\le10.
                \]
                The verification finds no choice satisfying all these necessary conditions.
                
                This contradicts the existence of a \((10,8,2)\) MDS array code $\mathcal{C}$ over
                \(\F_q\), \(q=8,9\), with \(\beta(\mathcal{C})=10\). Hence
                \[
                  \beta_{\opt}^{(9,10,2,2)}=11,
                  \qquad
                  \beta_{\opt}^{(8,10,2,2)}=11.
                \]
                \end{proof}

                \begin{proposition}\label{prop:no-beta9-q8-n9}
                  One has
                  \[
                    \beta_{\opt}^{(8,9,2,2)}=10.
                  \]
                  \end{proposition}
                  
                  \begin{proof}
                  By Theorem~\ref{thm:ZLH-lb} and Proposition~\ref{prop:short-bw}, it suffices
                  to show that there is no \((9,7,2)\) MDS array code \(\mathcal C\) over
                  \(\F_8\) with $\beta(\mathcal C)=9$. Suppose, for contradiction, that such a code exists. Choose a parity-check
                  matrix
                  \[
                    H=[\,H_1\ H_2\ \cdots\ H_9\,],
                    \qquad H_i\in\F_8^{4\times2}.
                  \]
                  As in Subsection~\ref{subsec:repair-bd-lb}, we may assume without loss of
                  generality that
                  \[
                    H_j=
                    \begin{bmatrix}
                    I_2\\
                    W_j
                    \end{bmatrix}
                    \quad(j\in[7]),
                    \qquad
                    H_8=
                    \begin{bmatrix}
                    I_2\\
                    0
                    \end{bmatrix},
                    \qquad
                    H_9=
                    \begin{bmatrix}
                    0\\
                    I_2
                    \end{bmatrix},
                  \]
                  with \(W_j\in\GL_2(\F_8)\).
                  
                  Since
                  \[
                    \beta_i(\mathcal C)=n-2+|\mathcal R_i|
                  \]
                  and \(\beta(\mathcal C)=9\), we have
                  \[
                    |\mathcal R_i|\le 9-(9-2)=2
                    \qquad(i\in[9]).
                  \]
                  
                  We first extract three two-node groups among the first seven nodes. If
                  \(t<n-2\), then Lemma~\ref{lem:mod-1-case1-structure}, applied with
                  \(m=1\), gives a partition
                  \[
                    [7]=P_1\sqcup P_2\sqcup P_3\sqcup Q,
                    \qquad
                    |P_1|=|P_2|=|P_3|=2,\quad |Q|=1,
                  \]
                  such that
                  \[
                    \mathcal R_i=P_a
                    \qquad(i\in P_a,\ a=1,2,3).
                  \]
                  
                  It remains to consider the case \(t=n-2=7\). By
                  \cite[Lemma~12]{zhang2025optimal}, the distinct sets among
                  \(\mathcal R_1,\dots,\mathcal R_7\) are at most four. Let these distinct sets
                  be
                  \[
                    S_1,\dots,S_d
                    \qquad(d\le4),
                  \]
                  and put
                  \[
                    P_a:=\{\,i\in[7]:\mathcal R_i=S_a\,\},
                    \qquad
                    s_a:=|P_a|.
                  \]
                  Since \(i\in\mathcal R_i\), we have \(P_a\subseteq S_a\). Hence
                  \[
                    s_a\le |S_a|\le2.
                  \]
                  Since
                  \[
                    \sum_{a=1}^d s_a=7
                  \]
                  and \(d\le4\), exactly three of the \(s_a\)'s are equal to \(2\), and the
                  remaining one is equal to \(1\). For each class with \(s_a=2\), the inclusion
                  \(P_a\subseteq S_a\) and the bound \(|S_a|\le2\) imply
                  \[
                    P_a=S_a.
                  \]
                  Thus, after relabeling, in either case there exist disjoint two-element
                  subsets
                  \[
                    P_1,P_2,P_3\subseteq[7]
                  \]
                  such that
                  \[
                    \mathcal R_i=P_a
                    \qquad(i\in P_a,\ a=1,2,3).
                  \]
                  Let the remaining index in \([7]\setminus(P_1\cup P_2\cup P_3)\) be \(7\).
                  
                  Since \(H_7=[\,I_2\; W_7\,]^T\) with \(W_7\in\GL_2(\F_8)\), write
                  \[
                    H_7=
                    \begin{bmatrix}
                    I_2\\
                    T_4
                    \end{bmatrix}
                  \]
                  for some \(T_4\in\GL_2(\F_8)\). Applying the admissible row operation
                  \[
                    \begin{bmatrix}
                    I_2&0\\
                    0&T_4^{-1}
                    \end{bmatrix}
                  \]
                  to the whole parity-check matrix and then restoring \(H_9\) by an invertible
                  column operation inside its block, we may assume
                  \[
                    T_4=I_2.
                  \]
                  Thus
                  \[
                    H_7=
                    \begin{bmatrix}
                    I_2\\
                    I_2
                    \end{bmatrix},
                    \qquad
                    H_8=
                    \begin{bmatrix}
                    I_2\\
                    0
                    \end{bmatrix},
                    \qquad
                    H_9=
                    \begin{bmatrix}
                    0\\
                    I_2
                    \end{bmatrix}.
                  \]
                  
                  We now put the first six blocks into a normal form. For \(a=1,2,3\), choose
                  one index \(i_a\in P_a\) and a repair matrix \(M^{(a)}\) for the node \(i_a\).
                  By the definition of \(\mathcal R_{i_a}\) and the condition
                  \(\mathcal R_{i_a}=P_a\), the corresponding repair matrix has rank \(2\) on
                  the two nodes in \(P_a\), rank \(1\) on the nodes in \(P_b\) for \(b\neq a\),
                  and rank \(1\) on \(H_7\).
                  
                  By \cite[Lemma~9]{zhang2025optimal}, after an invertible row operation on
                  \(M^{(a)}\), we may write
                  \[
                    M^{(a)}
                    =
                    \begin{bmatrix}
                    \mathbf u_a&0\\
                    0&\mathbf v_a
                    \end{bmatrix},
                  \]
                  where \(\mathbf u_a,\mathbf v_a\in\F_8^{1\times2}\) are non-zero row vectors.
                  Since
                  \[
                    M^{(a)}H_7=
                    \begin{bmatrix}
                    \mathbf u_a\\
                    \mathbf v_a
                    \end{bmatrix}
                  \]
                  has rank \(1\), we have
                  \[
                    \langle\mathbf u_a\rangle=\langle\mathbf v_a\rangle.
                  \]
                  Denote this common projective class by \(L_a\).
                  
                  The three projective classes \(L_1,L_2,L_3\) are pairwise distinct. Indeed, if
                  \(L_a=L_b\) for \(a\neq b\), choose any node \(j\in P_b\). Since
                  \(j\in P_b\), the repair matrix \(M^{(b)}\) has rank \(2\) on \(H_j\), so
                  \[
                    L_b W_j\not\subseteq L_b.
                  \]
                  On the other hand, since \(j\notin P_a=\mathcal R_{i_a}\), the repair matrix
                  \(M^{(a)}\) has rank \(1\) on \(H_j\), and hence
                  \[
                    L_a W_j\subseteq L_a.
                  \]
                  This contradicts \(L_a=L_b\).
                  
                  Using the sharp \(3\)-transitivity of \(\PGL_2(\F_8)\) on
                  \(\mathbb P^1(\F_8)\), and conjugating all \(W_j\)'s accordingly, we may assume
                  that
                  \[
                    L_1=\langle(0,1)\rangle,\qquad
                    L_2=\langle(1,0)\rangle,\qquad
                    L_3=\langle(1,1)\rangle.
                  \]
                  The admissible equivalence is the same as in
                  Proposition~\ref{prop:no-beta10-q98-n10}: if \(X\in\GL_2(\F_8)\) realizes the
                  projective change of coordinates, then the row operation
                  \[
                    \begin{bmatrix}
                    X^{-1}&0\\
                    0&X^{-1}
                    \end{bmatrix}
                  \]
                  followed by right multiplication of each block by \(X\) replaces every
                  \(W_j\) by \(X^{-1}W_jX\), while preserving the standard forms of
                  \(H_7,H_8,H_9\), the MDS property, and all repair ranks.
                  
                  Now let \(W\) be a block belonging to \(P_1\). Since the repair matrices
                  corresponding to \(P_2\) and \(P_3\) have rank \(1\) on this block, \(W\) fixes
                  the projective classes \(L_2\) and \(L_3\). Since the repair matrix
                  corresponding to \(P_1\) has rank \(2\) on this block, \(W\) does not fix
                  \(L_1\). Therefore \(W\) has the form
                  \[
                    W=\lambda
                    \begin{bmatrix}
                    1&0\\
                    x&1+x
                    \end{bmatrix},
                    \qquad
                    \lambda\in\F_8^\times,\quad x\in\F_8\setminus\{0,1\}.
                  \]
                  Similarly, every block belonging to \(P_2\) has the form
                  \[
                    W=\lambda
                    \begin{bmatrix}
                    1&y\\
                    0&1+y
                    \end{bmatrix},
                    \qquad
                    \lambda\in\F_8^\times,\quad y\in\F_8\setminus\{0,1\},
                  \]
                  and every block belonging to \(P_3\) has the form
                  \[
                    W=\lambda
                    \begin{bmatrix}
                    1&0\\
                    0&z
                    \end{bmatrix},
                    \qquad
                    \lambda\in\F_8^\times,\quad z\in\F_8^\times\setminus\{1\}.
                  \]
                  
                  Thus any equality case over \(\F_8\) gives two matrices from each of the three
                  families
                  \[
                    \mathcal F_1=
                    \left\{
                    \lambda
                    \begin{bmatrix}
                    1&0\\
                    x&1+x
                    \end{bmatrix}
                    :
                    \lambda\in\F_8^\times,\ x\in\F_8\setminus\{0,1\}
                    \right\},
                  \]
                  \[
                    \mathcal F_2=
                    \left\{
                    \lambda
                    \begin{bmatrix}
                    1&y\\
                    0&1+y
                    \end{bmatrix}
                    :
                    \lambda\in\F_8^\times,\ y\in\F_8\setminus\{0,1\}
                    \right\},
                  \]
                  and
                  \[
                    \mathcal F_3=
                    \left\{
                    \lambda
                    \begin{bmatrix}
                    1&0\\
                    0&z
                    \end{bmatrix}
                    :
                    \lambda\in\F_8^\times,\ z\in\F_8^\times\setminus\{1\}
                    \right\}.
                  \]
                  Let $\mathcal W$ be the resulting set of six matrices, together with \(I_2\) and \(0\), which
                  correspond to \(H_7\) and \(H_8\). The MDS property requires
                  \[
                    \det(W-W')\neq0
                    \qquad
                    \text{for all distinct }W,W'\in\mathcal W.
                  \]
                  
                  Since \(\beta(\mathcal C)=9\), the three nodes \(H_7,H_8,H_9\) must all have
                  repair bandwidth at most \(9\). We express these conditions by rank sums. For \(H_7\), a repair matrix can be normalized so that \(MH_7=I_2\), and hence
                  has the form
                  \[
                    M=[\,B\ \ I_2-B\,]
                  \]
                  for some \(B\in\F_8^{2\times2}\). Thus there must exist
                  \(B\in\F_8^{2\times2}\) such that
                  \[
                    \sum_{W\in\mathcal W\setminus\{0,I_2\}}
                    \rank\bigl(B+(I_2-B)W\bigr)
                    +
                    \rank(B)+\rank(I_2-B)
                    \le9.
                  \]
                  For \(H_8\), a repair matrix can be normalized so that \(MH_8=I_2\), and hence
                  has the form
                  \[
                    M=[\,I_2\ \ {-A}\,]
                  \]
                  for some \(A\in\F_8^{2\times2}\). Thus there must exist
                  \(A\in\F_8^{2\times2}\) such that
                  \[
                    \sum_{W\in\mathcal W\setminus\{0,I_2\}}
                    \rank(I_2-AW)
                    +
                    \rank(I_2-A)+\rank(A)
                    \le9.
                  \]
                  For \(H_9\), a repair matrix can be normalized so that \(MH_9=I_2\), and hence
                  has the form
                  \[
                    M=[\,{-D}\ \ I_2\,]
                  \]
                  for some \(D\in\F_8^{2\times2}\). Thus there must exist
                  \(D\in\F_8^{2\times2}\) such that
                  \[
                    \sum_{W\in\mathcal W}\rank(W-D)\le9.
                  \]
                  
                  We now carry out a finite computer verification over \(\F_8\). The verification
                  enumerates all unordered pairs of matrices from each of the three families
                  \(\mathcal F_1,\mathcal F_2,\mathcal F_3\), checks the MDS condition
                  \[
                    \det(W-W')\neq0
                    \qquad(W\neq W',\ W,W'\in\mathcal W),
                  \]
                  and then checks whether there exist matrices
                  \[
                    A,B,D\in\F_8^{2\times2}
                  \]
                  satisfying the three rank-sum inequalities above. The verification finds no
                  choice satisfying all these necessary conditions.
                  
                  This contradicts the existence of a \((9,7,2)\) MDS array code $\mathcal{C}$ over \(\F_8\)
                  with repair bandwidth \(\beta(\mathcal{C})=9\). Therefore
                  \[
                    \beta_{\opt}^{(8,9,2,2)}=10.
                  \]
                  \end{proof}

                  \refstepcounter{subsection}\label{app:gamma-opt-exc}
                  \subsection*{\thesubsection\quad Optimal Repair I/O for Small-Parameter Exceptional Cases}
                  \addcontentsline{toc}{appsubsection}{\protect\numberline{\thesubsection}Optimal Repair I/O for Small-Parameter Exceptional Cases}
                  
                  We handle the code length \(n=4\), which is not covered by Corollary~\ref{cor:short-optimal-io}.
                  
                  \begin{proposition}\label{prop:gamma-opt-4-2-2}
                    For any prime power $q$, one has $\gamma_{\opt}^{(q,4,2,2)}=3$.
                  \end{proposition}
                  
                    \begin{proof}
                    By Theorem~\ref{thm:ZLH-lb}, it suffices to construct, for every prime power \(q\), a \((4,2,2)\) MDS array code \(\mathcal C\) over $\mathbb{F}_q$ with \(\gamma(\mathcal C)=3\). We distinguish two cases.
                  
                    \smallskip
                    \noindent\textbf{Case 1: \(q\) is odd.}
                    Let $\mathcal{C}$ be the \((4,2,2)\) array code with parity-check matrix \(H=[\,H_1\ H_2\ H_3\ H_4\,]\) given by
                    \[
                    H_1=
                    \begin{bmatrix}
                    1&1\\
                    0&1\\
                    0&0\\
                    0&0
                    \end{bmatrix},\qquad
                    H_2=
                    \begin{bmatrix}
                    0&0\\
                    0&0\\
                    1&-1\\
                    0&1
                    \end{bmatrix},\qquad H_3=
                    \begin{bmatrix}
                    1&-1\\
                    0&1\\
                    1&-1\\
                    0&1
                    \end{bmatrix},\qquad H_4=
                    \begin{bmatrix}
                    1&1\\
                    0&1\\
                    0&-1\\
                    1&1
                    \end{bmatrix}.
                    \]
                    A direct calculation shows that for any indices $i,j\in [4]$ with $i<j$, the \(4\times4\) matrix
                    \(
                    [H_i\ H_j]
                    \)
                    has determinant $1$ or \(2\). Since \(2\neq 0\) in odd characteristic, \(\mathcal C\) is a \((4,2,2)\) MDS array code over every field of odd characteristic.
                    
                    Now define the repair matrices
                    \[
                    M_1=
                    \begin{bmatrix}
                    1&0&-1&-1\\
                    0&1&0&0
                    \end{bmatrix},\qquad
                    M_2=
                    \begin{bmatrix}
                    0&1&1&0\\
                    0&-1&0&1
                    \end{bmatrix},
                    \]
                    \[
                    M_3=
                    \begin{bmatrix}
                    1&-1&0&0\\
                    0&0&1&1
                    \end{bmatrix},\qquad
                    M_4=
                    \begin{bmatrix}
                    0&1&0&0\\
                    0&0&0&1
                    \end{bmatrix}.
                    \]
                    One can check directly that
                    \[
                    M_1H_1,\ M_2H_2,\ M_3H_3,\ M_4H_4\in \GL_2(\F_q),
                    \]
                    while for any indices $i,j\in [4]$ with $i\ne j$, the matrix \(M_iH_j\) has exactly one non-zero column. Consequently, $\gamma(\mathcal C)=3$.
                    
                    \smallskip
                    \noindent\textbf{Case 2: \(q\) is even.}
                    Let \(H=[\,H_1\ H_2\ H_3\ H_4\,]\) be the parity-check matrix, where
                    \[
                    H_1=
                    \begin{bmatrix}
                    1&1\\
                    1&0\\
                    0&0\\
                    0&0
                    \end{bmatrix},\qquad
                    H_2=
                    \begin{bmatrix}
                    0&0\\
                    0&0\\
                    1&1\\
                    0&1
                    \end{bmatrix},\qquad H_3=
                    \begin{bmatrix}
                    0&1\\
                    1&1\\
                    0&1\\
                    1&1
                    \end{bmatrix},\qquad H_4=
                    \begin{bmatrix}
                    1&1\\
                    1&0\\
                    0&1\\
                    1&1
                    \end{bmatrix}.
                    \]
                    A direct calculation shows that for any indices $i,j\in [4]$ with $i<j$, the matrix
                    \(
                    [H_i\ H_j]
                    \)
                    has determinant \(1\). Hence \(\mathcal C\) is a \((4,2,2)\) MDS array code over every field of characteristic \(2\).
                    
                    Now define the repair matrices
                    \[
                    M_1=
                    \begin{bmatrix}
                    0&1&0&1\\
                    1&0&0&1
                    \end{bmatrix},\qquad
                    M_2=
                    \begin{bmatrix}
                    0&0&1&0\\
                    0&1&0&1
                    \end{bmatrix},
                    \]
                    \[
                    M_3=
                    \begin{bmatrix}
                    0&0&1&1\\
                    0&1&0&0
                    \end{bmatrix},\qquad
                    M_4=
                    \begin{bmatrix}
                    0&0&1&1\\
                    1&1&0&0
                    \end{bmatrix}.
                    \]
                    Again, one can check directly that
                    \[
                    M_1H_1,\ M_2H_2,\ M_3H_3,\ M_4H_4\in \GL_2(\F_q),
                    \]
                    while for any indices $i,j\in [4]$ with $i\ne j$, the matrix \(M_iH_j\) has exactly one non-zero column. Consequently, $\gamma(\mathcal C)=3$. This completes the proof.
                    \end{proof}

\clearpage
\phantomsection
\addcontentsline{toc}{section}{\refname}
\bibliographystyle{alphaurl}
\bibliography{references}


\end{document}